\documentclass[a4paper,twocolumn,notitlepage,nofootinbib,superscriptaddress]{revtex4-2}
\usepackage{amsmath}
\usepackage{amssymb}
\usepackage{amsthm}
\usepackage{graphicx}
\usepackage{xcolor}
\usepackage[english]{babel}
\usepackage{csquotes}
\usepackage{braket}
\usepackage[bookmarks=true,colorlinks=true,linkcolor=orange,citecolor=orange,urlcolor=orange,bookmarks]{hyperref}
\usepackage{mathtools}
\usepackage{booktabs}
\usepackage{thmtools}
\usepackage{enumitem}
\usepackage{algorithm}
\usepackage[normalem]{ulem}

\usepackage{algpseudocode}
\usepackage{verbatim}
\usepackage[capitalize]{cleveref}

\newcommand{\diff}{\ensuremath\mathrm{d}}

\newcommand{\Id}{\ensuremath{\mathbb{I}}}

\newcommand*{\Tr}{\operatorname{Tr}}

\newcommand*{\argmin}{\operatornamewithlimits{argmin}}

\newcommand{\norm}[1]{\left\lVert#1\right\rVert}

\providecommand{\myvec}[1]{\ensuremath{\boldsymbol{#1}}}

\providecommand{\XX}{\ensuremath{\myvec{X}}}

\providecommand{\aa}{\ensuremath{\myvec{a}}}

\providecommand{\cc}{\ensuremath{\myvec{c}}}

\providecommand{\ff}{\ensuremath{\myvec{f}}}

\providecommand{\kk}{\ensuremath{\myvec{k}}}
\providecommand{\ll}{\ensuremath{\myvec{l}}}

\providecommand{\nn}{\ensuremath{\myvec{n}}}

\providecommand{\ss}{\ensuremath{\myvec{s}}}

\providecommand{\vv}{\ensuremath{\myvec{v}}}

\providecommand{\xx}{\ensuremath{\myvec{x}}}

\providecommand{\zz}{\ensuremath{\myvec{z}}}

\providecommand{\aalpha}{\ensuremath{\myvec{\alpha}}}

\providecommand{\eepsilon}{\ensuremath{\myvec{\epsilon}}}

\providecommand{\eeta}{\ensuremath{\myvec{\eta}}}
\providecommand{\ttheta}{\ensuremath{\myvec{\theta}}}

\providecommand{\rrho}{\ensuremath{\myvec{\rho}}}

\providecommand{\ssigma}{\ensuremath{\myvec{\sigma}}}

\providecommand{\calA}{\ensuremath{\mathcal{A}}}

\providecommand{\calM}{\ensuremath{\mathcal{M}}}
\providecommand{\calN}{\ensuremath{\mathcal{N}}}
\providecommand{\calO}{\ensuremath{\mathcal{O}}}

\providecommand{\calT}{\ensuremath{\mathcal{T}}}
\providecommand{\calU}{\ensuremath{\mathcal{U}}}

\providecommand{\calX}{\ensuremath{\mathcal{X}}}

\providecommand{\bbC}{\ensuremath{\mathbb{C}}}

\providecommand{\bbE}{\ensuremath{\mathbb{E}}}

\providecommand{\bbI}{\ensuremath{\mathbb{I}}}

\providecommand{\bbN}{\ensuremath{\mathbb{N}}}

\providecommand{\bbP}{\ensuremath{\mathbb{P}}}

\providecommand{\bbR}{\ensuremath{\mathbb{R}}}

\providecommand{\bbZ}{\ensuremath{\mathbb{Z}}}

\newcommand{\tvert}[1]{{\left\vert\kern-0.25ex\left\vert\kern-0.25ex\left\vert #1 
    \right\vert\kern-0.25ex\right\vert\kern-0.25ex\right\vert}}

\AtBeginDocument{
\heavyrulewidth=.08em
\lightrulewidth=.05em
\cmidrulewidth=.03em
\belowrulesep=.65ex
\belowbottomsep=0pt
\aboverulesep=.4ex
\abovetopsep=0pt
\cmidrulesep=\doublerulesep
\cmidrulekern=.5em
\defaultaddspace=.5em
}

\declaretheoremstyle[
  headfont=\color{red}\normalfont\bfseries,
  bodyfont=\color{red}\normalfont\itshape,
]{colored}

\renewcommand{\lll}{\ensuremath{\boldsymbol{l}}}

\newtheorem{theorem}{Theorem}
\newtheorem{lemma}[theorem]{Lemma}

\newtheorem{definition}[theorem]{Definition}
\newtheorem{corollary}[theorem]{Corollary}

\newtheorem{assumption}{Assumption}

\newtheorem{stheorem}{Theorem}
\newtheorem{slemma}[stheorem]{Lemma}

\providecommand{\customgenericname}{}
\newtheorem{innercustomgeneric}{\customgenericname}
\newcommand{\newcustomtheorem}[2]{%
  \newenvironment{#1}[1]
  {%
   \renewcommand\customgenericname{#2}%
   \renewcommand\theinnercustomgeneric{##1}%
   \innercustomgeneric
  }
  {\endinnercustomgeneric}
}
\newcustomtheorem{customthm}{Theorem}

\usepackage{apptools}
\AtAppendix{\counterwithin{stheorem}{section}}
\AtAppendix{\counterwithin{slemma}{section}}
\AtAppendix{\counterwithin{sproposition}{section}}
\AtAppendix{\counterwithin{sdefinition}{section}}
\AtAppendix{\counterwithin{scorollary}{section}}
\AtAppendix{\counterwithin{sconjecture}{section}}
\AtAppendix{\counterwithin{sproblem}{section}}
\AtAppendix{\counterwithin{sassumption}{section}}

\renewcommand{\aa}{\boldsymbol{a}}

\newcommand{\dccqs}{Dahlem Center for Complex Quantum Systems, Freie Universit{\"a}t Berlin, 14195 Berlin, Germany}
\newcommand{\hzb}{Helmholtz-Zentrum Berlin f{\"u}r Materialien und Energie, 14109 Berlin, Germany}
\newcommand{\hhi}{Fraunhofer Heinrich Hertz Institute, 10587 Berlin, Germany}

\newcommand{\papertitle}{Tomography of parametrized quantum states}

\begin{document}	

\title{\papertitle}
\date{\today}

\author{Franz J.\ Schreiber}
\affiliation{\dccqs}

\author{Jens Eisert}
\affiliation{\dccqs}
\affiliation{\hzb}
\affiliation{\hhi}

\author{Johannes Jakob Meyer}
\affiliation{\dccqs}

\begin{abstract}
Characterizing quantum systems is a fundamental task that enables the development of quantum technologies. 
Various approaches, ranging from full tomography to instances of classical shadows, have been proposed to this end. 
However, quantum states that are being prepared in practice often involve families of quantum states characterized by continuous parameters, such as the time evolution of a quantum state.
In this work, we extend the foundations of quantum state tomography to 
parametrized quantum states.
We introduce a framework that unifies different notions of tomography and use it to establish a natural figure of merit for tomography of parametrized quantum states.
Building on this, we provide an explicit algorithm that combines signal processing techniques with a tomography scheme to recover an approximation to the parametrized quantum state equipped with explicit guarantees.
Our algorithm uses techniques from compressed sensing to exploit structure in the parameter dependence and operates with a \enquote{plug and play} nature, using the underlying tomography scheme as a black box. 
In an analogous fashion, we derive a figure of merit that applies to parametrized quantum channels. 
Substituting the state tomography scheme with a scheme for process tomography in our algorithm, we then obtain a protocol for tomography of parametrized quantum channels. 
We showcase our algorithm with two examples of shadow tomography of states time-evolved under an NMR Hamiltonian and a free fermionic Hamiltonian.
\end{abstract}

\maketitle

Characterizing quantum states is a central task in quantum information science. 
The importance of methods of tomographic recovery, benchmarking and verification~\cite{BenchmarkingReview,Certification} is further elevated by
the advent of \emph{NISQ (noisy intermediate-scale quantum)} devices, where, in particular, quantum state tomography is widely seen as a key method to certify that a device is working correctly, 
and to provide actionable advice on how to improve the experiment in case of a deviation. 
The practical utility of full state tomography beyond the study of very small quantum systems is, however, severely limited by its stark resource requirements in the number of copies of the state, which scales exponentially in the number of involved quantum systems, such as qubits or modes. The same holds true for the classical memory necessary to store the information gathered in the process. 

To overcome this issue, novel methods for \emph{approximate} state tomography have been developed, with \emph{shadow tomography} being a particularly prominent example. In these approaches, a better performance is obtained by relaxing the requirements involved: instead of characterizing the whole quantum state, only some of its properties should be learned.
Owing to these developments, it is possible to construct classical approximations of quantum states that allow for high precision predictions for large classes of observables, while maintaining resource requirements which scale polynomially with respect to the number of involved systems. 

Still, in many scenarios of practical relevance, one is not only interested in a quantum state, but in a whole family of quantum states, depending on possibly multiple continuous parameters -- a \emph{parametrized quantum state}. The archetypal example here is quantum time evolution, where the quantum states are parametrized by the time parameter. Other examples are certification of analog quantum devices, where control parameters naturally give rise to such a setting, or the electronic ground states of a molecule, depending on the coordinates of the heavy nuclei. 

\begin{figure}
    \centering
    \includegraphics{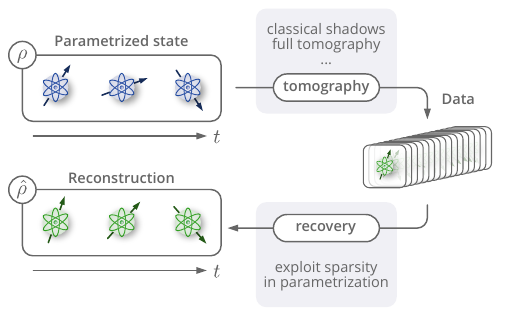}
    \caption{Schematic sketch of our tomography procedure. A parametrized quantum state $\rho(\cdot)$, exemplified by a 
    family of time-dependent states, 
    is to be reconstructed across the whole range of the parameter. We do so by running a tomography procedure as a black box at different values of the parameter. The obtained data is then processed into an estimate $\hat\rho(\cdot)$ of the parametrized state. In this part, a sparsity structure of the parametrization can be exploited through compressed sensing. The estimate has guarantees that are inherited from the tomography procedure.}
    \label{fig:front_figure_parametrized_tomography}
\end{figure}

In this work, we extend the foundations of quantum state tomography to parametrized quantum states, lifting procedures for state tomography to procedures for the tomography of parametrized quantum states. 
Our first key contribution to this end is to introduce a framework that unifies different 
notions of tomography, 
like full tomography and shadow tomography. 
We achieve this by expressing 
the figure of merit of these tomographic notions 
as an optimization over restricted sets of observables. 
For parametrized quantum states observables no longer give rise to scalar expectation values, but to functions taking the parameters as input. By combining the unified framework for state tomography with the associated 
$L^p$ function norms, 
we obtain a natural figure of merit for the tomography of parametrized quantum states that is expressed as the largest $L^p$ norm difference for an observable in the restricted set.

Our next key contribution is to present a 
algorithm that recovers a full parametrized quantum state with respect to the natural figure of merit we introduce. This is realized by using a given tomography scheme as a black box at a number of parameter values and
classically
combining the obtained data through signal processing techniques, as shown in \cref{fig:front_figure_parametrized_tomography}. 
Our algorithm allows us to explicitly use techniques of compressed sensing to exploit \emph{sparsity} in the parametrization. This happens, for example, when only few frequencies contribute to the time-evolution of a quantum state. 

Our algorithm is efficient if the used tomography scheme is efficient and the parameter dependence can be 
accurately  
captured by a polynomial number of basis functions. Furthermore, the representation obtained for the parametrized quantum state inherits its data structure from the tomography scheme. Many classical representations like density matrices or tensor networks allow for more than just estimating observables, but are well suited to compute other properties of the underlying system. 

We emphasize that our algorithm constitutes significant method development, as the facts that computing norms of quantum states can be computationally hard and that quantum states only allow for access via measurements prohibit a straightforward application of compressed sensing techniques to recover parametrized quantum states. 
Our algorithm can be applied to both finite- and infinite-dimensional systems. Like classical shadows, it also gives an \emph{agnostic} representation of the parametrized quantum state that can be used to infer the parameter-dependent expectation value of arbitrary observables as long as they are in the set for which the tomography procedure has guarantees. 

This new tomographic paradigm can be naturally extended to parametrized quantum channels. Analogously to the state case, we take a unifying figure of merit for process tomography and combine it with function $L^p$ norms to obtain a figure of merit for tomography of parametrized quantum channels. Substituting as input the tomographic procedure with a scheme for process tomography, our algorithm outputs a classical representation of the parametrized channel. The efficiency of the resulting scheme depends on an efficient description of the parametrization and on the efficiency of the scheme for process tomography given as input.

There have already been instances where tomographic procedures, in particular classical shadows, have been applied to questions involving parametrized quantum states. The authors of Ref.~\cite{chan2023algorithmic} introduce an algorithmic technique coined \emph{shadow spectroscopy}.
It allows to estimate the energy gaps of a Hamiltonian from classical shadows taken at different short times of an evolution under the Hamiltonian. In Ref.~\cite{zhan2024learning}, the authors learn conservation laws that can be expressed as the sum of few geometrically local Pauli observables from classical shadows constructed at different times of the quantum time evolution. In Ref.~\cite{shivam2023classical}, the authors investigate how well the time evolution of a quantum state can be predicted from its classical shadow. Beyond time evolution, there has been work to predict ground states of different phases of matter~\cite{huang2021provably, lewis2024improved}. There, the ground states of a family of Hamiltonians form a parametrized quantum state, where the task is to obtain classical representations for new parameter values from classical shadows constructed at different points of the parametrization. 
We briefly comment on the relation between this and prior work. Our work is fundamentally different in scope in that we construct \emph{classical representations} of parametrized quantum states, generalizing state tomography to the parameter-dependent setting. However, in doing so, we provide a unified approach as many of the above mentioned techniques can be re-expressed as a tomographic procedure for parametrized quantum states combined with classical post-processing. This is reminiscent of how quantum state tomography appears as a subroutine in various different contexts. We also examine how to exploit structure in the form of sparsity in the parametrization to obtain more efficient algorithms.   

Our work is structured as follows: In \cref{sec:parametrized quantum states}, we rigorously define what we mean by a parametrized quantum state and introduce some conventions regarding notation that we adopt throughout this work. In \cref{sec:a general framework for tomographic procedures}, we define the notion of a tomographic procedure, which encompasses full state tomography as well as approximation techniques like shadow tomography and classical shadows. We especially emphasize how different tomographic procedures 
use different ways of measuring distances between quantum states. This sets the stage for \cref{sec:notions of approximation for tomographic procedures of parametrized quantum states}, where we investigate how such distance measures generalize from quantum states to parametrized quantum states. Then, in \cref{sec:short introduction to compressed sensing}, we provide a crash course in the topic of compressed sensing, leading up to \cref{sec:tomography of parametrized quantum states}, where we give an algorithm detailing how to construct tomographic procedures on parametrized quantum states as well as provide a corresponding error analysis. Special attention is given to classical shadows. In \cref{sec:identification of a sparse support}, we explain how to identify sparse supports in function bases for parametrized quantum states. 
In \cref{sec:application}, we discuss at the hand of two example how our framework enables efficient tomography of two types of parametrized quantum states: First, we show the efficient tomography of states with a sub-Gaussian energy spectrum evolving under an NMR Hamiltonian in the Fourier basis. Second, we efficiently reconstruct a quantum state under free fermionic time evolution through Chebychev polynomials. 
In \cref{sec:channel_tomography}, we show that the ideas and techniques introduced so far can be naturally extended to parametrized quantum channels. 
For the convenience of the reader, we give a high-level overview of our work and techniques in \cref{sec:practicioners_guide}.
An outlook on future directions is given in \cref{sec:future_directions} and we conclude the manuscript in \cref{sec:conclusion}.

\section{Parametrized quantum states} \label{sec:parametrized quantum states}
On a practical level, parametrized quantum states arise in many contexts, for example in the time evolution $\rho(t) = e^{-iHt} \rho(0) e^{iHt}$
of a quantum state under a given Hamiltonian $H$, the encoding of a phase on a probe state in a quantum sensor $\rho(\phi) = \calN(\phi)[\rho(0)]$ or when parametrized quantum circuits are used 
in the context of quantum machine learning $\rho(\ttheta) = \calU(\ttheta)[\rho(0)]$. 

On an abstract level, a parametrized quantum state $\rho(\cdot)$ is a function from the parameter space $\calX$ to the space of linear operators on a Hilbert space $\calM_n$ associated to $n$ physical systems that can be finite- or infinite-dimensional. We write $\rho\colon \calX \to \calM_n$ and require that for all $\xx \in \calX$, $\rho(\xx)$ is a valid quantum state. 

For our purpose, it will be especially important to consider expansions of such operator-valued functions in terms of orthonormal functions from $\calX$ to $\bbC$. We thus introduce some notation, for functions $f,g\colon \calX \to \bbC$, we denote with
\begin{align}
    \braket{f, g} \coloneqq \int_{\calX}\mathrm{d}\mu(\xx) \, f^*(\xx) g(\xx) 
\end{align}
the scalar product of the corresponding function space~\cite{rudin1987analysis}
equipped with probability measure $\mu$ over $\calX$. This scalar product induces the $L^2$ function norm as $\norm{f}_{2} = \sqrt{\braket{f,f}}$. 
For general $p$, this generalizes to
\begin{align}\label{eqn:sec_pqs_lp_norm}
    \norm{f}_p \coloneqq \left( \int_{\calX} \mathrm{d}\mu(\xx) \, |f(\xx)|^p \right)^{\frac{1}{p}}  \, ,
\end{align}
called 
the 
$L^p$-function 
norm, 
where 
\begin{align}
    \norm{f}_{\infty} \coloneqq \inf \left\{\text{$C:$ $|f(\xx)| \leq C$ $\mu$-almost everywhere} \right\} 
\end{align}
is the special case $p=\infty$. Later, these will important for the definition of distance measures between two parametrized quantum states $\rho, \sigma\colon\calX \to \calM_n$. Note that only the $L^2$-norm is induced by a scalar product.

An \emph{orthonormal function basis} (ONB) is a set of functions $\{\varphi_k\colon\calX \to \bbC\}_{k}$, such that
\begin{align}
    \braket{\varphi_k, \varphi_j} = \delta_{k,j}
\end{align}
which allows the representation of arbitrary functions $f\colon\calX \to \bbC$ as
\begin{align}
    f(\xx) = \sum_{k\in \Lambda} c_k \varphi_k(\xx)
\end{align}
with and $c_k = \braket{f, \varphi_k} \in \bbC$. We call the ONB bounded by the constant $K$ if 
\begin{align}
    \lVert \varphi_k \rVert_{\infty} \leq K \text{ for all } k \in \Lambda \, .
\end{align}
Every entry of the parametrized quantum state $\rho$ can be seen as a function from $\calX$ to $\bbC$. Therefore, for any ONB, we can expand the operator-valued function $\rho(\xx)$ as
\begin{align} \label{eqn:sec_pqs_expanding_rho}
    \rho(\xx) = \sum_{k \in \Lambda} \alpha_k \varphi_k(\xx)
\end{align}
with operator-valued coefficients $\alpha_k \in \calM_n$. The coefficients are given explicitly as
\begin{align} \label{eqn:sec_pqs_coefficients}
    \alpha_k \coloneqq& \langle \varphi_k, \rho\rangle = \int \diff \mu(\xx) \, \varphi_k^{*}(\xx)\rho(\xx).
\end{align}
While, in principle, infinitely many basis functions might be necessary
to perfectly describe $\rho(\xx)$, we assume that we are able to pick some finite set $\Lambda$ with $|\Lambda|=D$, such that the error from truncating the infinite series is negligible.

As an example, consider a Hamiltonian $H$ with integer $\lambda_i \in \bbN$ and an energy bound $0 \leq \lambda_i \leq E_{\max}$. We have $\calX = [0, 2\pi)$ and
\begin{align} \label{eqn:sec_pqs_time_evo}
    \rho(t) = e^{-iHt} \rho(0) e^{iHt} \, .
\end{align}
Any orthonormal basis of functions from $[0, 2\pi) \to \bbC$ would allow for an expansion in the style of \cref{eqn:sec_pqs_expanding_rho}, but in this case it is natural to choose the Fourier basis $\varphi_k(t) = \exp(-ikt)$ and $\mu(t)=1/(2\pi)$. We can then express any time-evolution of the form in \cref{eqn:sec_pqs_time_evo} as 
\begin{align}
    \rho(t) = \sum_{k \in \Lambda} \alpha_k e^{-ikt}  \, ,
\end{align}
where the set of frequencies $\Lambda$ is determined by the maximal eigenvalues bound of the Hamiltonian,
\begin{align}
    \Lambda = \left\{ -E_{\max}, -E_{\max}+1, \dots, 0, \dots, E_{\max} \right\} \, ,
\end{align}
and again, $\alpha_k$ are operator-valued.
While the size of $\Lambda$ can be very large, situations arise where the structure of the initial state of the time evolution are such that only a comparatively small set $S \subset \Lambda$ contribute significantly.

In our work, we are especially interested in cases like the above example, where we can pick a (possibly very large) set $\Lambda$ with an unknown subset $S \subset \Lambda$ with $|S| \ll |\Lambda|$ that also describes $\rho(\cdot)$ exactly or at least well. In such cases, we refer to $\rho(\cdot)$ as \emph{sparse} or \emph{approximately sparse}, respectively. In \cref{sec:a general framework for tomographic procedures,sec:notions of approximation for tomographic procedures of parametrized quantum states} we will develop the tools to provide rigorous definitions of both cases, while in \cref{sec:identification of a sparse support}, we provide an algorithm for the identification of such a set $S$.

\section{A general framework for state tomography} \label{sec:a general framework for tomographic procedures}
Before we can establish a scheme for the tomography of parametrized quantum states, we first need to establish an understanding of tomography of non-parametrized states.
In the typical form of state tomography, we are given $T$ copies of the system in state $\rho$ that we assume to be composed of $n$ quantum systems. We are tasked with outputting an approximation $\hat\rho$ that is close in trace distance with high probability. 
Formally, the output $\hat\rho$ of the tomography procedure needs to fulfill
\begin{align}\label{eqn:full_tomography_guarantee}
    \bbP\big[ \lVert \rho - \hat\rho \rVert_1 \leq \epsilon \big] \geq 1 - \delta,
\end{align}
where $\lVert.\rVert_1$ denotes the trace-norm.
We know that solving this task of \emph{full state tomography} necessitates at least~\cite{ODonnel2016efficient}
\begin{align}
    T = \Omega \left( \frac{2^{2n} + \log\frac{1}{\delta} }{\epsilon^2}  \right)
\end{align}
many samples, even if entangling measurements are being allowed, a simple consequence of the fact that quantum states have exponentially many degrees of freedom in the number of qubits. Under the additional assumption that the state $\rho$ has rank at most $r$ (low-rank tomography), we can reduce the factor $2^{2n}$ to $r 2^n$. While this is an exponential gain in sample complexity, it is still scaling exponentially in the number of qubits.

But for many practical applications, performing full tomography of the state $\rho$ is not necessary. A particularly prominent example of this is \emph{shadow tomography}~\cite{aaronson2018shadow,huang2020predicting}, where we are satisfied if out estimate of the state faithfully reproduces the expectation values of a given set of observables $\calO = \{ O_1, O_2, \dots, O_M \}$
\begin{align}\label{eqn:shadow_tomography_guarantee}
    \bbP\big[ {\textstyle \max_{1 \leq i \leq M}} |\Tr[ O_i \rho] - \Tr[ O_i \hat\rho] | \leq \epsilon \big] \geq 1 - \delta.
\end{align}
Assuming without loss of generality that we have $\Tr[O_i] = 0$ for all $O_i \in \calO$, the classical shadows protocol of Ref.~\cite{huang2020predicting} achieves the above using
\begin{align}
    T = O \left( \frac{{\max_{1 \leq i \leq M}}\lVert O_i \rVert_{\mathrm{shadow}}^2}{\epsilon^2} \log\frac{M}{\delta}\right) 
\end{align}
many samples.
The \emph{shadow norm} $\lVert \cdot \rVert_{\mathrm{shadow}}$ essentially captures the compatibility of the observable with the particular protocol used to construct the classical shadow. In the practically relevant case where all observables $O_i$ act only on $\ell$ subsystems, we can use local Clifford shadows, 
to get \cite[Proposition S3]{huang2020predicting} 
\begin{align}
    \lVert O_i \rVert_{\mathrm{shadow}}^2 \leq 4^{\ell} \lVert O_i \rVert_{\infty}^2 
\end{align}
which is an efficient scaling for constant $\ell$.

Instead of demanding that the difference of estimates over a set of observables is small in the worst case we can also demand that it should be small in expectation relative to a distribution over the observables or that the probability of an observable drawn from said distribution being large is small. This average case notion of tomography also allows for much more efficient protocols~\cite{aaronson2007learnability}.

\subsection*{Tomographic procedures}

Having reviewed two seemingly very different approaches to quantum state tomography, we will no go on to establish a unifying viewpoint that allows us to actually treat both full tomography and shadow tomography on the same footing, and that also extends to other ways of performing tomography.

The objective in shadow tomography (see \cref{eqn:shadow_tomography_guarantee}) is expressed as an optimization over different observables, whereas the objective in full tomography (see \cref{eqn:full_tomography_guarantee}) is written in terms of the trace norm. However, we can also express the trace norm as an optimization 
\begin{align}
    \lVert \rho - \hat\rho \rVert_1 = \sup_{\lVert O \rVert_{\infty} \leq 1} \Tr[ O (\rho - \hat\rho ) ]
\end{align}
over observables.
We use this \enquote{dual} viewpoint where a norm is expressed as an optimization to establish a unifying treatment of tomographic procedures. To this end, we define the \emph{induced semi-norm}
of a set of observables $\calO$. 
\begin{definition}[Induced semi-norm]\label{def:induced_semi-norm}
Let $\calO$ be a set of observables. The semi-norm induced by $\calO$ is defined as
\begin{align}
 \lVert X \rVert_{\calO} \coloneqq \sup_{O \in \calO} |\Tr[ X O ]|.
\end{align}
\end{definition}
This is indeed a semi-norm: it is non-negative by construction, homogeneity follows from the homogeneity of trace and absolute value and the triangle inequality can be easily established from the triangle inequality of the absolute value.
If $\calO$ is a compact convex set, then the supremum is achieved on the boundary $\partial \calO$. If $\calO$ is also closed under negation, \emph{i.e.},  $-O \in \calO$ if $O \in \calO$, then the absolute value is not necessary.

To achieve an average-case notion of approximation, one could in principle replace the supremum in the definition of the induced semi-norm with an expectation value relative to some distribution over observables and would still retain the semi-norm property.

The induced semi-norm
\begin{align} \label{eqn:sec_AGF_semi-norm_difference}
    \lVert \rho - \sigma \rVert_{\calO}=\sup_{O \in \calO} |\Tr[ \rho O ] - \Tr[\sigma O ]|
\end{align}
quantifies the largest deviation between the predictions of two quantum states for observables in the set $\calO$.
A semi-norm differs from a norm only in the fact that $\lVert X \rVert_{\calO} = 0$ does not imply that $X = 0$, which in the example of shadow tomography reflects that two distinct quantum states can give the exact same expectation values for a given set of observables. 

With the definition of the induced semi-norm as a figure of merit, we are ready to state the definition of a \emph{tomographic procedure} which will form the backbone of our generalization to parametrized quantum states.
\begin{definition}[Tomographic procedure]\label{def:tomographic_procedure}
An $(\epsilon, \delta, n)$ \emph{tomographic procedure} relative to a set of observables $\calO$ is an experimental procedure that uses $T(\epsilon, \delta, n)$ copies of a quantum state $\rho$ on $n$ copies of a physical system to produce a classical representation $\hat{\rho}$ such that
\begin{align}
    \bbP\big[ \lVert \rho - \hat\rho \rVert_{\calO} \leq \epsilon\big] \geq 1 - \delta.
\end{align}
The number of copies $T(\epsilon, \delta, n)$ is the \emph{sample complexity} of the tomographic procedure.
\end{definition}
The fact that $\lVert \cdot \rVert_{\calO}$ (typically) is a semi-norm instead of a norm reflects that we relaxed the requirements of full state tomography and means that even in the limit $\epsilon \to 0$ the tomographic procedure cannot necessarily differentiate between distinct states.  

In the light of the above definition, we can now understand full state tomography as a tomographic procedure relative to the 
set of observables $\calO = \{ O : \lVert O \rVert_{\infty} \leq 1\}$. If we, for example, restrict the observables in this definition to only act non-trivially on $\ell$ of the $n$ subsystems, the induced semi-norm is the $\ell$-local trace norm
\begin{align}
    \lVert X \rVert_{1,\ell} \coloneqq \max_{I \subseteq [n], |I| = \ell} \lVert X_{I} \rVert_1,
\end{align}
where we use the notation $X_I$ to denote the reduced operator on the subsystem $I$, \emph{i.e.},  what is left after tracing out the complement $\bar{I}$ of $I$, $X_I = \Tr_{\bar{I}}[X]$. 
The local Clifford shadows of Ref.~\cite{huang2020predicting} can be used to construct an $(\epsilon, \delta, n)$ tomographic procedure for the $\ell$-local trace norm with a sample complexity of
\begin{align}
    T(\epsilon, \delta, n) = O\left( \frac{\ell 12^\ell}{\epsilon^2}\log \frac{n}{\delta}\right).
\end{align}
Another norm captured by the definition of the induced semi-norm is the quantum Wasserstein distance of order 1~\cite{de_palma2021quantum}.

\section{Distance measures for parametrized quantum states} \label{sec:notions of approximation for tomographic procedures of parametrized quantum states}

We have introduced a general framework that captures the quality of a tomographic procedure through the semi-norm induced by a set of observables $\calO$. In this section, we give a natural way of lifting this definition to parametrized states to obtain a measure of distance between functions $\rho, \sigma\colon \calX \to \calM_n$ that can act as a figure of merit for tomography. 

We do so by thinking along similar lines as in the preceding section, where we computed the largest expectation value relative to a set of observables. 
Here, we now lift this reasoning to the parametrized case by recognizing that after an observable is fixed, the expectation value relative to this observable becomes itself a function from $\calX$ to $\bbR$ of which we can compute the $L^p$ function norm. Optimizing now again over all observables in the set $\calO$ leads us to the definition of the induced semi-norm for parametrized operators.
\begin{definition}[Induced $L^p$ semi-norm for parametrized operators]\label{def:sec_NoA_semi-norm_parametrized_operators} 
    Let $X\colon\calX \to \calM_n$ be a parametrized operator and $\calO$ a set of observables. Then the $L^p$ semi-norm induced by $\calO$ 
    is defined as
    \begin{align}
        \tvert{X(\cdot)}_{\calO, p} &\coloneqq \sup_{O \in \calO} \norm{\Tr[OX(\cdot)]}_p.
    \end{align}
    For $p < \infty$, it is 
    \begin{align}
        \tvert{X(\cdot)}_{\calO,p} = \sup_{O \in \calO} \left(\int_{\calX} \mathrm{d}\mu(\xx) \,|\Tr[OX(\xx)]|^p \right)^{\mathrlap{\frac{1}{p}}}.
    \end{align}
\end{definition}
Note that \cref{def:sec_NoA_semi-norm_parametrized_operators} does not require $X(\cdot)$ to be a quantum state. Analogous to \cref{eqn:sec_AGF_semi-norm_difference}, if the above semi-norm is applied to a pair of quantum states, it quantifies the largest deviation between predictions for the set of observables $\calO$, only here the predictions are functions $\Tr[O\rho(\cdot)]\colon \calX \to \bbR$ whose difference is accordingly measured with an $L^p$ function norm as
\begin{align}
    \tvert{\rho(\cdot) - \sigma(\cdot)}_{\calO, p} = \sup_{O \in \calO} \norm{\Tr[O\rho(\cdot)] - \Tr[O \sigma(\cdot)]}_p .
\end{align}

For a scalar function $f(\xx) = \sum_{k \in \Lambda} c_k \varphi_k(\xx)$, the coefficients of the basis expansion $\{ c_k \}_{k \in \Lambda}$ carry all the information about the function and the $\ell^q$ norms of the vector of coefficients $\cc = (c_k)_{k \in \Lambda}$ can be related to the $L^p$ norms of $f$ in certain cases, most strikingly in the case $p=q=2$ through Parseval's Theorem and in more general cases through the Hausdorff-Young inequalities. We would like to use similar tools to relate the induced $L^p$ semi-norm of a parametrized operator $X(\xx) = \sum_{k\in \Lambda} \alpha_k \varphi_k(\xx)$ to some norm of the vector of coefficients $\aalpha = (\alpha_k)_{k\in\Lambda}$. 

To define a compatible norm of vectors of operators, we will again make use of the same intuition -- namely that after fixing an observable, a vector of operators turns into a simple vector of complex numbers. This leads us to the definition of the induced $\ell^p$ semi-norm for vectors of operators.
\begin{definition}[Induced $\ell^p$ semi-norm for vectors of operators] \label{def:sec_NoA_semi-norm_vectors_of_operators}
    Let $\XX = (X_1, X_2,  \dots , X_m)$ be a vector of operators $X_i \in \calM_n$. 
    Then the $\ell^p$ semi-norm of $\XX$ induced by $\calO$ is defined as 
    \begin{align}
    \tvert{\XX}_{\calO, p} \coloneqq \sup_{O \in \calO} \left( \sum_{i=1}^m |\Tr[OX_i]|^p \right)^{\frac{1}{p}} \,.
    \end{align}
\end{definition}
Again, we quantify the largest deviation between predictions for the set of observables $\calO$, but here we have a vector of predictions, which we treat with the usual $p$-norm:
Let $X_O \coloneqq (\Tr[O X_1], \Tr[O X_2],\ldots  
, \Tr[O X_m)]$, note that $X_O \in \bbR^m$, then
\begin{align}
    \tvert{\rrho - \ssigma}_{\calO, p} &= \sup_{O \in \calO} \norm{(\rrho-\ssigma)_O}_p \\
    &= \sup_{O \in \calO} \norm{\rrho_O - \ssigma_O}_p
     \, .\nonumber
\end{align}
In the following, we will also need to understand linear transformations of vectors of operators.
For a matrix $A \in \bbC^{m' \times m}$ with elements $a_{i,j}$, we denote
\begin{align}
    \norm{A}_{p \to q} \coloneqq \sup_{\norm{\xx}_p = 1} \norm{A\xx}_q
\end{align}
and define multiplication between a matrix and a operator-valued vector in the usual sense as
\begin{align}
    A \XX \coloneqq \begin{pmatrix}
        a_{1,1} X_1 + a_{1,2} X_2 +\ldots  + a_{1,m} X_m \\
        a_{2,1} X_1 + a_{2,2} X_2 +\ldots  + a_{2,m} X_m \\
        \vdots \\
        a_{m',1} X_1 + a_{m',2} X_2 +\ldots  + a_{m',m} X_m
    \end{pmatrix} \, .
\end{align}
Then, the semi-norm is submultiplicative, \emph{i.e.}, 
\begin{align}
    \tvert{A\XX}_{\calO, p} \leq \norm{A}_{p \to p} \tvert{\XX}_{\calO, p} \, . \label{eqn:sec_NoA_submultiplicativity}
\end{align}
As already hinted at, in the 
special case $p=2$, there is a connection between the semi-norms defined in \cref{def:sec_NoA_semi-norm_parametrized_operators} and \cref{def:sec_NoA_semi-norm_vectors_of_operators} via an analogue of Parseval's Theorem. 
\begin{lemma}[Parseval's Theorem for induced norms] \label{lem:sec_NoA_parseval}
    Let $X$ be a parametrized operator $X\colon\calX \to \calM_n$ and consider the expansion into an orthonormal basis $\{\varphi_k(\xx)\}_{k \in \Lambda}$ as
    \begin{align}
        X(\xx) = \sum_{k \in \Lambda} \alpha_k \varphi_k(\xx).
    \end{align}
    Denote with $\aalpha$ the operator-valued vector consisting of the coefficient operators $\alpha_k$. Then,
    \begin{align}
        \tvert{X(\cdot)}_{\calO, 2} = \tvert{\aalpha}_{\calO, 2} \,.
    \end{align}
\end{lemma}
\begin{proof}
We have
\begin{align}
    \tvert{X(\cdot)}_{\calO, 2} &= \sup_{O \in \calO} \norm{\Tr [OX(\cdot)]}_2 \\
    \nonumber 
    &= \sup_{O \in \calO} \norm{\left( \Tr [O\alpha_{k_1}], \Tr [O\alpha_{k_2}],  \dots  \right)}_2 \\
    \nonumber
    &= \tvert{\aalpha}_{\calO, 2} \, .
\end{align}
To arrive at the second equality, we have used Parseval's theorem for the $L^2$-space. Note the change from the $L^2$ function norm to the euclidean vector norm going from the first to the second equality.
\end{proof}

Having established natural distance measures between parametrized quantum states, we are now ready to study the tomography of parametrized states, or, respectively, how to lift a tomographic procedure relative to a set of observables $\calO$ to a tomographic procedure for parametrized quantum states. Such a procedure should give a generalized guarantee of a form similar to
\begin{align}
    \bbP\big[ \tvert{ \rho(\cdot) - \hat\rho(\cdot) }_{\calO, p} \leq \epsilon \big] \geq 1 - \delta.
\end{align}

\subsection*{Sparse parametrized quantum states}

As already said, in this context we care a lot about the case when $\rho(\cdot)$ can be described using a set of indices $S \subseteq \Lambda$ such that $|S|$ is small. Even if that is not the case, it might still be that $\rho(\cdot)$ can be well approximated using a small set $S$ and is thus approximately sparse. This situation can arises for instance if a sparse parametrized quantum state is subject to noise. With the tools developed in this section, we can give a precise definition of sparsity and approximate sparsity. 

\begin{definition}[Approximately sparse parametrized quantum state]
    Let $\{\varphi_k(\xx)\}_{k \in \Lambda}$ be an orthonormal basis and 
    \begin{align}
        \rho(\xx) = \sum_{k \in \Lambda} \alpha_k \varphi_k(\xx)
    \end{align}
    a parametrized quantum state. 
    If there exists a a set $S$ of cardinality $|S|=s$ such that the approximation
    \begin{align}
        \rho_{S}(\xx) = \sum_{k \in S} \alpha_k \varphi_k(\xx)
    \end{align}
    has bounded error with respect to the induced function $p$-norm
    \begin{align}
        \tvert{\rho(\cdot) - \rho_S(\cdot)}_{\calO, p} \leq \gamma_{L^p},
    \end{align}
    or the induced vector $p$-norm
    \begin{align}
        \tvert{\aalpha(\cdot) - \aalpha_S(\cdot)}_{\calO, p} \leq \gamma_{\ell^p},
    \end{align}
    we call the parametrized quantum state $(s, \gamma_{L^p})$-sparse or $(s, \gamma_{\ell^p})$-sparse 
     with respect to the orthonormal basis $\{\varphi_k(\xx)\}_{k \in \Lambda}$ and a set of observables $\calO$.
\end{definition}

For $p=2$, per \cref{lem:sec_NoA_parseval}, we have
\begin{align}
    \tvert{\rho(\cdot)-\rho_S(\cdot)}_{\calO,2} &= \tvert{\aalpha - \aalpha_S}_{\calO, 2} \\
    &= \tvert{\aalpha_{\bar{S}}}_{\calO,2} \, ,\nonumber
\end{align}
where $\bar{S} = \Lambda \setminus S$. This means that in the case $p=2$, the two notions of sparsity are equivalent.

\section{Compressed sensing in orthonormal function bases} \label{sec:short introduction to compressed sensing}

Before we come to our algorithm for tomography of parametrized quantum states, we have to review the basics of compressed sensing~\cite{boche2015survey,CompressedSensingGitta}.
In signal theory, sampling rates were traditionally determined in accordance with the celebrated Nyquist-Shannon sampling theorem, which states that for a band-limited signal
\begin{align}
    f(x) = \sum_{k \in \Lambda} c_{k} e^{-i k x},
\end{align}
sampling at twice the maximum frequency $k_{\max} = \max \{ |k| : k \in \Lambda \}$ of the signal -- in other words, taking $O(|\Lambda|)$ many samples -- suffices to capture the signal in its entirety. 

This ultimate limit can, however, be beaten when additional information about the \emph{structure} of the signal is available. The most studied structural assumption is the aforementioned \emph{sparsity} of a signal, \emph{i.e.},  when a signal is only supported on $s$ frequencies $S \subseteq \Lambda$. In this case, it can be shown that evaluating the signal at $\tilde{O}(s)$ positions is sufficient to obtain the sparse vector of Fourier coefficients $\cc = (c_k)_{k\in\Lambda}$ that describes the signal. 
Other structures considered in compressed sensing are sparsity with respect to frames instead of orthonormal function bases~\cite{candes2011compressed}, hierarchical sparsity~\cite{eisert2021hierarchical},
or low rank structures
\cite{CandesRecht,gross_quantum_2010}. 

Let us explore how this can be achieved.
Formally, evaluating a general scalar signal 
\begin{align}
    f(\xx) = \sum_{k \in \Lambda} c_k \varphi_k(\xx)
\end{align}
at a point in parameter space gives us an observation that we can express as an inner product between the coefficient vector $\cc = (c_k)_{k\in\Lambda}$ and a vector $\aa_i = (\varphi_k(\xx_i))_{k\in\Lambda}$  that represents the particular values of the basis functions 
\begin{align}
    f(\xx_i) = \sum_{k \in \Lambda} c_k \varphi_k(\xx_i) = \langle \cc, \aa_i \rangle
\end{align}
at the point $\xx_i$. If we now obtain $M$ different observations, we can again arrange them in a vector $\ff = (f(\xx_i))_{i=1}^M$ and form a matrix $A$ whose columns are the vectors $\aa_i$ to arrive at the following linear relation between observations and the underlying coefficient vector
\begin{align}
    \begin{pmatrix}
        f(\xx_1) \\ f(\xx_2) \\ \vdots \\ f(\xx_M)
    \end{pmatrix} = \begin{pmatrix}
        \varphi_{1}(\xx_1) & \varphi_2(\xx_1) & \dots & \varphi_D(\xx_1)\\
        \varphi_{1}(\xx_2) & \varphi_2(\xx_2) & \dots & \varphi_D(\xx_2) \\
        \vdots &  \vdots & \ddots & \vdots \\
        \varphi_{1}(\xx_M) & \varphi_2(\xx_M) & \dots & \varphi_D(\xx_M) \\
    \end{pmatrix} \begin{pmatrix}
        c_{1} \\ c_2 \\ \vdots \\ c_D
    \end{pmatrix},
\end{align}
or in short $\ff = A \cc$. 

As we know both $\ff$ and $A$, we can try to solve the system of equations $\ff = A \cc$ for $\cc$. But if we obtain less than $D$ different observations the system of equations will be underdetermined and thus have many different solutions. Formally speaking, any vector from the \emph{kernel} of $A$, $\operatorname{ker}(A) = \{ \vv : A\vv = 0\}$, can be added to a solution $\cc$ and the result still satisfies the above system of equations. 

The standard approach to deal with an underdetermined system of equations is to choose the $\cc$ which has minimal $\ell^2$ norm, \emph{i.e.},  which solves the optimization problem 
\begin{align}
    \begin{split}
    \min\mathstrut &\ \lVert \cc \rVert_2 ,\\
    \text{ subject to}&\  \ff = A \cc.
    \end{split}
\end{align}
Interestingly, the above optimization problem has a closed-form solution that can be obtained from the \emph{pseudoinverse} of $A$. 
For injective $A$, the pseudo-inverse is defined as 
\begin{align}
    A^{+} \coloneqq (A^{\dagger}A)^{-1} A^{\dagger} \, .
\end{align} 
The solution to the optimization problem is then
\begin{align}
    \cc_{\ell^2} = A^{+} \ff.
\end{align}
This approach does, however, not make use of the additional knowledge we have that the solution to the system of equations should be \emph{sparse} and will usually not produce sparse vectors in the first place.

To produce a sparse solution, it is tempting to instead solve the optimization problem
\begin{align}\label{eqn:l0_optimization_problem}
    \begin{split}
    \min\mathstrut &\ \lVert \cc \rVert_0 ,\\
    \text{ subject to}&\  \ff = A \cc,
    \end{split}
\end{align}
where $\lVert \cc \rVert_0$ is the number of non-zero components of $\cc$. This optimization problem can be shown to be NP-hard~\cite{boche2015survey}. But this is no reason to despair because we can efficiently solve a relaxed version of the above problem 
\begin{align}\label{eqn:l1_optimization_problem}
    \begin{split}
    \min\mathstrut &\ \lVert \cc \rVert_1 ,\\
    \text{ subject to}&\  \ff = A \cc
    \end{split}
\end{align}
where $\lVert \cc \rVert_0$ is replaced with $\lVert \cc \rVert_1$.
In fact, this is the tightest convex relaxation of the above problem.
Under the guarantee that the true solution only has $s$ non-zero entries, the solution of the optimization problem of \cref{eqn:l1_optimization_problem} can be shown to coincide 
with the one of \cref{eqn:l0_optimization_problem} when $A$ has the so-called \emph{null space property} of order $s$~\cite{cohen2009compressed}. On an intuitive level, this property states that the kernel of $A$ contains only vectors that are not very sparse. 

While the null space property is both a necessary and a sufficient condition, it is not easy to 
work with. This is why most works in compressed sensing are based on a stronger sufficient notion called the \emph{restricted isometry property (RIP)} \cite{foucart2013mathematical}, 
which in turn implies the null-space property.

\begin{definition}[Restricted isometry property]
A matrix $A \in \bbC^{M \times D}$ has the \emph{restricted isometry property (RIP)} with constant $\Delta_s$ if for all $s$-sparse vectors $\vv$, we have that
\begin{align}
    (1-\Delta_s)\lVert \vv \rVert_2^2 \leq \lVert A \vv \rVert_{2}^2 \leq (1 + \Delta_s)\lVert \vv \rVert_2^2.
\end{align}
We will use the notation $\Delta_s(A)$ to denote the smallest $\Delta_s$ such that the above holds for a matrix $A$.
\end{definition}
An equivalent reading of the restricted isometry property is given by introducing a projector $\Pi_S$ onto all the entries in an index set $S$. Then, defining 
\begin{align}
    A_S \coloneqq A \Pi_S,
\end{align}
we see that the restricted isometry property is equivalent to 
\begin{align}
    \lVert \Pi_S - A_S^{\dagger} A_S \rVert_{\infty} \leq \Delta_s
\end{align}
for all sets of indices $S$ of cardinality at most $s$. In other words, to sparse vectors, the matrix $A$ looks like an isometry.

The theory of compressed sensing tells us that we can in principle determine the sparse coefficients $\cc$ as long as the matrix $A$ that contains the basis function values $\varphi_k(\xx_i)$ has the restricted isometry property with a low enough value of the RIP constant $\Delta_{2s} < 1/3$~\cite{boche2015survey}. Other, sometimes even more practical approaches to compressed sensing are available that have slightly different requirements on the RIP constant~\cite{rauhut2010compressive,foucart2011hard,foucart2013mathematical,lee2013oblique}.

While it is notoriously difficult to hand-craft deterministic matrices $A$ that have a good RIP constant
\cite{bandeira2013road, Bourgain_2011, Deterministic}, 
randomized methods perform surprisingly well \cite{boche2015survey,CompressedSensingGitta}. In the case of bounded orthonormal systems we care about, we can make use of the following Corollary of a result of Rauhut~\cite[Theorem 8.4]{rauhut2010compressive} that bounds the number of points $\xx_i$ we need to sample according to the orthogonality measure $\mu(\xx)$ of the bounded orthonormal system to achieve a desired RIP constant.

\begin{corollary}[Sample complexity bound]\label{corr:rip_sample_complexity_bos}
Let $A \in \bbC^{M \times D}$ be the random sampling matrix associated to a bounded orthonormal system $\{ \varphi_{\kk}(\cdot)\}$ with constant $K$. Then,
\begin{align}
    \bbP\big[ \Delta_s(A/\sqrt{M}) \leq \Delta \big] \geq 1 - \delta
\end{align}
if
\begin{align}
    M = \frac{s K^2}{\Delta^2} \left( C_1 \ln 300 s \ln 4 D + C_2 \ln \frac{2}{\delta}\right) \, ,
\end{align}
where the values of the constants do not exceed $C_1 \leq 103\,140$ and $C_2 \leq 2\,736$.
\end{corollary}
We see that we, indeed, get the desired dependence $\tilde{O}(s)$.

So far, we only considered the setting in which all the observations we obtain are perfect. 
In reality, however, the observations can come with errors, which means we observe $\hat{f}_i = f(\xx_i) + \eta_i$ for some error $\eta_i$. In this case, the condition $\mathrlap{\hat{\ff \,}}\hphantom{\ff} = A \cc$ might not be satisfiable, which in turn means that the $\ell^1$ optimization problem of \cref{eqn:l1_optimization_problem} cannot be solved. A relaxation of the problem, 
\begin{align}\label{eqn:l1_optimization_problem_with_errors}
    \begin{split}
    \min\mathstrut &\ \lVert \cc \rVert_1 ,\\
    \text{ subject to}&\  \lVert A \cc - \mathrlap{\hat{\ff \,}}\hphantom{\ff} \rVert_2 \leq \kappa,
    \end{split}
\end{align}
can, however, remedy this issue. The parameter $\kappa$ controls how much \enquote{slack} the solution can have, and given it is chosen suitably relative to the magnitude of the errors, the solution $\cc$ can be recovered with good accuracy even from noisy observations~\cite{boche2015survey}. The basic approach presented here has found numerous generalizations and improvements across the compressed sensing literature.

So far, the solution strategy for compressed sensing reconstruction was relaxing the problem to the convex optimization \cref{eqn:l1_optimization_problem} or, in the noisy setting, \cref{eqn:l1_optimization_problem_with_errors}. 
Next to that, a large family of greedy strategies for compressed sensing reconstruction exist. These algorithms are generally much easier to implement and offer very similar performance guarantees. The most prototypical example of these is \emph{iterative hard thresholding}~\cite{blumensath2009iterative}, which is both conceptually simple and still achieves surprising performance. The algorithm consists of starting from a candidate $\cc_0 = 0$ and then iteratively updating with 
the 
rule
\begin{align}\label{eqn:iht_update_rule}
    \cc_{t+1} = \calT_s[ \cc_t + A^{\dagger} (\mathrlap{\hat{\ff \,}}\hphantom{\ff}  - A \cc_t) ],
\end{align}
where $\calT_s$ is the \emph{thresholding operator} that cuts off all entries of the given vector except for the $s$ largest in absolute value. Essentially, iterative hard thresholding can be seen as gradient descent on the objective $\lVert \mathrlap{\hat{\ff \,}}\hphantom{\ff}  - A \cc \rVert_2^2$, interleaved with projections onto the set of sparse solutions. The greedy optimization algorithms for compressed sensing have also been extensively studied and improved. For a broader review of 
compressive sensing algorithms, see 
Ref.~\cite{foucart2013mathematical}.

\section{Tomography for parametrized quantum states} \label{sec:tomography of parametrized quantum states}

In this section, we explain how to perform tomographic procedures on parametrized quantum states to obtain approximate classical representations $\hat{\rho}(\cdot)$ of parametrized quantum states $\rho(\cdot)$
with rigorous error guarantees. 

\subsection*{Algorithm}
Our approach brings techniques from compressed sensing to the quantum realm, and as such generalizes the approaches presented in the preceding section. There, the goal was to reconstruct a scalar signal $f(\cdot)$. In the quantum case, we aim to find a good classical approximation of a parametrized quantum state
\begin{align}
    \rho(\xx) = \sum_{k\in \Lambda} \alpha_k \varphi_k(\xx).
\end{align} 

In the scalar case, the reconstruction has been obtained from observations of the signal at different parameter values.
To reconstruct a parametrized quantum state, we can proceed in a similar fashion.
We randomly sample $M$ parameter values $\xx_i \in \calX$ according to the orthogonality measure $\mu(\cdot)$ of the orthonormal system $\{ \varphi_k(\cdot) \}$ (see \cref{sec:parametrized quantum states}). 
Analogous to the scalar case, we can arrange the actual states in an operator-valued vector $\rrho=(\rho(\xx_1), \dots, \rho(\xx_M))$ such that
\begin{align}
    A \aalpha = \rrho \, ,
\end{align}
with 
\begin{align}
    \aalpha_k &= \alpha_k, \\
    A_{ik} &= \varphi_k(\xx_i) \, , \label{eqn:sec_ToP_measurement_matrix}
\end{align}
where the measurement matrix $A \in \bbC^{M \times D}$. Note that we can always write the basis functions $\varphi_k(\cdot)$ in a form that ensures that the index $k$ goes $k=1, 2, 3, \dots$, which we assume here for ease of notation. 
The fact that we sample the $\xx_i$ according to the measure $\mu(\cdot)$ for which the $\varphi_k(\cdot)$ form an orthonormal basis ensures we can use \cref{corr:rip_sample_complexity_bos} to establish a recovery guarantee later.

At the parameter values $\{\xx_i\}_{i=1}^M$, we perform tomographic procedures with $N(\epsilon_i, \delta_i, n)$ shots each. These constitute operator-valued {observations} $\hat{\rho}(\xx_i)$ which approximate the true quantum states $\rho(\xx_i)$. 
We also arrange the observations in an operator-valued vector $\hat{\rrho} = (\hat{\rho}(\xx_1), \dots, \hat{\rho}(\xx_M) )$. The fact that we use a tomographic procedure means that the approximation is not perfect, which is why we need to write
\begin{align}
    \hat{\rrho} = \rrho + \eeta = A \aalpha + \eeta.
\end{align}
Here $\eeta$ is an operator-valued vector that captures the deviations of observations $\hat{\rho}(\xx_i)$ from the true states $\rho(\xx_i)$. A guarantee for the control of the magnitude of these deviations $\tvert{\eta_i}_{\calO}$ is inherited from the tomographic procedure used to construct the approximations $\hat{\rho}(\xx_i)$. 
Here, we already see that the quantum setting holds more challenges than the scalar one, as there is a multitude of possible distance measures relative to which the error operators $\eta_i$ are considered \enquote{small}. 

One might now be temped to think that a generalization of the standard approaches for compressed sensing with noise presented in the preceding section would be straightforward. One can of course replace the norms $\lVert \cdot \rVert_{1/2}$ in either the optimization-based approach of \cref{eqn:l1_optimization_problem_with_errors} or the update rule of iterative hard thresholding given in \cref{eqn:iht_update_rule} with their induced semi-norm counterparts $\tvert{\cdot}_{\calO,1/\calO,2}$. Executing the algorithms would then, however, necessitate to hold the full vector $\aalpha$ in memory, which, as it is a vector of operators, is very inefficient. It would also require us to compute the norm $\tvert{\cdot}_{\calO,1/\calO,2}$ explicitly, which can also be numerically intractable, for example when the induced semi-norm is equal to the trace norm. 

We circumvent this problem by splitting the recovery procedure in two parts: one where we identify the sparse support $S$ of a parametrized quantum state (see \cref{sec:identification of a sparse support}) and a recovery procedure that takes this set as input, both using the same data. This section is devoted to showing that upon input of a set $S$, we can recover the parametrized quantum state on this support with bounded error. Note that this also subsumes the non-sparse case where $S = \Lambda$. At the end of this section, we show that in this particular case, the sample complexity of our algorithm can be further improved.

We now assume that we are given a set $S$ such that $\rho(\cdot)$ is $(\gamma_{\ell^1},s)$-sparse with respect to $S$. 
Given $S$, we can approximate the operator valued $\alpha_k$ with $k \in S$ by applying the pseudo-inverse of $A_S$ to the linear, operator valued equation $A\hat{\aalpha} = \hat{\rrho}$, obtaining
\begin{align}
    \hat{\aalpha} &= A_S^+ \left( A\aalpha + \eeta \right) \\
    \nonumber
    &= \aalpha_S + A_S^+ A_{\bar{S}} \aalpha_{\bar{S}} + A_S^+ \eeta \, . \label{eqn:sec_ToP_alpha_hat}
\end{align}
For the second equality, we have used that $A_S \xx_{S'}=0$ for disjoint sets $S, S'$. If the number of measurements $M$ is chosen sufficiently large (more on this in the error analysis), one can guarantee with high probability that (i) $A_S$ is injective, such that $A_S^+A_S = \Id_S$, which we also used for the second equality and (ii) that combined with estimates on $\eeta$ and $\aalpha_{\bar{S}}$, we are able to control the error $\tvert{\rho(\cdot)- \hat{\rho}_S(\cdot)}_{\calO,2} = \tvert{\aalpha - \hat{\aalpha}_S}_{\calO,2}$. In the limit $s \to D$ and $\tvert{\eeta}_{\calO,p} \to 0$ we recover the original parametrized state $\hat{\rho}(\cdot)=\rho(\cdot)$. We summarize this procedure in \cref{alg:classical_representation}.

\begin{algorithm}[H]
\caption{Sparse recovery of a parametrized quantum state} \label{alg:classical_representation}
\begin{algorithmic}
\Require $S$ \Comment{Support of sparse coefficients}
\Require $K$ \Comment{Bound of orthonormal system}
\Require $\Delta$ \Comment{Attenuation of spillover}
\Require $\epsilon$ \Comment{Tolerance}
\Require $\delta$ \Comment{Failure probability}
\State $s \gets |S|$
\State $M \gets \frac{sK^2}{\Delta^2} \left(C_1 \log^2(300s)\log(4D) + C_2\log\frac{2}{\delta} \right)$
\State $\epsilon' \gets \epsilon/\sqrt{6}$
\State $\delta' \gets {\delta}/{2 M}$
\For{$i \in [M]$}
\State \textbf{sample} $\xx_i \sim \mu$ \Comment{Sample from orthogonality measure}
\State $\hat\rho_i \gets \textbf{tomographic procedure}(\rho(\xx_i)), \epsilon', \delta')$ \\ \Comment{Perform tomographic procedure}
\State $(A_{i, \kk})_{\kk} \gets (\varphi_{\kk}(\xx_i))_{\kk}$ \Comment{Construct measurement matrix}
\EndFor
\State \textbf{compute} $A_S^{+}$ \Comment{Pseudoinverse}
\State \textbf{compute} $\hat\aalpha_S = A_S^{+} \hat\rrho$ \Comment{$s$-sparse vector}
\Ensure $\hat\rho_S(\cdot) = \sum_{\kk \in S} \hat\alpha_{\kk} \varphi_{\kk}(\cdot)$ \Comment{$s$-sparse approximation}
\end{algorithmic}
\end{algorithm}

\subsection*{Error analysis}

To bound the error between parametrized quantum state and classical approximation, we start from \cref{eqn:sec_ToP_alpha_hat} to obtain
\begin{align}
    \tvert{\rho(\cdot) - \hat{\rho}_S(\cdot)}_{\calO,2} 
    &= \tvert{\aalpha_{\bar{S}} - A_S^+A_{\bar{S}}\aalpha_{\bar{S}} - A_S^+ \eeta}_{\calO,2} \\
    \nonumber
    \begin{split}
    \nonumber
        &\leq \tvert{\aalpha_{\Bar{S}}}_{\calO,2} + \tvert{A_S^+A_{\bar{S}}\aalpha_{\bar{S}}}_{\calO,2}  \\
    &\phantom{\leq} \qquad + \norm{A_S^+}_{\infty} \tvert{\eeta}_{\calO,2} .
    \end{split} \label{eqn:sec_ToP_error_terms}
\end{align}
To provide a rigorous error analysis for \cref{alg:classical_representation}, we need to establish control over the quantities \smash{$\tvert{A_S^+A_{\bar{S}} \aalpha_{\bar{S}}}_{\calO,2}$}, \smash{$\norm{A_S^+}_{\infty}$} and \smash{$\tvert{\eeta}_{\calO,2}$} while ensuring injectivity of $A_S$. The term \smash{$\tvert{\aalpha_{\Bar{S}}}_{\calO,2}$} can be bounded by the $\ell^2$ sparsity constant $\gamma_{\ell^2}$ and cannot be avoided in general.
We start by observing a bound on $\tvert{\eeta}_{\calO,2}$. 
\begin{lemma}[Bound on the deviations] \label{lem:sec_ToP_observation_deviation_bound}
    If observations $\hat{\rrho} = (\hat{\rho}(\xx_1), \dots, \hat{\rho}(\xx_M))$ are constructed from true states $\rrho = (\rho(\xx_1), \dots, \rho(\xx_M))$ at points $\{\xx_i\}_{i=1}^M$ using an $(\epsilon_i, \delta_i, n)$ tomographic procedure, then the deviations $\eeta = \hat{\rrho}-\rrho$ are bounded as 
    \begin{align}
    \lVert \eeta \rVert_{\calO,p}  &\leq \norm{\eepsilon}_p \,,
\end{align}
where $\eepsilon = (\epsilon_1, \dots, \epsilon_M)$, with probability at least $1-\sum_{i=1}^M \delta_i$. 
\end{lemma}
Note that for constant $\epsilon_i = \epsilon$, one obtains
\begin{align}
    \tvert{\eeta}_{\calO,p} \leq M^{\frac{1}{p}} \epsilon \, .
\end{align}
\begin{proof}
We know that the entries of $\eeta$ fulfill
\begin{align}
    \lVert \eta(\xx_i) \rVert_{\calO} = \lVert \hat\rho(\xx_i) - \rho(\xx_i) \rVert_{\calO} \leq \epsilon_i
\end{align}
with probability greater than $1-\delta_i$ by the definition of a tomographic procedure. As such, the desired result holds with probability at least
\begin{align}
    \prod_{i=1}^M (1- \delta_i ) \geq 1 - \sum_{i=1}^M \delta_i
\end{align}
by the union bound. 
\end{proof}

Next, we give a bound for the term \smash{$\tvert{A_S^+A_{\bar{S}}\aalpha_{\bar{S}}}_{\calO,2}$} which quantifies the \enquote{spillover} from imperfect sparsity into our estimate.
\begin{lemma}[Imperfect sparsity]\label{lem:sec_ToP_approximately_sparse_error_term}
    If $\tvert{\aalpha}_{\calO, 1} \leq \gamma_{\ell^1}$ and $\Delta_{2s}(A/\sqrt{M}) \leq \Delta/2 \leq 1/2$, then
    \begin{align}
        \tvert{A_S^+A_{\bar{S}} \aalpha_{\bar{S}}}_{\calO,2} \leq \Delta \gamma_{\ell^1} \, .
    \end{align}
\end{lemma}
\begin{proof}
    We give a sketch of the proof, for details see \cref{sec:appendix_proof_sparse_spillover}.
    For two disjoint sets $S$, $S'$ of equal cardinality $s$, we can control \smash{$\lVert A_S^\dagger A_{S'} \rVert_{\infty} \leq \Delta_{2s}(A)$}~\cite{foucart2013mathematical}. Armed with this result for disjoint sets of equal cardinality, one then splits $\bar{S}$ into sets of cardinality at most $s$ and then bounds the $\ell^2$ with the $\ell^1$ norm to obtain the desired result.
\end{proof}

Lastly, we state an error bound on the term $\norm{A_S^+}_{\infty}$.
\begin{lemma}[Error bound to $\norm{A_S^+}_{\infty}$] \label{lem:sec_ToP_inf_norm_bound}
    If $\Delta_s(A/\sqrt{M}) \leq \Delta \leq 1/2$, then 
    \begin{align}
        \norm{A_S^+}_{\infty} \leq \frac{1}{\sqrt{M}}\frac{\sqrt{1+\Delta}}{1-\Delta} \leq \sqrt{\frac{6}{M}} \, .
    \end{align}
\end{lemma}
\begin{proof}
    By definition, the bound on the RIP constant $\Delta_s(A/\sqrt{M})$ ensures that the singular values of all submatrices of size at most $s$ are in the interval $[1-\Delta, 1+\Delta]$.  Then, using the definition of the pseudo inverse for injective matrices, we find
    \begin{align}
        \norm{A_S^+}_{\infty} &= \norm{\left(A_S^\dagger A_S\right)^{-1} A_S^\dagger}_{\infty} \\
        &\leq \frac{1}{\sqrt{M}}\norm{\left(\frac{A_S^\dagger}{\sqrt{M}} \frac{A_S}{\sqrt{M}}\right)^{-1}}_{\infty} \norm{\frac{A_S^\dagger}{\sqrt{M}}}_{\infty}
        \nonumber\\
        &\leq \frac{1}{\sqrt{M}} \frac{\sqrt{1+\Delta}}{1-\Delta} \, .
        \nonumber
    \end{align}
    The full statement follows from the fact that this expression is monotonically increasing in $\Delta$ and evaluating at $\Delta = 1/2$.
\end{proof}

Combining these three lemmas, we arrive at the desired error bound for the quantity $\tvert{\rho(\cdot)-\hat{\rho}_S(\cdot)}_{\calO,2}$ and therefore a recovery guarantee for our \cref{alg:classical_representation}.

\begin{theorem}[Sample complexity bound] \label{thrm:sec_ToP_error_bound}
Let $\rho(\cdot)$ be a parametrized quantum state with coefficient vector $\aalpha$ relative to a bounded orthonormal system with constant $K$ that is $\gamma_{\ell^1}$- and $\gamma_{\ell^2}$-sparse with respect to a set of indices $S$ of cardinality $s$. The output of \cref{alg:classical_representation} satisfies
\begin{align}\label{eqn:algorithm_output_abstract_guarantee}
    \bbP\big[\tvert{\rho(\cdot) - \hat\rho_S(\cdot)}_{\calO,2} &\leq \gamma_{\ell^2} + \Delta \gamma_{\ell^1} + \epsilon \big] \geq 1 - \delta
\end{align}
for a parameter $0 < \Delta \leq 1$ using a number
\begin{align}
    M = \frac{s K^2}{\Delta^2} \left( C_1 \ln 300 s \ln 4 D + C_2 \ln \frac{2}{\delta}\right)
\end{align}
of randomly sampled parameters and a total number of samples
\begin{align}
    N(\epsilon, \delta, n) = M \, T\left( \frac{\epsilon}{\sqrt{6}}, \frac{\delta}{2M}, n\right).
\end{align}
The procedure furthermore guarantees that $\Delta_{3s}(A/\sqrt{M}) \leq 1/2$. The values of the constants do not exceed $C_1 \leq 103\,140$ and $C_2 \leq 2\,736$.
\end{theorem}
\begin{proof}
    $M$ is chosen according to \cref{corr:rip_sample_complexity_bos} such that that $\Delta_{3s}(A/\sqrt{M}) \leq \Delta/2$ with probability at least $1-\delta_{\text{RIP}}$, such that \cref{lem:sec_ToP_approximately_sparse_error_term} and \cref{lem:sec_ToP_inf_norm_bound} apply. Those together with \cref{lem:sec_ToP_observation_deviation_bound} applied to \cref{eqn:sec_ToP_error_terms} give the desired result. Choosing $\epsilon' = \epsilon/\sqrt{6}$, $\delta' = \delta/2M$ and $\delta_{\text{RIP}} = \delta/2$ yields the desired statement.
\end{proof}   

Let us reflect on this result: We first observe that the sample complexity of our reconstruction is of order $O(s \log s \log D)$, which is much more efficient than full recovery in the practically relevant setting where $s\ll D$.
We observe that the error of our reconstruction is bounded by three separate contributions. The constant $\gamma_{\ell^2}$ quantifies how well the optimal set of indices $S$ approximates the true signal and is therefore unavoidable. The next term $\Delta \gamma_{\ell^1}$ stems from the possible spillover of contributions outside of $S$ into our estimate that stem from the fact that in compressed sensing, we can only control subsets of cardinality $s$. We can mitigate this term by making the parameter $\Delta$ smaller, at inverse quadratic cost. The last term, $\epsilon$, captures the inherent error of our tomographic procedure and can therefore also be reduced with a cost depending on the specific procedure chosen, usually this cost is also inverse quadratic in $\epsilon$. As such, \cref{thrm:sec_ToP_error_bound} provides a way for us to approximate $\rho(\cdot)$ as well as possible on a sparse support in a controlled way.

One might wonder why $M$ is chosen such that a RIP constant of order $3s$ is guaranteed, despite the fact that \cref{lem:sec_ToP_approximately_sparse_error_term} and \cref{lem:sec_ToP_inf_norm_bound} only require RIP constants of order $2s$ and $s$, respectively. The reason for this is that this ensures that the data we use for the sparse recovery can equivalently used in the support identification procedure introduced in the next section. 

\subsection*{Full recovery}

The preceding discussion and analysis pertains to the case of sparse recovery, where an approximation is only performed on a small subset $S$ of the whole set of basis functions $\Lambda$. Often, full recovery, \emph{i.e.},  the setting where $S = \Lambda$ and $s = D$, is still desired. This is for example the case when there is no known sparsity structure or simply when full knowledge of $\rho(\cdot)$ is wanted. In this case, the 
analysis of \cref{thrm:sec_ToP_error_bound} still applies and the constants $\gamma_{\ell^1} = \gamma_{\ell^2} = 0$ by construction. In this setting, however, the result we have used to impose a certain RIP constant on the measurement matrix $A$, \cref{corr:rip_sample_complexity_bos}, is sub-optimal as the additional sparsity considerations are not needed anymore. By using a different result, we can improve both the logarithmic dependencies as well as the constants to obtain the following guarantee for full recovery through the simplified \cref{alg:full_recovery}.
\begin{algorithm}[H]
\caption{Full recovery of a parametrized quantum state} \label{alg:full_recovery}
\begin{algorithmic}
\Require $K$ \Comment{Bound of orthonormal system}
\Require $\epsilon$ \Comment{Tolerance}
\Require $\delta$ \Comment{Failure probability}
\State $M \gets C D K^2 \log ({2D}/{\delta})$
\State $\epsilon' \gets \epsilon/\sqrt{6}$
\State $\delta' \gets {\delta}/{2 M}$
\For{$i \in [M]$}
\State \textbf{sample} $\xx_i \sim \mu$ \Comment{Sample from orthogonality measure}
\State $\hat\rho_i \gets \textbf{tomographic procedure}(\rho(\xx_i)), \epsilon', \delta')$ \\ \Comment{Perform tomographic procedure}
\State $(A_{i, \kk})_{\kk} \gets (\varphi_{\kk}(\xx_i))_{\kk}$ \Comment{Construct measurement matrix}
\EndFor
\State \textbf{compute} $A^{+}$ \Comment{Pseudoinverse}
\State \textbf{compute} $\hat\aalpha = A^{+} \hat\rrho$
\Ensure $\hat\rho(\cdot) = \sum_{\kk \in \Lambda} \hat\alpha_{\kk} \varphi_{\kk}(\cdot)$
\end{algorithmic}
\end{algorithm}

We obtain the following recovery guarantee.
\begin{theorem}[Full recovery guarantee]\label{thrm:sec_ToP_error_bound_full_recovery}
Let $\rho(\cdot)$ be a parametrized quantum state with coefficient vector $\aalpha$ relative to a bounded orthonormal system with constant $K$. The output of \cref{alg:full_recovery} satisfies
\begin{align}\label{eqn:algorithm_output_abstract_guarantee}
    \bbP\big[\tvert{\rho(\cdot) - \hat\rho_S(\cdot)}_{\calO,2} &\leq \epsilon \big] \geq 1 - \delta
\end{align}
using a number
\begin{align}
    M = CDK^2 \log \frac{2D}{\delta}
\end{align}
of randomly sampled parameters and a total number of samples
\begin{align}
    N(\epsilon, \delta, n) = M T\left( \frac{\epsilon}{\sqrt{6}}, \frac{\delta}{2M}, n\right).
\end{align}
The procedure furthermore guarantees that $\Delta_{3s}(A/\sqrt{M}) \leq 1/2$. The value of the constant does not exceed $C \leq 11$.
\end{theorem}
\begin{proof}
Theorem~12.12 of Ref.~\cite{foucart2013mathematical} posits that when sampling in a bounded orthonormal system with constant $K$,
\begin{align}
    M = \frac{8}{3}\frac{D K^2}{\Delta^2} \log \frac{2D}{\delta_{A}}
\end{align}
samples are sufficient to guarantee that the singular values of $A/\sqrt{M}$ lie in the interval $[1-\Delta, 1+\Delta]$ with probability at least $1-\delta_A$. We now proceed as in the proof of \cref{thrm:sec_ToP_error_bound} but choose $\Delta = 1/2$, $\epsilon'=\epsilon/\sqrt{6}$, $\delta' = \delta/2M$ and $\delta_A = \delta/2$ to arrive at the desired statement. This yields a constant $C = 32/3 \leq 11$.
\end{proof}

\subsection*{Predicting expectation values}
On the level of the data, the approximation of $\rho(\cdot)$, $\hat\rho(\cdot)$, is given by the coefficients $\{ \hat\alpha_k \}_{k \in S}$. The $\hat\alpha_k$ themselves are given as linear combinations of observations,
\begin{align}
    \hat{\alpha}_k = \sum_{i=1}^M (A_S^+)_{k,i} \,\hat{\rho}(\xx_i).
\end{align}
We can use this data to predict the expectation value of an observable $O$ evaluated on $\rho(\xx)$ for all $\xx$.
\begin{align}
\nonumber
    \Tr[O\hat{\rho}(\xx)] &= \sum_{k \in S} \Tr[O \alpha_k] \varphi_k(\xx) \\
    \nonumber
    &= \sum_{k \in S} \sum_{i=1}^M (A_S^+)_{k,i} \, \varphi_k(\xx) \Tr[O\hat{\rho}(\xx_i)]  \\
    &= \sum_{i=1}^M m_i(\xx) \Tr[O\hat{\rho}(\xx_i)] . \label{eqn:coeffs_for_prediction}
\end{align}
The coefficients $m_i(\xx)$ are given as
\begin{align}
    m_i(\xx) \coloneqq \sum_{k \in S} (A_S^+)_{k,i} \varphi_k(\xx) .
\end{align}
The observations $\hat{\rho}(\xx_i)$ are classical representations obtained from a tomographic procedure, therefore there is a protocol to obtain expectation values $\Tr[O\hat{\rho}(\xx_i)]$. Overall, this is efficient as long as computing $\Tr[O\hat{\rho}(\xx_i)]$ is efficient. 
This is not the case for full tomography, where the observations $\hat{\rho}(\xx_i)$ are density matrices. For local Clifford classical shadows, however, it would be efficient for all observables for which the shadow protocol gives guarantees.

In principle, parametrized expectation values could also be computed without performing the recovery of $\rho(\cdot)$ in terms of $\hat\rho(\cdot)$ that we described it in this section. This is because
expectation values could also be computed from the observations $\hat{\rho}(\xx_i)$ directly. For a fixed observable $O$, the values $\Tr[O\hat{\rho}(\xx_i)]$ are noisy measurements of the signal 
\begin{align}
    f(\xx) &= \Tr[O\rho(\xx)] \\
    &= \sum_{k \in \Lambda} c_k(O) \varphi_k(\xx) .\nonumber
\end{align}
Therefore, the coefficients $c_k(O)$ could be approximately recovered from the linear system
\begin{align}
    A\hat{\cc}_O = \hat{\rrho}_O 
\end{align}
using compressed sensing techniques for scalar signals as presented in \cref{sec:short introduction to compressed sensing}. This, however, would require us to run a compressed sensing algorithm for every new observable we want to evaluate, incurring unnecessary overheads. 

Furthermore, this approach would not give rise to a meaningful tomography of the parametrized quantum state $\rho(\cdot)$ where we want explicit access to the operator-valued coefficients $\alpha_k$.
Our algorithm gives rise to an operationally meaningful classical representation of the parametrized state $\rho(\cdot)$, where we have the same access to the coefficients $\hat{\alpha}_k$ that we have to the observations $\hat{\rho}(\xx_i)$. For example, when constructing the observations using full state tomography, we can explicitly give the coefficients $\hat{\alpha}_k$ in matrix form. From this, we can for example compute a classical representation of the reduced parametrized state $\hat{\rho}_A(\cdot)=\Tr_A[\hat{\rho}(\cdot)]$, where $A$ is a subsystem. Without direct access to the coefficients $\hat{\alpha}_k$, it is not obvious how to do this.

Direct access to the coefficients $\hat{\alpha}_k$ can also speed up the computation of the expectation value $\Tr[O\hat{\rho}(\xx)]$. For a classical shadow $\hat{\sigma}$ of a quantum state $\sigma$, expectation values are computed via snapshots $\hat{\sigma}_j$ as
\begin{align} \label{eqn:shadow_snapshots}
    \Tr[O\hat{\sigma}] = \frac{1}{J}\sum_{j=1} \Tr[O\hat{\sigma}_j] .
\end{align}
If we take our observations $\hat{\rho}(\xx_i)$ to be classical shadows, computing $\Tr[O\hat{\rho}(\xx)]$ means evaluating a weighted sum of classical shadows. Let $t$ denote the run time needed to compute the expectation value of a single classical shadow, then the overall computational complexity of computing $\Tr[O\hat{\rho}(\xx)]$ is $O(Mt)$. However, if we have access to the coefficients $\hat{\alpha}_k$ as a linear combination of observations, we can compute the coefficients $m_i(\xx)$ in \cref{eqn:coeffs_for_prediction} for a fixed $\xx$ before evaluating the expectation values $\Tr[O\hat{\rho}(\xx_i)]$. The precision of a prediction from a classical shadow scales with the number of snapshots evaluated. Evaluating 
\begin{align}
    \Tr[O\hat{\rho}(\xx)] = \sum_{i=1}^M m_i(\xx) \Tr[O\hat{\rho}(\xx_i)]
\end{align}
can be interpreted as a sampling problem, where we compute the values $\Tr[O\hat{\rho}(\xx_i)]$ from samples of the snapshots that make up the observations $\hat{\rho}(\xx_i)$. If we have access to the coefficients $m_i$ beforehand though, we can do much better than uniform sampling. Instead, we apply an importance sampling, where we sample from the distribution $\Tr[O\hat{\rho}(\xx_i)]$ with probability $|m_i(\xx)|/\sum_{i'} |m_{i'}(\xx)|$. It takes a moment of though to see that this reduces the computational time required to $O(t)$. Naturally, this method can also be applied if the observable $O$ is given as a sum $O=\sum_j h_j O_j$, where each component $O_j$ can be efficiently evaluated by a classical shadow. For an explicit calculation and further details on this see Ref.~\cite{guo2023estimating}.

\section{Identification of a sparse support} \label{sec:identification of a sparse support}
In this section, we tackle the question of identifying a sparse support $S$ such that $\rho_S(\cdot)$ is a good approximation to $\rho(\cdot)$ which we postponed in \cref{sec:tomography of parametrized quantum states}. 
While it might at first sound like a straightforward exercise, it comes with a lot of subtle difficulties
\cite{5766202}. 
The reason for this is that we only have indirect access to the coefficients $\alpha_k$ in the expansion of $\rho(\cdot) = \sum_{k \in \Lambda} \alpha_k \varphi_k(\cdot)$ and that we can usually not compute the norms $\lVert \cdot \rVert_{\calO}$ in a computationally efficient manner. Furthermore, we desire a computational procedure that is compatible with the underlying restrictions of the chosen tomographic procedure. Classical shadows~\cite{huang2020predicting}, for example, only allow an efficient computation of expectation values of Pauli words.

In an ideal world, we would like to find the set $S$ that achieves the best constant $\gamma_{\ell^2}$ in \cref{thrm:sec_ToP_error_bound}, \emph{i.e.}, 
\begin{align}
    S \coloneqq \argmin_{S \subseteq \Lambda, |S| = s} \tvert{\aalpha_{\bar{S}}}_{\calO,2}.
\end{align}
If we disregard computational complexity considerations, the data we gathered in \cref{alg:classical_representation} is sufficient to identify a good candidate for the support. The guarantees on the RIP constant of the measurement matrix $A$ allow us to run standard compressed sensing approaches for any fixed observable $O$, and as there is an observable whose scalar coefficients achieve the induced semi-norm $\tvert{\aalpha}_{\calO,2}$, one could in principle always find a good candidate for the support by trying enough observables.

In general, it is however computationally intractable due to the aforementioned constraints. We,  therefore, present a probabilistic algorithm that uses the expectation values of randomly chosen observables to guarantee that $\tvert{\aalpha_{\bar{S}}}_{\calO,2}$ is small for $\calO = \left\{ O : \norm{O}_2 \leq 1 \right\}$. In particular our method gives a guarantee that the coefficients $\alpha_k$ outside $S$ have small Hilbert-Schmidt norm. We believe this to be a reasonable proxy for sparsity -- first of all, the notions of sparsity coincide in the setting of exact sparsity where a set $S$ exists such that all $\alpha_k$ outside of $S$ are zero because the Hilbert-Schmidt norm is a proper norm. Second, the Hilbert-Schmidt norm subsumes all sets of observables with small Hilbert-Schmidt norm $\calO \subseteq \{ O : \lVert O \Vert_2 \leq 1 \}$ and is compatible with global Clifford shadows. 

\subsection*{Algorithm}

We give a protocol to estimate the Hilbert-Schmidt norm $\alpha_k$ such that the $s$ coefficient operators largest in Hilbert-Schmidt norm can be determined. We estimate from random Pauli measurements, which is convenient as it is hardware friendly as well as compatible with broadly used tomographic procedures like classical shadows. 

Recall the parameter dependent expectation value for some observable $O$, which is 
given by
\begin{align}
    \Tr[O\rho(\xx)] &= \sum_k \Tr[O\alpha_k] \varphi_k(\xx) \\
    &\eqqcolon \sum_k c_k(O) \varphi_k(\xx) \, .
    \nonumber
\end{align}
Now, the task of identifying $S$ is equivalent to identifying the indices of the $s$ largest expectation values $\bbE\left(|c_k(P)|^2\right)$, where $P$ is drawn uniformly at random from all $n$-qubit Pauli words. To see this, consider writing an operator $\alpha_k$ in terms of the normalized Pauli basis. We denote
\begin{align}
    \alpha_k = \sum_{l=1}^{d} a_{k,l} \frac{B_l}{\sqrt{d}} \, ,
\end{align}
where $d=4^n$, $a_{k,m} \in \bbC$ and $\{ B_l \}_{l=1,  \dots , d} = \{\Id, X, Y, Z\}^{\otimes n }$ is the normalized Pauli basis. 
In other words, we can identify the operator $\alpha_k$ with a vector $\aa_k$ corresponding to its Pauli expansion.
Note that 
\begin{align}
    c_k(B_l) &= \sqrt{d} |a_{k,l}|, \\
    \bbE[ |c_k(P)|^p ] &= d^{\frac{p}{2}-1} \norm{\aa_k}_p^p,
\end{align}
and especially
\begin{align}
    \norm{\alpha_k}_2 &= \norm{\aa_k}_2= \sqrt{\bbE\left[|c_k(P)|^2\right]} \, .
\end{align}

Realizations of the random variable $c_k(P)$ can be obtained by solving the linear system 
\begin{align} \label{eqn:sec_ISS_ideal_linear_system}
    A \cc(P) = \rrho_{P} \, ,
\end{align}
where we denote $\cc(P) := (c_1(P),  \dots , c_D(P))$ and $\rrho_{P} = (\Tr[P \rho(\xx_1)], \dots, \Tr[P\rho(\xx_M)])$. 
We only want to probe the parametrization at $M=O(\log(D))$ points. Therefore, $M \ll D$ and the above system is heavily underdetermined and can only be  reasonably queried by means of sparse recovery, thus further restricting our ability to read out vectors $\cc(P)$ directly.

Additionally, we do not have direct access to the above linear system, as we can only obtain observations $\hat{\rho}(\xx_i)$, not quantum states $\rho(\xx_i)$. Thus, we can only obtain information from the linear system
\begin{align} \label{eqn:sec_ISS_realistic linear system}
    A \hat{\cc}(P) &= \hat{\rrho}_{P} \\
    &= \rrho_{P} + \eeta_{P}  \, .
    \nonumber
\end{align}

In principle, any sparse recovery algorithm might be employed to extract information from \cref{eqn:sec_ISS_realistic linear system}. 
In practice, most sparse recovery algorithms work reasonably well \cite{foucart2013mathematical}. 
Here, we employ \emph{hard thresholding pursuit} (HTP) \cite{blumensath2009iterative}, mainly since it is easy to implement.
The constants in \cref{thrm:sec_ToP_error_bound} are chosen such that $\Delta_{3s}(A/\sqrt{M}) \leq 1/2$, which ensures convergence of HTP. Applying HTP to \cref{eqn:sec_ISS_realistic linear system}, we recover a vector $\cc^{\#}(P)$. The error between $\cc(P)$ and $\cc^{\#}(P)$ depends the maximum error rate over all Pauli observables
\begin{align}
    \eta \coloneqq \max_{P} \norm{\eeta_{P}}_2
\end{align}
and the distance between $\cc(P)$ and its best possible $s$-sparse approximation in the 1-norm
\begin{align}
    \sigma_s(\cc(P))_1 \coloneqq \min \left\{ \norm{\cc(P) - \zz}_1 : \text{$\zz$ is $s$-sparse} \right\} \, .
\end{align}
The error between $\cc(P)$ and $\cc^{\#}(P)$ then obeys the error bound
\begin{align} \label{eqn:sec_ISS_sparse_recovery_guarantee}
    \norm{\cc(P) - \cc^\#(P)}_2 &\leq \frac{D_1}{\sqrt{s}} \sigma_s(\cc(P))_1 + D_2 \norm{\eeta_{P}}_2 \\
    &\leq \frac{D_1}{\sqrt{s}} \sigma_s(\cc(P))_1 + D_2 \eta
    \nonumber
\end{align}
where $D_1, D_2$ are constants. This matches the intuition that the performance of sparse recovery algorithms depend on the noise level and the compressibility of the target vector, quantified here via $\sigma_s(\cc(P))$. For the right-hand side in \cref{eqn:sec_ISS_sparse_recovery_guarantee} we will use the shorthand 
\begin{align}
    \kappa(P) \coloneqq \frac{D_1}{\sqrt{s}} \sigma_s(\cc(P))_1 + D_2 \eta \, .
\end{align}
The algorithm we use to identify $S$ computes estimates of the largest expectation values 
\begin{align}\label{eqn:estimator_construction_supp_id}
    X_k \coloneqq \sqrt{\bbE[|c_k^{\#}(P)|^2]}
\end{align}
and outputs the support of the $s$ largest values as the set $S$. We estimate the expectation value $X_k$ by the empirical mean
\begin{align}
    \hat{X}_k \coloneqq \sqrt{\frac{1}{L} \sum_{l=1}^L |c_k^{\#}(P_l)|^2} \,,
\end{align}
where the Pauli observables $P_1, \dots, P_L$ are drawn uniformly at random. We summarize this procedure in \cref{alg:support_identification}.
\begin{algorithm}[H]
\caption{Support identification algorithm} \label{alg:support_identification}
\begin{algorithmic}
\Require $s$ \Comment{Size of sparse support}
\Require $\epsilon$ \Comment{Tolerance}
\Require $\delta$ \Comment{Failure probability}
\Require $\{\hat{\rho}(\xx_i)\}_{i=1}^M$ \Comment{$M$ s.t. $\Delta_{3s} \leq \frac{1}{2}$, see \cref{thrm:sec_ToP_error_bound}}
\State $L \gets   \log\left( \frac{2D}{\delta}\right) \frac{\left(1 + \kappa\right)^2}{2 \epsilon^4} $ \Comment{see \cref{thrm:sec_ISS_hoeffding}}
\For{$l=1, \dots, L$}
\State \textbf{solve} $A\hat{\cc}(P_l) = \hat{\rrho}_{P_l}$ \Comment{HTP algorithm}
\State \textbf{store} $\cc^{\#}(P_l)$
\EndFor
\For{k=1, \dots, D}
\State $\hat{X}_k \gets \sqrt{\frac{1}{L} \sum_{l=1}^L |c_k^{\#}(P_l)|^2}$
\EndFor
\Ensure $S = \left\{ k : \text{$\hat{X}_k$ belongs to $s$ largest} \right\}$ 
\end{algorithmic}
\end{algorithm}
Clearly, this algorithm does not necessarily succeed. It is not clear that the indices of the largest expectation values $X_k$ coincide with the largest values $\norm{\alpha_k}_2$. Note that the parameters $\epsilon$ and $\delta$ only control how well we approximate $X_k$ with the empirical mean $\hat{X}_k$ (see \cref{thrm:sec_ISS_hoeffding}) and have no implications for how close $X_k$ and $\norm{\alpha_k}_2$ are. Therefore, we analyze the requirements which are needed to ensure the success of \cref{alg:support_identification}.

\subsection*{Error analysis}
\cref{alg:support_identification} guarantees that the expectation values $X_k$ are estimated to precision $\epsilon$. We can ensure that picking the set $S$ associated to the largest expectation values $X_k$ is the correct set, only if there is a gap of at least $2\epsilon$ between the possible values in $S$ and in $\bar{S}$, or, formally, 
\begin{align} \label{eqn:sec_ISS_separability_criterion}
    \min_{k \in S} X_k - \max_{k' \in \bar{S}} X_{k'} \geq 2 \epsilon,
\end{align}
a condition that is not automatically fulfilled by construction. 
We can, however, relate this condition to the parameters of the parametrized quantum state and the guarantees of the compressed sensing procedure.

\begin{theorem}[Performance guarantees] \label{thrm:sec_ISS_separability_criterion_lowerB}
    If 
    \begin{align}
    \begin{split}
        \min_{k \in S} \norm{\alpha_k}_2 &- \max_{k' \in \bar{S}} \norm{\alpha_{k'}}_2  \geq \\
        &2\epsilon + \frac{2D_1}{\sqrt{s}} \sqrt{\bbE \left[\sigma_s(\cc(P))_1^2 \right]} 
          + 2D_2 \eta,
    \end{split}
    \end{align}
    then we can guarantee that
    \begin{align}
        \min_{k \in S} X_k - \max_{k' \in \bar{S}} X_{k'} \geq 2 \epsilon \, .
    \end{align}
\end{theorem}
\begin{proof}
    For any $P$, it follows from \cref{eqn:sec_ISS_sparse_recovery_guarantee} that $|c_k(P) - c^{\#}(P)| \leq \kappa(P)$. Applying the triangle inequality and taking the expectation value then leads to the desired result.
\end{proof}
We see that the success of \cref{alg:support_identification} depends on the magnitude of the squared expectation value of $\kappa(P)$, it is clear that as the size of the subset $s$ approaches $D$, this term vanishes and $S$ is identified correctly. Besides giving a condition on the success of \cref{alg:support_identification}, the theorem also dictates the necessary precision for the estimate $\hat{X}_k$. What is left is to determine the sample complexity required such that $|X_k - \hat{X}_k|\leq \epsilon$. 
First, note that $c_k^{\#}(P)$ is bounded. Denote with 
\begin{align}
    \kappa \coloneqq \max_P \kappa(P) \, .
\end{align}
the maximum error bound between vector $\cc(P)$ and $\cc^{\#}_k(P)$ over all Pauli observables.

\begin{lemma}[Coefficient bound]\label{lem:sec_ISS_coefficient_bound}
    For an $n$-Pauli observable 
    $P \in \{\Id, X, Y, Z\}^{\otimes n}$, we have that
    \begin{align}
        |c^{\#}_k(P)| \leq 1 + 
        \kappa \, . 
    \end{align}
\end{lemma}
\begin{proof}
    We start by observing a 
    bound on the absolute value of the original coefficients $c_k(O)$.
    By Hölder's inequality, $\Tr[O\rho(\xx)] \leq \norm{O}_{\infty}$. Then, it follows
    that
    \begin{align}
        \norm{f}_2^2 = \sum_k |c_k(O)|^2 &= \int_{\calX} \mathrm{d} \mu(\xx) \, |f(\xx)|^2 \\
        &\leq \norm{O}_{\infty}^2 
        \nonumber
    \end{align}
    and thus
    \begin{align} \label{eqn:origin_coeff_bound}
        |c_k(O)| \leq \norm{O}_{\infty} \, .
    \end{align}
    Now, apply the triangle inequality to a single term of the right-hand side in \cref{eqn:sec_ISS_sparse_recovery_guarantee} and apply \cref{eqn:origin_coeff_bound} to bound $|c_k(P)|$. Note that for any Pauli $P$, $\norm{P}_{\infty}=1$.
\end{proof}

As we deal with bounded random variables, we obtain performance guarantees by applying Hoeffding's inequality. 
\begin{theorem}[Performance guarantee for hard thresholding pursuit]\label{thrm:sec_ISS_hoeffding}
    Using HTP to probe linear systems of the form \cref{eqn:sec_ISS_realistic linear system} to obtain data vectors $\cc^{\#}(P)$, one achieves
    \begin{align}
        | \hat{X}_k - X_k |\leq  \epsilon 
    \end{align}
    with probability at least $1-\delta$ by using 
    \begin{align} \label{eqn:SIT_general_scaling_in_M}
        L \geq \log\left( \frac{2D}{\delta}\right) \frac{\left(1 + \kappa\right)^2}{2 \epsilon^4} 
    \end{align}
    many data vectors.
\end{theorem}
\begin{proof}
    We give a proof for completeness in \cref{sec:appendix_proof_estimator}.
\end{proof}
The scaling of $\epsilon^4$ naturally appears because we require guarantees on $\hat{X}_k$ and not on only on $\frac{1}{L}\sum_{l=1}^L |c^{\#}(P_l)|^2$.

\subsection*{Error analysis under additional assumptions}
Here, we introduce a reasonable additional assumption on the input of the support identification task and, as a consequence, provide a more refined analysis of the precision requirement for the success of \cref{alg:support_identification} given in \cref{thrm:sec_ISS_separability_criterion_lowerB}.
As a starting observation, the dependence on the term $\sigma_s(\cc(P))_1$ provides overly pessimistic bounds on the sample complexity in certain practical scenarios. 
To see this, consider first the following example: For all $k \in S$, let $|c_k(P)|=1$ for all $P$. For $k' \in \bar{S}$, let $|c_{k'}(P)|=x$ for all $P$, $x \ll 1$. Under these conditions, one would expect that a precision $\epsilon = O(1-x)$ suffices to provably determine the correct support. 

However, from \cref{thrm:sec_ISS_separability_criterion_lowerB}, specifically from the term $\sigma_s(\cc(P))_1$, we would conclude that a precision of $O((D-s)x)$ is necessary. This is due to the form of the error bounds for sparse recovery in \cref{eqn:sec_ISS_sparse_recovery_guarantee}, from which the dependence on $\sigma_s(\cc(P))_1$ is inherited. This indicates that for certain inputs, the hardness of the support identification task might not be best described by sparsity or compressibility, indicated by small $\bbE[\sigma_s(\cc(P))_1^2]$, but rather by some measure of distance between coefficients corresponding to $S$ and $\bar{S}$. The purpose of this section is to equip this intuition with rigorous results.

Throughout this paragraph, we will assume that what we call the \emph{local support identification} assumption holds true.

\begin{assumption}[Local support identification] \label{assptn:sec_ISS_local_support_identification}
Consider the linear system given in \cref{eqn:sec_ISS_realistic linear system}, which is
\begin{align}
    A\hat{\cc}(P) = \rrho_{P} + \eeta_{P} \, .
\end{align}
Then, the $s$ non-zero entries of the approximate solution $\cc^{\#}(P)$ obtained by means of sparse recovery coincide with the $s$ largest entries of the exact solution $\hat{\cc}(P)$.
\end{assumption}
It is important to note that we are not assuming that the non-zero entries coincide with the $s$ largest entries of the noise-free solution $\cc(P)$. \cref{assptn:sec_ISS_local_support_identification} becomes increasingly valid the larger the gap between the smallest of the $s$ largest values of $\hat{\cc}(P)$ and the largest of the rest. For a sufficiently large gap, \cref{assptn:sec_ISS_local_support_identification} is even provably true: According to \cref{eqn:sec_ISS_sparse_recovery_guarantee}, it holds that  
\begin{align}
    \norm{\hat{\cc}(P) - \cc^{\#}(P)}_2 \leq \frac{D_1}{\sqrt{s}} \sigma_s(\hat{\cc}(P))_1 \, .
\end{align}
As the gap increases, $\sigma_s(\hat{\cc}(P))_1$ goes to zero, such that at some point violating \cref{assptn:sec_ISS_local_support_identification} would immediately violate the above error bound.

As we will see, \cref{assptn:sec_ISS_local_support_identification} allows us to analyze a slightly different error model. We can express  the deviations in the parametrized quantum state as
\begin{align}
    \hat{\rho}(\xx) &= \sum_k (\alpha_k + \gamma_k) \varphi_k(\xx) \\
    &= \rho(\xx) + \eta(\xx) \, .
    \nonumber
\end{align}
If we solved the linear system $A\hat{\cc}(P) = \rrho_{P} + \eeta_{P}$ traditionally, \emph{i.e.},  without using sparse recovery techniques, we would obtain coefficients 
\begin{align}
    \hat{c}_k(P) &= \Tr[(\alpha_k + \gamma_k)P] \\
    &\coloneqq c_k(P) + n_k(P) \, .
    \nonumber
\end{align}
When using sparse recovery techniques, we instead obtain $c^{\#}_k(P)$, which are not easily decomposed in elementary terms, making the analysis of the precision requirements for the success of \cref{alg:support_identification} cumbersome. However, by virtue of \cref{assptn:sec_ISS_local_support_identification}, for any fixed Pauli $B_l$ we can write
\begin{align}
    c_k^{\#}(B_l) = \left(c_k(B_l) + n_k(B_l)\right) \chi_{k,l} \, ,
\end{align}
where 
\begin{align}
    \chi_{k,l} \coloneqq \chi\left(|c_k(B_l)+n_k(B_l)|>\max_{k' \in \bar{S}} |c_{k'}(B_l) + n_{k'}(B_l)| \right) \, ,
\end{align}
that is the indicator function which assumes $1$ if there is no entry in $\hat{\cc}(B_l)$ corresponding to an index in $\bar{S}$ that is larger than $\hat{c}_k(B_l)$. The role of \cref{assptn:sec_ISS_local_support_identification} is to give rise to the clear condition that defines the indicator function $\chi_{k,l}$. With this, we can write 
\begin{align} \label{eqn:sec_ISS_coefficient_expr_assumption}
    X_k = \left( \frac{1}{d} \sum_{l=1}^d \left|c_k(B_l) - n_k(B_l) \right|^2 \chi_{k,l} \right)^{\frac{1}{2}} \, ,
\end{align}
where, as a reminder, 
\begin{align}
    X_k \coloneqq \sqrt{\bbE\left[|c_k^{\#}(P)|^2\right]} \, .
\end{align}
The support identification task cannot succeed if the $s$ largest $X_k$ do not correspond to the $s$ largest $\norm{\alpha_k}_2$. From \cref{eqn:sec_ISS_coefficient_expr_assumption}, we see that there are essentially two mechanisms that can lead to such an outcome: (i) The values $c_k^{\#}(B_l)$ are significantly smaller than $c_k(B_l)$ for many measurements $B_l$, underselling the contribution of the $\alpha_k$. This is controlled by the magnitude of the noise $n_k(B_l)$. (ii) It cannot be ruled out that the coefficients $\hat{c}_{k'}(B_l) = |c_{k'}(B_l) + n_{k'}(B_l)|$ for $k' \in \bar{S}$ act as adversaries against a certain index $k \in S$, collectively trying to push $\hat{c}_k(B_l)$ out of the $s$ largest entries for as many measurements $B_l$ as possible. Despite being small in expectation, $\hat{c}_{k'}(B_l)$ can be large for few measurements $B_l$, such that a coordinated adversarial effort of all indices in $\bar{S}$ could mask the contributions an $\alpha_k$ with $k \in S$. One possibility 
to rule out such a case is requiring that for all measurements $B_l$, 
\begin{align}
    |\hat{c}_k(B_l)| > \sum_{k' \in \bar{S}} |\hat{c}_{k'}(B_l)| \, .
\end{align}
With this in mind, we can also better understand the term $\sigma_s(\cc(B_l))_1$ in the original precision requirement, which sums over contributions from indices that are mostly in $\bar{S}$.

However, such an adversarial coordination across multiple indices in $\bar{S}$ seems highly unlikely, which warrants going beyond a worst-case analysis. In order to do so, our key insight is that the ability to perform such adversarial attacks as described above is rooted in a large variance of the magnitude of $\hat{c}_{k'}(B_l)$ for different $B_l$. In an extreme case like in the above example, where there is no variation across different measurements $B_l$, no adversarial strategies can be employed at all. In order to capture this phenomenon, we denote the vectors $\cc(k') \coloneqq (c_{k'}(B_1), c_{k'}(B_2), \dots c_{k'}(B_d))$ and $\nn({k'}) \coloneqq (n_{k'}(B_1), n_{k'}(B_2), \dots, n_{k'}(B_d))$ and define the so-called \emph{flatness} constant as
\begin{align}
    \beta_{\cc({k'})} \coloneqq \frac{1}{\sqrt{d}} \frac{\norm{\cc({k'})}_2}{\norm{\cc({k'})}_{\infty}}
\end{align}
and analogous $\beta_{\nn({k'})}$ for $\nn({k'})$. 
The value lies in
$\beta_{\cc({k'})} \in [1, 1/\sqrt{d}]$, where $\beta_{\cc({k'})}=1$ is obtained for a perfectly flat (\emph{i.e.},  all entries are of the same magnitude) vector and $\beta_{\cc({k'})}=1/\sqrt{d}$ for a perfectly peaked (\emph{i.e.},  only one entry is non-zero) vector. 
We further denote
\begin{align}
    \beta_{\cc} \coloneqq \min_{k' \in \bar{S}} \beta_{\cc(k')}
\end{align}
and analogously $\beta_{\nn}$.
Now, with \cref{assptn:sec_ISS_local_support_identification} in place and the concept of flatness at hand to parametrize the hardness of the support identification task, we can give a more fine grained alternative to \cref{thrm:sec_ISS_separability_criterion_lowerB}. 
\begin{theorem}\label{thrm:sec_ISS_separability_criterion_lowerB_assumptions}
Choosing $\epsilon$ such that 
\begin{align}
\begin{split}
    &\min_{k \in S} \left\{\norm{\alpha_k} - \norm{\gamma_k}_2 \right\}   \\
     &\ \geq 2 \epsilon +  \frac{\beta_{\cc}+1 }{\beta_{\cc}}\max_{k'\in \bar{S}} \norm{\alpha_{k'}}_2 
     + \frac{\beta_{\nn}+1}{\beta_{\nn}}\max_{k' \in \bar{S}} \norm{\gamma_{k'}}_2
\end{split}
\end{align}
    guarantees correct identification of the sparse support.
\end{theorem}
\begin{proof}
    See \cref{sec:appendix_proof_lowerB_assumptions}.
\end{proof}
First, note that compared to \cref{thrm:sec_ISS_separability_criterion_lowerB}, we managed to characterize the precision requirement for successful support identification purely in terms of the Hilbert-Schmidt norms of the original coefficient operators and two flatness parameters. We observe that indeed, for certain inputs, the hardness of the support identification task is determined by some measure of distance between coefficients corresponding to $S$ and $\bar{S}$. The precision requirement only depends on the smallest $\norm{\alpha_k}_2$ in $S$ and the largest in $\bar{S}$, with the flatness constants quantifying to what extend this picture holds.  

\section{Applications}\label{sec:application}

The definitions and algorithms we propose in this work may appear quite abstract. In this section, 
we present two families of examples
that showcase the workings of the
approach. As a first family of examples,
we approach the problem of predicting all $\ell$-local reduced density matrices of an $n$-qubit state evolving under a Hamiltonian with an 
equally spaced spectrum, for which the established framework is particularly easily applicable. 
Within this first family of examples, we mainly consider 
Hamiltonians capturing settings in
\emph{nuclear magnetic resonance} (NMR)~\cite{smith1992hamiltonians} for all times $t$, under the promise that the initial state has a \emph{sub-Gaussian} energy spectrum, which will in turn induce a sparse structure on the time dependence.
We also point out, however, that a similar description is applicable to \emph{quantum many-body scars} \cite{PhysRevB.101.205107}, where perfect revivals of quantum states imply eigenstates with energies placed in an equally spaced ladder. 
As a second family of examples, we discuss how to similarly predict all $\ell$-local reduced density matrix of a \emph{fermionic} system undergoing a fermionic Gaussian time evolution, which is a non-interacting time evolution. 
This second example, in particular, shows that orthonormal function bases beyond the Fourier basis can be very useful in the tomography of parametrized quantum states.

\subsection*{Spectrally equally spaced  Hamiltonians}

We start by discussing in detail 
notions of tomography of parametrized quantum states at hand of systems in NMR.
We denote the NMR Hamiltonian with $H$ and assume its energies are in the range $[0, E_{\max}]$. If we let a state $\rho_0$ evolve under such a Hamiltonian, we obtain a parametrized quantum state
\begin{align}
    \rho(t) &= e^{-i H t} \rho_0 e^{i H t}.
\end{align}
NMR Hamiltonians have integer-valued energies, which means we can write 
$H = \sum_{e=0}^{E_{\max}} e \Pi_e$, where $\Pi_e$ are projectors onto the eigenspaces associated to energy $e$. Inserting this shows that we can expand $\rho(t)$ into a regular Fourier series
\begin{align}
    \rho(t) &= \sum_{e=0}^{E_{\max}}\sum_{e'=0}^{E_{\max}} e^{-i (e - e') t} \Pi_e \rho_0 \Pi_{e'} \label{eqn:energy_differences_nmr}\\
    &= \sum_{k=-E_{\max}}^{E_{\max}} \alpha_k \exp( i k t),\nonumber
\end{align}
where the operators $\alpha_k$ can be obtained by grouping all terms such that $e' - e = k$. 

In the language of our paper, the Fourier basis functions $\varphi_k(t) = \exp(i k t)$ constitute a bounded orthonormal system over the domain $\calX = [0, 2\pi]$ with orthogonality measure $\mu(t) = 1/(2\pi)$ and constant $K=1$. While there are in principle infinitely many basis functions, the fact that the energy of the Hamiltonian is bounded by $E_{\max}$ means it is sufficient to consider $\Lambda = \{ - E_{\max}, -E_{\max} + 1, \dots, E_{\max}-1, E_{\max}\}$.

Having established the structure of the parametrization of the time-evolved state, we now turn to the objective. We wish to recover all $\ell$-local reduced density matrices of the parametrized quantum state. This corresponds to the semi-norm induced by the set of observables
\begin{align}
    \calO_{\ell} = \{ O : \lVert O \rVert_{\infty} \leq 1, O \text{ is } \ell\text{-local} \}.
\end{align}
To this end, we will employ the tomographic procedure based on local Clifford shadows of Ref.~\cite{huang2020predicting}, which has a sample complexity
\begin{align}
    T(\epsilon, \delta, n) = O\left( \frac{\ell 12^\ell}{\epsilon^2}  \log \frac{n}{\delta} \right),
\end{align}
as already states in \cref{sec:a general framework for tomographic procedures}. 

With the objective clear, we now have to analyze the sparsity of the parametrized state with respect to the $\ell$-local trace norm.
We now assume that the energy spectrum of $\rho_0$ is sub-Gaussian, which means that there exist constants $\tau$, $\sigma$ and $e_0$ such that
\begin{align}
    \Tr[ \rho_0 \Pi_e ] \leq \tau \exp\left(-\frac12 \frac{(e - e_0)^2}{\sigma^2} \right).
\end{align}
This assumption naturally induces sparsity, as we only expect that a energies whose distance to $e_0$ is on the order of the standard deviation $\sigma$ will contribute significantly to the time evolution. In \cref{sec:appendix_sub-Gaussian_energy_spectrum} (\cref{lem:sub-Gaussian_energy_sparsity_guarantee}), we show that the sub-Gaussian assumption additionally allows us to control the $\ell^1$ sparsity defect $\gamma_{\ell^1}$. Explicitly, we can consider the sparse support $S = \{-R, -R+1, \dots, R - 1, R\}$ for
\begin{align}\label{eqn:R_for_sparse_tomography}
    R = \tilde{O}\left( \sigma \sqrt{n +  \log \frac{\tau}{\gamma}} \right)
\end{align}
and guarantee that $\gamma_{\ell^1} \leq \gamma$. With this result at hand, we can apply \cref{alg:classical_representation} with the guarantees of \cref{thrm:sec_ToP_error_bound} to obtain an efficient tomography scheme for $\rho(\cdot)$. Note that in this particular example, we can forgo the support identification procedure, because we can infer the sparse support from the properties of the initial state.

Bringing all ingredients together, we obtain a tomographic procedure whose output $\hat{\rho}(\cdot)$ fulfills the guarantee
\begin{align}
    \bbP\big[ \tvert{ \rho(\cdot) - \hat\rho(\cdot)}_{\calO_{\ell}, 2 } \leq 2 \gamma + \epsilon \big] \geq 1 - \delta
\end{align}
by performing local Clifford shadow tomography at 
\begin{align}
    M = \tilde{O}\left( \sigma \sqrt{n + \log \frac{\tau}{\gamma}} \log \frac{E_{\max}}{\delta}  \right)
\end{align}
uniformly sampled values of $t$. This is achieved by using \cref{thrm:sec_ToP_error_bound} with $\Delta = 1$, bounding $\gamma_{\ell^2} \leq \gamma_{\ell^1}$ and using $s = 2R + 1$ with $R$ given in \cref{eqn:R_for_sparse_tomography}. At every parameter value, a number of samples
\begin{align}
    T = O\left( \frac{\ell 12^\ell}{\epsilon^2} \log \frac{n M}{\delta} \right)
\end{align}
are used, resulting in a total number of samples of
\begin{align}
    N(\epsilon, \delta, n, \ell) = \tilde{O}\left( \frac{\ell 12^\ell}{\epsilon^2} \sigma \sqrt{n + \log \frac{\tau}{\gamma}} \log \frac{E_{\max}}{\delta} \right).
\end{align}

We stress that spectrally equally spaced Hamiltonians naturally emerge in the context of
\emph{quantum many-body scars}. Signatures of 
many-body scarring have first been observed when 
a system of 51 Rydberg atoms quenched out of equilibrium did not show a
relaxation to a thermal state, but instead featured distinct recurrences \cite{GreinerSpeedup}, stimulating a body of work on quantum many-body scars~\cite{PhysRevLett.119.030601,PhysRevB.101.205107,Scars}.
In Ref.\
\cite{PhysRevB.101.205107},
the converse question has been asked, namely whether recurrences would imply many-body scars, presenting an answer to the affirmative. If recurrences arise, a large number of eigenstates must exist with entanglement features that are typical for quantum many-body scars and with energy spectral values that are equally spaced. For systems prepared in this suitable subspace, all of the above discussion apply.

\subsection*{Gaussian fermionic time evolution}

As a second example, we consider the time evolution of fermions under an arbitrary fermionic Gaussian Hamiltonian. 
In particular, we look at $n$ fermionic modes. In the Majorana representation, we construct $2n$ Hermitian operators $\gamma_i$ that pairwise anti-commute $\{ \gamma_i, \gamma_j \} = \delta_{i,j}$. The first $n$ operators can be seen as \enquote{position} and the second $n$ operators as \enquote{momentum} operators. In this representation, a Hamiltonian is given by~\cite{surace2022fermionic}
\begin{align}
H &= i \sum_{i=1}^n \sum_{j = 1}^n F_{i,j} \gamma_i \gamma_j,
\end{align}
where $F$ is a real skew-symmetric matrix. We denote $J \coloneqq \max_{i,j} |F_{i,j}|$ as the \emph{interaction strength}.

These kinds of Hamiltonians create non-interacting dynamics. This manifests as the fact that we can bring the matrix $F$ to a normal form of $n$ non-interacting fermionic modes by performing an orthogonal transformation of the vector of Majorana operators, which in turn can be physically realized by a fermionic Gaussian unitary transformation of the $n$ modes of the system. 
Practically speaking, we can find a new set of Majorana operators $\tilde\gamma = O \gamma$ such that 
\begin{align}
    H &= i \sum_{i=1}^n \lambda_i (\tilde{\gamma}_i \tilde{\gamma}_{n+i} - {\tilde\gamma}_{n+i} \tilde{\gamma}_i),
\end{align}
where $\pm i \lambda_i$ are the eigenvalues of $F$. As promised, in this form the only Majoranas associated to the same transformed mode interact. 

One particular ingredient we will need shortly is the possibility to perform a fermionic operation that maps $H$ to $-H$ to emulate evolution into the negative time direction. This is easily achieved by the orthogonal transformation
\begin{align}
    \bbZ = \begin{pmatrix}
        \bbI & 0 \\
        0 & -\bbI
    \end{pmatrix}
\end{align}
which corresponds to an exchange of the fermionic annihilation and creation operators $\tilde{a}_i$, $\tilde{a}_i^{\dagger}$ associated to the transformed modes, which is nothing but a Pauli $X$ transformation on the transformed modes. Hence, this transformation can be realized by applying the unitary $U_O$ implementing the transformation $O$, $X^{\otimes n}$, and then the inverse $U_O^{\dagger}$
\begin{align}
    \tilde\gamma &\mapsto \bbZ \tilde\gamma ,
    \\
    H &\mapsto U_{O}^{\dagger} X^{\otimes n} U_O H U_O^{\dagger} X^{\otimes n} U_O = -H. \label{eqn:transformation_time_reversal}
\end{align}
We note that for a particle-preserving fermionic Hamiltonian, there are no terms that mix \enquote{position} and \enquote{momentum} operators in the first place, and as a consequence the time reversal can be emulated without the additional unitary $U_O$.

From the normal form of the Hamiltonian we can also deduce its spectrum completely, as all possible eigenvalues of $H$ are given by the possible sums of the $\pm\lambda_i$,
\begin{align}
    \operatorname{spec}(H) = \{ \pm \lambda_1 \pm \lambda_2 \dots \pm \lambda_n \},
\end{align}
which especially means that $\lVert H \rVert_{\infty} = \lVert F \rVert_1$ and that every eigenvalue can be associated to a vector $\kk \in \{-1,1\}^n$. We will therefore use the expansion 
\begin{align}
    H = \sum_{\kk \in \{ -1,1\}^n} \lambda_{\kk} \Pi_{\kk}.
\end{align}
The time evolution of a fermionic initial state $\rho_0$ under the aforementioned Hamiltonian is then given by
\begin{align}
    \rho(t) &= e^{-i H t} \rho_0 e^{i H t} \\
    &= \sum_{\kk \in \{ -1,1\}^n } \sum_{\lll \in \{ -1,1\}^n } e^{-i (e_{\kk} - e_{\lll}) t} \Pi_{\kk} \rho_0 \Pi_{\lll}
    \nonumber
    \\
    &= \sum_{\kk \in \{ -1,1\}^n } \sum_{\lll \in \{ -1,1\}^n } e^{-i \omega_{\kk,\lll} t} \Pi_{\kk} \rho_0 \Pi_{\lll} . \nonumber
\end{align}
Here, we note that the maximum frequency is bounded as
\begin{align}
    |\omega_{\kk, \lll}| \leq \omega_{\max} = 2 \lVert H \rVert_{\infty} = 2\lVert F \rVert_1 \leq 2 n^2 J,
\end{align}
where we have recalled $J = \max_{i,j} |F_{i,j}|$.

We can well-approximate this time evolution over the interval $[-1, 1]$ using Chebychev polynomials. The Chebychev polynomials are defined for $k \geq 0$ and are given by
\begin{align}
    T_k(t) = \cos ( k \arccos t).
\end{align}
They are the unique polynomials fulfilling the relation $T_k(\cos t) = \cos kt$ and fulfill the orthogonality relation
\begin{align}
    \int_{-1}^1 \diff t \, \frac{T_k(t) T_{l}(t)}{\sqrt{1-t^2}} &=  \delta_{kl} \times \begin{cases}
        \pi & \text{ if } k=0, \\
        \frac{\pi}{2} & \text{ if } k \geq 1.
    \end{cases}
\end{align}
To obtain a system of orthonormal basis functions in the sense of this work, we define the measure
\begin{align}
    \tilde\mu(t) = \frac{1}{\pi} \frac{1}{\sqrt{1-t^2}}
\end{align}
and the rescaled Chebychev polynomials
$T_k(.)$ that satisfy
\begin{align}
    \tilde{T}_k(t) = \xi_k T_k(t),
\end{align}
where $\xi_k = 1$ if $k=0$ and $\xi_k = \sqrt{2}$ otherwise. This ONB is bounded with constant $K = \sqrt{2}$.

We can express the time evolution in terms of this orthonormal function basis as
\begin{align}
    \rho(t) = \sum_{k=0}^{\infty} \tilde\alpha_k \tilde{T}_k(t).
\end{align}
We can find the coefficients $\tilde\alpha_k$ by computing the Chebychev expansion of the Fourier basis function $e^{-i \omega t}$. 
As we establish in \cref{sec:appendix_chebychev} (\cref{lemma:exponential_to_chebychev}), we have that over the interval $[-1, 1]$
\begin{align}
        e^{-i \omega t} &= \sum_{k=0}^{\infty} i^k \xi_k J_k(\omega) \tilde{T}_k(t),
\end{align}
where $J_k$ is the $k$-th Bessel function of the first kind. Inserting this expression and grouping by $\tilde{T}_k$ yields the following formula
\begin{align}
    \tilde\alpha_k &= i^k \xi_k \sum_{\kk \in \{ 0,1\}^n } \sum_{\lll \in \{ 0,1\}^n } J_k(\omega_{\kk,\lll}) \Pi_{\kk} \rho_0 \Pi_{\lll}.
\end{align}

To apply our algorithm, we need to specify a set $S$ of coefficients (in this case indices of the relevant Chebychev polynomials) and bound the $\ell^1$ sparsity defect relative to the norm induced by the set of observables of the used tomographic procedure. As in the preceding example, we are able to construct $S$ explicitly, such that we do not need the support identification procedure.
Again, we care about $\ell$-local observables. In this case, we can bound the sparsity defect in the induced norm by the sum of the trace norm of the coefficients $\tilde\alpha_k$ for $k \in \bar{S}$ as
\begin{align}
    \gamma_{\ell^1} &= \tvert{ \tilde\aalpha_{\bar{S}}}_{\calO_{\ell},1} 
\leq \sum_{k \in \bar{S}} \tvert{ \tilde\alpha_k }_{\calO_{\ell}} 
\leq \sum_{k \in \bar{S}} \lVert { \tilde\alpha_k }\rVert_{1} .
\end{align}
We can bound the magnitude of the coefficients $\tilde\alpha_k$ by exploiting the following upper bound on the Bessel function of the first kind (see \cref{lemma:bessel_bound} of \cref{sec:appendix_chebychev})
\begin{align}
    |J_k(\omega)| \leq \left(\frac{e\omega}{2k}\right)^{k},
\end{align}
which implies that for $m > e\omega/2$, the coefficients decay at least exponentially.
Formally, we have that
\begin{align}
    \lVert \tilde\alpha_k \rVert_1 &\leq |\xi_k| \sum_{\kk \in \{ 0,1\}^n } \sum_{\lll \in \{ 0,1\}^n } |J_k(\omega_{\kk,\lll})| \lVert \Pi_{\kk} \rho_0 \Pi_{\lll} \rVert_1\\
    &\leq 2^{2n+1/2} \left(\frac{e \omega_{\max} }{2 k}\right)^{k},
\end{align}
where we have upper-bounded $\omega \leq \omega_{\max}$ and $|\xi_k| \leq \sqrt{2}$. 

Our plan is to approximate $\rho(t)$ by only including terms with $k$ up to a cutoff $R$, i.e.\ by using the set $S = \{0, 1, \dots, R\}$. 
Let us now set $R = \lceil e \omega_{\max} \rceil + R'$ for $R' \geq 0$, which guarantees that for all $k > R$, $e \omega_{\max} / 2 k \leq 1/2$. 
In this case,
\begin{align}
    \sum_{k=R+1}^{\infty} \lVert \tilde\alpha_k \rVert_1 &\leq 2^{2n+1/2} \sum_{k=R+1}^{\infty} 2^{-k} \\
    \nonumber
    &= 2^{2n+1/2} 2^{-R-1} \sum_{k=0}^{\infty} 2^{-k} \\
    \nonumber
    &= 2^{2n +1/2 -R}
\end{align}
where we have evaluated the 
geometric series for the last step. 
To achieve a certain value $\gamma$ for the tail, we need
\begin{align}
    R \geq 2n + \frac{1}{2} + \log_2 \frac{1}{\gamma}.
\end{align}
This compares favorably to the minimal cutoff $\lceil e \omega_{\max} \rceil$, which we can bound as $\lceil 2en^2 J \rceil$, which is $O(n^2 J)$. 
This means that 
\begin{align}
    R &= \max\left\{ O\left( n + \log \frac{1}{\gamma} \right), O\left( n^2 J \right) \right\} \\
    &\leq O\left( n + n^2 J + \log \frac{1}{\gamma} \right)
    \nonumber
\end{align}
is sufficient to achieve the desired sparsity defect $\gamma_{\ell^1} \leq \gamma$.

Additionally to finding a set $S$ of coefficients that approximates $\rho(t)$ to small error, we also need to truncate the infinite series of $\alpha_k$ coefficients at some $k = D$ to manipulate it in the memory of a computer. The preceding analysis suggests that this can be done at logarithmic cost in the inverse error of the finite-size approximation. As the factor $D$ only enters into the sample complexity of the analysis as $\log D$, we have a doubly-logarithmic dependence on the finite-size approximation error which we can safely neglect in the following.

Having established that a low number of Chebychev basis functions are sufficient to well-approximate the fermionic Gaussian time evolution, we are left to combine this with a tomographic procedure to obtain an actual tomography protocol. As in the previous example, we will consider the estimation of all $\ell$-local observables with bounded operator norm, which corresponds to an $\ell$-local trace distance. A shadow tomography algorithm for this task tailored to fermionic systems was proposed in Ref.~\cite{zhao2021fermionic} and has a sample complexity of
\begin{align}
    T(\epsilon, \delta, n) = O\left(n^{\ell} \frac{\ell^{3/2}}{\epsilon^2} \log \frac{n}{\delta} \right),
\end{align}
where the factor $n^{\ell}$ comes from the binomial coefficients $n$ choose $\ell$ for constant $\ell$.

Exploiting the bound $\gamma_{\ell^2} \leq \gamma_{\ell^1} \leq \gamma$ and looking at the guarantees of \cref{thrm:sec_ToP_error_bound}, we have a tomographic procedure whose output $\hat{\rho}(\cdot)$ fulfills the guarantee
\begin{align}
    \bbP\big[ \tvert{ \rho(\cdot) - \hat\rho(\cdot)}_{\calO_{\ell}, 2 } \leq 2 \gamma + \epsilon \big] \geq 1 - \delta
\end{align}
by performing local fermionic shadow tomography at 
\begin{align}
    M = \tilde{O}\left( \left(n + n^2 J + \log \frac{1}{\gamma}\right) \log \frac{D}{\delta} \right)
\end{align}
values of $t$ sampled according to $\tilde\mu$. If $t < 0$, we apply the unitary transformation of \cref{eqn:transformation_time_reversal} and evolve for time $|t|$ to emulate a time evolution in the negative direction. At every parameter value, a number of samples
\begin{align}
    T =  O\left(n^{\ell} \frac{\ell^{3/2}}{\epsilon^2} \log \frac{n M }{\delta} \right)
\end{align}
are used, resulting -- for constant $J$ -- in a total number of samples of
\begin{align}
\begin{split}
    &N(\epsilon, \delta, n, \ell) \\
    & \qquad = \tilde{O}\left( n^{\ell + 2} J \frac{\ell^{3/2}}{\epsilon^2} \log \frac{n M D}{\gamma} \log^2 \frac{1}{\delta}  \right).
\end{split}
\end{align}
In this example, it was possible to construct $S$ explicitly. However, this is not fundamental to our techniques, as the above sample complexity remain unchanged even if $S$ is unknown and needs to be determined with the support identification procedure detailed in \cref{alg:support_identification}. 

Finally, we emphasize that the same algorithm allows us to recover the time evolution in the interval $[-T, T]$ instead of $[-1, 1]$ by simply rescaling the time parameter appropriately, at the expense of increasing the interaction strength by a factor of $T$, yielding a total sample complexity of
\begin{align}
\begin{split}
    &N(\epsilon, \delta, n, \ell) \\
    & \qquad = \tilde{O}\left( n^{\ell + 2} J T \frac{\ell^{3/2}}{\epsilon^2} \log \frac{n M D}{\gamma} \log^2 \frac{1}{\delta}  \right).
\end{split}
\end{align}

\section{Tomography of parametrized quantum channels}\label{sec:channel_tomography}
While we have formulated both our framework and our algorithm for parametrized quantum states, our strategies equally well apply to parametrized quantum channels. These arise very natural, for example in the context of quantum metrology, where the knowledge of the parameter-encoding evolution is crucial to devise good sensing protocols.
In this section, we outline how the different parts of our work are generalized.

\subsection*{Parametrized quantum channels}
A parametrized quantum channel is a function from the set of parameters $\calX$ to the set of quantum channels from system $\calM_n$ to $\calM_m$, which we denote by $\calN\colon \calX \to \calM_{n \to m}$. This means that for all $\xx \in \calX$, $\calN(\xx) \in \calM_{n \to m}$ needs to be a completely positive trace-preserving map. In other words, $\calN(\cdot)$ is a \emph{parametrized superoperator}.

Similar to a parametrized quantum state, we can expand a parametrized quantum channel in terms of an orthonormal function basis $\{ \varphi_k(\cdot) \}_k$ as
\begin{align}
    \calN(\xx) = \sum_{k} \calA_k \varphi_k(\xx),
\end{align}
where the coefficients of the expansions are now superoperator-valued and given by
\begin{align}
    \calA_k \coloneqq \int \diff \mu(\xx) \,  \varphi^{*}_k(\xx) \calN(\xx),
\end{align}
where $\mu$ is the orthogonality measure of the function basis.

\subsection*{Tomographic procedures}
We generalize the notion of a tomographic procedure from states to channels by generalizing the semi-norm induced by a set of observables. In the case of quantum channels, we have an additional degree of freedom next to the observable, namely the choice of the state that the quantum channel takes as input. Also, it is now necessary to include an ancillary system of the same dimension as the input. 
We can thus define a semi-norm for channels by a set of pairs of input states and observables $\calT$ as
\begin{align}
    \lVert \calN \rVert_{\calT} \coloneqq \sup_{(\rho, O) \in \calT} | \Tr[ (\bbI \otimes \calN)[\rho] O]|.
\end{align}
Another way of looking at this definition which is closer in form to the state definition is obtained by considering \emph{quantum combs}, also known as \emph{quantum process tensors}~\cite{gutoski2007toward,chiribella2009theoretical}. In this language, the combination of state preparation and the measurement of a POVM forms a single object called a \emph{tester}~\cite{ziman2008process}. Our definition slightly generalizes this idea by weighing the different outcomes of the POVM to obtain an expectation value. 

If we use the set $\calT$ constructed by combining all possible input states on the joint system of input and ancilla with all possible joint observables $\lVert O \rVert_{\infty}\leq 1$, the induced semi-norm will be equal to the diamond norm $\lVert \calN \rVert_{\diamond}$ which is the natural generalization of the trace norm on the level of quantum channels. Even more so than in the state case, performing tomography with respect to the diamond norm is very resource intensive. 

Efficient recovery of a quantum channel is, nevertheless, possible under relaxed assumptions. For example, the method of Ref.~\cite{caro2023learning} guarantees an efficient approximation if the states and observables in $\calT$ are restricted to \emph{Pauli-sparse} states and observables, \emph{i.e.}, whose Pauli expansion has only few terms. Another way to relax the requirements of channel tomography is to replace the joint supremum over inputs and observables with a combination of supremum and expected value. An example 
for this is Ref.~\cite{huang2022learning}, which shows that channels can be learned efficiently when we only consider local observables and take the expectation 
\begin{align}
    \operatornamewithlimits{\bbE}_{\rho} \left\{ \sup_{O \text{ local}} |\Tr[ \calN[\rho] O ] |\right\} 
\end{align}
over Haar random single-qubit states on the input.
As a parametrized superoperator becomes a parametrized scalar function upon fixing a combination of input and observable, the induced $L^p$ semi-norm of parametrized superoperators and the induced $\ell^p$ semi-norm of vectors of superoperators generalize in the natural ways and the Parseval Theorem also carries over.
As the notions of induced semi-norms generalize directly, so does the notion of an approximately sparse parametrized quantum channel.

\subsection*{Algorithm}
Our algorithm itself is agnostic to the underlying quantum object, and hence can be applied in the exact same way as in the case of a parametrized quantum state. The only thing that changes are the guarantees and the sparsity constants $\gamma_{\ell^1}$ and $\gamma_{\ell^2}$, which are now given relative to the semi-norm for channels.

\subsection*{Support identification}

The support identification algorithm we outlined in this work can also be generalized in a relatively straightforward manner. In it, we have used the Hilbert-Schmidt 2-norm as a proxy for the magnitude of the coefficients $\alpha_k$. The natural generalization of this is to use the Hilbert-Schmidt 2-norm of the superoperator-valued coefficients $A_k$. It can be obtained by preparing a maximally entangled state $\ket{\Omega}$ between the input and the ancillary register, then applying the quantum channel and subsequently measuring the expectation value of a Pauli word $P \otimes Q$. This gives us access to the entries of the Pauli transfer matrix of $\calA$ given by
\begin{align}
    \Tr[ (\bbI \otimes \calA_k)[|\Omega \rangle\!\langle \Omega |] (P \otimes Q)] = 2^{-n} \Tr[Q \calA_k [P]].
\end{align}
Choosing the Pauli operators uniformly at random then gives us an estimate of the coefficients, from which we can construct an estimator as in \cref{eqn:estimator_construction_supp_id}.

\section{A short practitioners guide}\label{sec:practicioners_guide}
For convenience, we give a more intuitive summary of our results and guidance for their practical application.

Tomography is the task of constructing a classical representation $\hat{\rho}(\cdot)$ of a parametrized quantum state $\rho(\cdot)$ which should be sufficiently close for all possible values of the parameters $\xx \in \calX$. Our method is based on the simple realization that we can always expand the parameter dependence in terms of a function basis $\{ \varphi_k(\cdot) \}$ 
\begin{align}
    \rho(\xx) &= \sum_{k \in \Lambda} \alpha_k \varphi_k(\xx)
\end{align}
for $\calX \to \bbC$.  
The choice of the function basis itself is crucial and will necessarily depend on the kind of parametrized quantum state we deal with. In this work, we considered two very different examples with very different properties, namely the Fourier basis on one hand and the Chebychev polynomials on the other hand. Note that we always have to assume that the state can be approximated with a finite number of basis functions $k \in \Lambda$ as we have to hold a vector of size $D = |\Lambda|$ in memory. 

It comes to no great surprise that for an arbitrary parametrized quantum state, the sample complexity must scale as $\Omega(D)$ as we have to recover all the coefficients $\{ \alpha_k \}$. We show that a randomized protocol can achieve this sample complexity in \cref{thrm:sec_ToP_error_bound_full_recovery} with favorable overheads. 

Interestingly, we can do much better if the parameter dependence of the state is \emph{sparse}, i.e.\ there exists a set $S \subset \Lambda$ of small cardinality $s = |S|$ such that
\begin{align}
    \rho_S(\xx) = \sum_{k \in S} \alpha_k \varphi_k(\xx)
\end{align}
approximates $\rho(\xx)$ well. Sparsity of the time dependence is a property that exists relative to the chosen function basis $\{ \varphi_k \}$, which means that a state that is sparse with respect to one basis will usually not be sparse with respect to another basis. The flip side of this is that there often exist bases that allow for a sparse representation of a specific parametrized quantum state $\rho(\cdot)$. The challenge lies in finding a function basis from first principles that allows for a sparse representation of a whole family of parametrized quantum states, something that we exemplified in our two example applications.

Given a sparse state, we present a protocol that performs tomography of the parametrized quantum state with a sample complexity of $\tilde{O}(s \log D)$, see \cref{thrm:sec_ToP_error_bound}, which means that the savings in samples for tomography of a parametrized quantum state can be exponentially large.  

A crucial property of our proposed protocols is that they use a state tomography scheme as a black box. This allows us to switch between different schemes of full or approximate tomography (like shadow tomography), while still retaining the guarantees of the tomographic procedure in a sensible way over the full parameter space $\calX$. 

It is important to mention a few caveats and details that are relevant in practice. 
When we say efficient, we generally mean sample efficient as well as computationally efficient. 
However, when choosing $\Lambda$ it should be noted that our algorithm involves computations on the $M \times D$ sized matrix $A$ and scales linearly in $D$ in computational resources~\cite{blumensath2009iterative}. 
This contrasts with the sample complexity which scales only logarithmic in $D$. 
When finding an orthonormal basis, care needs to be taken that the constant $K=\norm{\varphi_k}_{\infty}$ that bounds the function basis does not scale unfavorably, as this could negate some of the gains we get from our randomized procedures.

\section{Future directions}\label{sec:future_directions}
Our work opens up many exciting directions of research into the better characterization and understanding of parametrized quantum circuits. In this section, we sketch some of these
future directions that may emerge from our work.

\subsection*{Applications}
Key to the success of our methods is knowledge of an orthonormal basis in which the parametrized quantum state can be expressed with few basis elements. The main challenge in applying our techniques lies in identifying such good basis sets for problems of practical interest. Inspiration here might be taken from the vast literature on basis sets for computational methods in quantum chemistry and solid state physics, for which a body of heuristic knowledge is available \cite{Perlt,Scheffler}.
We expect that often, such basis sets will be found by observing what works in practice and will, therefore, be of a heuristic nature, not necessarily accompanied by formal proofs. However, we consider the rigorous study of basis sets and their approximation quality for parametrized quantum states to be of equal interest and as a question of fundamental relevance in its own right.

An interesting way of realizing compressed sensing approaches that work well in practice is to employ overcomplete bases or use dictionary methods that combine different bases together to induce a sparse representation of a signal~\cite{candes2011compressed,PhysRevB.90.115110,Ozolins}. Taking the 
same avenue to characterize parametrized quantum systems in practice should prove fruitful.

\subsection*{Algorithms}
The plug and play nature of our framework is well suited for the incorporation of other compressed sensing techniques, because large parts of the algorithm are agnostic to the actual technique used. However, as seen in \cref{sec:identification of a sparse support}, the extension to the quantum case brings its own difficulties and pitfalls and is not guaranteed to be straight forward.

There are many advanced techniques to perform compressed sensing in different scenarios that one could presumably integrate with our recovery algorithm. The aforementioned approach of using frames and overcomplete dictionaries allows to better recover signals that do not have a sparse expansion in a known orthonormal basis, a likely feature in some practical applications. Another promising approach allows one to replace the sampling from the orthogonality measure of the orthonormal basis $\{\mu(.)\}$, with a more measure that is more accessible in an experiment~\cite{lee2013oblique}.

Beyond harmonizing our protocols with more powerful tools from compressed sensing, there are also other directions that we find promising. As explained in \cref{sec:identification of a sparse support}, finding the optimal support $S$ involves a minimization that is computationally intractable. We, therefore, give a different figure of merit that we argue works well in practice. Can we find support identification protocols with yet other figures of merit? Here, it is interesting to device algorithms that find the largest coefficients $\alpha_k$ not in Hilbert-Schmidt norm, but in another induced semi-norms $\norm{\cdot}_{\calO}$. On the highest
level, these are ideas of model selection.
At the end of the day, one will have to strike a balance between mathematical and conceptual approaches \cite{Akaike} as well as
heuristic ones \cite{ExperimentalCompressed}.
This would, presumably, also require algorithms that efficiently evaluate the semi-norm of a given operator.

A further noteworthy point is that the procedure introduced in this manuscript is modular: the tomographic problem is split into a recovery problem and a support identification problem. This contrasts with the situation in traditional compressed sensing, where support identification and recovery are carried out in an integrated fashion. Can such an integrated algorithm also be given in the quantum setting? Are there potential advantages to such an approach?

\subsection*{Tomographic tasks}

A reoccurring theme in the research field of quantum state tomography has been to simplify the task by focusing it on the most important properties of a state. 
Central to the success of \emph{shadow tomography} \cite{aaronson2007learnability,huang2020predicting} is the insight that most of the time, one is not interested in knowing all observables, but only a finite collection of them, which turns out the make a significant difference \cite{aaronson2007learnability}.

While we already include the different restrictions on the tomographic task in our figure of merit, we note that requiring that predictions can be made in the whole parameter space $\calX$ might be excessive in some applications. This could be because the parametrization is too complex and eludes an efficient description or because only small parts of the parameter space $\calX$ are relevant, for examples when performing an optimization step in a quantum neural network.
Here, it is interesting to consider relaxations of the tomographic task for parametrized quantum states, \emph{i.e.}, a tomographic procedures with performance guarantees for a small environment around a fixed point in parameter space $\calX$.

Going in the opposite direction of relaxations, the protocol we give for tomography of parametrized quantum states solves the tomographic task relative to the norm $\tvert{\cdot}_{\calO,p}$, where $p=2$. Open questions and conceivably
harder are the cases $p=1$ and $p=\infty$, where the latter would ensure $\epsilon$-closeness to the original state for all points in $\calX$, not in an average sense, but even in the worst case. Here, we expect that schemes efficient in the number of basis functions cannot be given without further restrictions on the parametrized quantum 
state. 

\section{Conclusions}\label{sec:conclusion}

Parametrized quantum states are of fundamental interest in quantum theory. They occur in many different areas, including quantum optimization and machine learning, dynamical quantum systems, quantum metrology and quantum many body physics. Characterizing such objects well is of crucial importance. For quantum states without parameter dependence, this task is broadly known as \emph{tomography}. But what is its analogue for parametrized quantum states? 

We address this question by first defining the notion of a tomographic procedure, which unifies existing approaches to quantum state tomography in a single framework. This is done by expressing the figures of merit of different tomography tasks in a systematic fashion: as a semi-norm induced by the set of observables targeted. We show that this framework can be naturally combined with function $L^p$ norms such that it extends to parametrized quantum states. This way, we can rigorously define the task of performing tomography for parametrized quantum states.  

With this in place, we provide a first class of learning algorithms for this new task. Our algorithm combines a tomographic procedure for quantum states with a recovery algorithm from signal processing.

We show that by using techniques from compressed sensing as recovery algorithms, we can exploit \emph{structure} in the form of sparsity of the parameter dependence to significantly lower the
number of quantum states that need to be prepared. Overall, the resulting scheme is efficient if two conditions hold: The tomographic procedure chosen as input is efficient and the parametrization admits an efficient description in terms of an orthonormal basis. 

In quantum information science, we further encounter parametrized evolutions, \emph{i.e.},  parametrized quantum channels. This happens for example in the context of quantum metrology, where knowledge of the dependence of the evolution on the parameter to be estimated is crucial to devise protocols or when error channels in quantum devices depend on adjustable parameters of the system.
Analogous to the case of parametrized quantum states, one may extend notions of quantum process tomography, and define the task of tomography for parametrized quantum channels. By substituting a tomographic procedure with a protocol for process tomography as subroutine of our algorithm, we obtain a tomography scheme that applies to parametrized quantum channels. Again, the resulting scheme is efficient if both the process tomography protocol chosen as input and the description of the parametrization is efficient.

To showcase the techniques introduced in this manuscript, we apply them to quantum states with sub-Gaussian energy spectrum that are parametrized by the time evolution under an NMR Hamiltonian. We combine a tomographic procedure based on local Clifford shadows with a compressed sensing algorithm to obtain a classical representation that recovers all $\ell$-local reduced density matrices from the original parametrized quantum state.

There are many open questions and novel research directions associated to the tomography of parametrized quantum states. At the core, there is a need for further understanding of these states and their structural properties. To move to applications of practical interest, it is central to find orthonormal basis functions that allow for an efficient decomposition of the parametrized quantum states in question. In quantum state tomography, shadow tomography caused a change in paradigm by arguing for an alternative figure of merit, giving unexpected insights into the structure of quantum states. In the same way, it is interesting to investigate novel figures of merit for tomography of parametrized quantum states. These might be tailored with certain practical applications in mind or be motivated by purely theoretical considerations. We give a detailed survey of these and further open questions.

To summarize, tomography of parametrized quantum states is a new paradigm in quantum information theory. Derived from quantum state tomography, it is a task of fundamental interest in its own right. Furthermore, it is motivated by the need to characterize the class of parametrized quantum states, whose importance is only pronounced by the advent of ever more complex quantum devices. Here, we give a first treatment of this novel tomographic task, including an efficient class of algorithms to address it. We hope that our work stimulates further 
compressed sensing based methods for 
tomography of parametrized quantum states.

\section{Acknowledgments}
The authors wish to thank Matthias Caro and Lennart Bittel for insightful discussions.
We also acknowledge early discussions with Siriu Lu. This work has been supported by the DFG (CRC 183, FOR 2724), by the BMBF (Hybrid), the BMWK (PlanQK, EniQmA), the Munich Quantum Valley (K-8), QuantERA (HQCC) and the Einstein Foundation (Einstein Research Unit on Quantum Devices). This work has also been funded by the DFG under Germany's Excellence Strategy – The Berlin Mathematics Research Center MATH+ (EXC-2046/1, project ID: 390685689).
\vspace{1em}

\section{Author contributions}
F.S.\ has led the project. 
The project has been initiated by J.E.
F.S.\ and J.J.M.\ have developed the algorithms and conducted their theoretical analysis. J.J.M.\ has supervised the project with support of J.E.
All authors contributed to 
the development of meaningful examples 
and the writing of the manuscript.

\bibliography{main}

\begin{thebibliography}{50}%
\makeatletter
\providecommand \@ifxundefined [1]{%
 \@ifx{#1\undefined}
}%
\providecommand \@ifnum [1]{%
 \ifnum #1\expandafter \@firstoftwo
 \else \expandafter \@secondoftwo
 \fi
}%
\providecommand \@ifx [1]{%
 \ifx #1\expandafter \@firstoftwo
 \else \expandafter \@secondoftwo
 \fi
}%
\providecommand \natexlab [1]{#1}%
\providecommand \enquote  [1]{``#1''}%
\providecommand \bibnamefont  [1]{#1}%
\providecommand \bibfnamefont [1]{#1}%
\providecommand \citenamefont [1]{#1}%
\providecommand \href@noop [0]{\@secondoftwo}%
\providecommand \href [0]{\begingroup \@sanitize@url \@href}%
\providecommand \@href[1]{\@@startlink{#1}\@@href}%
\providecommand \@@href[1]{\endgroup#1\@@endlink}%
\providecommand \@sanitize@url [0]{\catcode `\\12\catcode `\$12\catcode
  `\&12\catcode `\#12\catcode `\^12\catcode `\_12\catcode `\%12\relax}%
\providecommand \@@startlink[1]{}%
\providecommand \@@endlink[0]{}%
\providecommand \url  [0]{\begingroup\@sanitize@url \@url }%
\providecommand \@url [1]{\endgroup\@href {#1}{\urlprefix }}%
\providecommand \urlprefix  [0]{URL }%
\providecommand \Eprint [0]{\href }%
\providecommand \doibase [0]{https://doi.org/}%
\providecommand \selectlanguage [0]{\@gobble}%
\providecommand \bibinfo  [0]{\@secondoftwo}%
\providecommand \bibfield  [0]{\@secondoftwo}%
\providecommand \translation [1]{[#1]}%
\providecommand \BibitemOpen [0]{}%
\providecommand \bibitemStop [0]{}%
\providecommand \bibitemNoStop [0]{.\EOS\space}%
\providecommand \EOS [0]{\spacefactor3000\relax}%
\providecommand \BibitemShut  [1]{\csname bibitem#1\endcsname}%
\let\auto@bib@innerbib\@empty
\bibitem [{\citenamefont {Eisert}\ \emph {et~al.}(2020)\citenamefont {Eisert},
  \citenamefont {Hangleiter}, \citenamefont {Walk}, \citenamefont {Roth},
  \citenamefont {Markham}, \citenamefont {Parekh}, \citenamefont {Chabaud},\
  and\ \citenamefont {Kashefi}}]{BenchmarkingReview}%
  \BibitemOpen
  \bibfield  {author} {\bibinfo {author} {\bibfnamefont {J.}~\bibnamefont
  {Eisert}}, \bibinfo {author} {\bibfnamefont {D.}~\bibnamefont {Hangleiter}},
  \bibinfo {author} {\bibfnamefont {N.}~\bibnamefont {Walk}}, \bibinfo {author}
  {\bibfnamefont {I.}~\bibnamefont {Roth}}, \bibinfo {author} {\bibfnamefont
  {D.}~\bibnamefont {Markham}}, \bibinfo {author} {\bibfnamefont
  {R.}~\bibnamefont {Parekh}}, \bibinfo {author} {\bibfnamefont
  {U.}~\bibnamefont {Chabaud}},\ and\ \bibinfo {author} {\bibfnamefont
  {E.}~\bibnamefont {Kashefi}},\ }\bibfield  {title} {\bibinfo {title} {Quantum
  certification and benchmarking},\ }\href
  {https://doi.org/10.1038/s42254-020-0186-4} {\bibfield  {journal} {\bibinfo
  {journal} {Nature Rev. Phys.}\ }\textbf {\bibinfo {volume} {2}},\ \bibinfo
  {pages} {382} (\bibinfo {year} {2020})}\BibitemShut {NoStop}%
\bibitem [{\citenamefont {Kliesch}\ and\ \citenamefont
  {Roth}(2021)}]{Certification}%
  \BibitemOpen
  \bibfield  {author} {\bibinfo {author} {\bibfnamefont {M.}~\bibnamefont
  {Kliesch}}\ and\ \bibinfo {author} {\bibfnamefont {I.}~\bibnamefont {Roth}},\
  }\bibfield  {title} {\bibinfo {title} {Theory of quantum system
  certification: a tutorial},\ }\href
  {https://doi.org/10.1103/PRXQuantum.2.010201} {\bibfield  {journal} {\bibinfo
   {journal} {PRX Quantum}\ }\textbf {\bibinfo {volume} {2}},\ \bibinfo {pages}
  {010201} (\bibinfo {year} {2021})}\BibitemShut {NoStop}%
\bibitem [{\citenamefont {Chan}\ \emph {et~al.}(2023)\citenamefont {Chan},
  \citenamefont {Meister}, \citenamefont {Goh},\ and\ \citenamefont
  {Koczor}}]{chan2023algorithmic}%
  \BibitemOpen
  \bibfield  {author} {\bibinfo {author} {\bibfnamefont {H.~H.~S.}\
  \bibnamefont {Chan}}, \bibinfo {author} {\bibfnamefont {R.}~\bibnamefont
  {Meister}}, \bibinfo {author} {\bibfnamefont {M.~L.}\ \bibnamefont {Goh}},\
  and\ \bibinfo {author} {\bibfnamefont {B.}~\bibnamefont {Koczor}},\
  }\href@noop {} {\bibinfo {title} {Algorithmic shadow spectroscopy}} (\bibinfo
  {year} {2023}),\ \Eprint {https://arxiv.org/abs/2212.11036}
  {arXiv:2212.11036} \BibitemShut {NoStop}%
\bibitem [{\citenamefont {Zhan}\ \emph {et~al.}(2024)\citenamefont {Zhan},
  \citenamefont {Elben}, \citenamefont {Huang},\ and\ \citenamefont
  {Tong}}]{zhan2024learning}%
  \BibitemOpen
  \bibfield  {author} {\bibinfo {author} {\bibfnamefont {Y.}~\bibnamefont
  {Zhan}}, \bibinfo {author} {\bibfnamefont {A.}~\bibnamefont {Elben}},
  \bibinfo {author} {\bibfnamefont {H.-Y.}\ \bibnamefont {Huang}},\ and\
  \bibinfo {author} {\bibfnamefont {Y.}~\bibnamefont {Tong}},\ }\bibfield
  {title} {\bibinfo {title} {Learning conservation laws in unknown quantum
  dynamics},\ }\href {https://doi.org/10.1103/PRXQuantum.5.010350} {\bibfield
  {journal} {\bibinfo  {journal} {PRX Quantum}\ }\textbf {\bibinfo {volume}
  {5}},\ \bibinfo {pages} {010350} (\bibinfo {year} {2024})}\BibitemShut
  {NoStop}%
\bibitem [{\citenamefont {Shivam}\ \emph {et~al.}(2023)\citenamefont {Shivam},
  \citenamefont {von Keyserlingk},\ and\ \citenamefont
  {Sondhi}}]{shivam2023classical}%
  \BibitemOpen
  \bibfield  {author} {\bibinfo {author} {\bibfnamefont {S.}~\bibnamefont
  {Shivam}}, \bibinfo {author} {\bibfnamefont {C.~W.}\ \bibnamefont {von
  Keyserlingk}},\ and\ \bibinfo {author} {\bibfnamefont {S.~L.}\ \bibnamefont
  {Sondhi}},\ }\bibfield  {title} {\bibinfo {title} {On classical and hybrid
  shadows of quantum states},\ }\href
  {https://doi.org/10.21468/SciPostPhys.14.5.094} {\bibfield  {journal}
  {\bibinfo  {journal} {SciPost Physics}\ }\textbf {\bibinfo {volume} {14}},\
  \bibinfo {pages} {094} (\bibinfo {year} {2023})}\BibitemShut {NoStop}%
\bibitem [{\citenamefont {Huang}\ \emph {et~al.}(2022)\citenamefont {Huang},
  \citenamefont {Kueng}, \citenamefont {Torlai}, \citenamefont {Albert},\ and\
  \citenamefont {Preskill}}]{huang2021provably}%
  \BibitemOpen
  \bibfield  {author} {\bibinfo {author} {\bibfnamefont {H.-Y.}\ \bibnamefont
  {Huang}}, \bibinfo {author} {\bibfnamefont {R.}~\bibnamefont {Kueng}},
  \bibinfo {author} {\bibfnamefont {G.}~\bibnamefont {Torlai}}, \bibinfo
  {author} {\bibfnamefont {V.~V.}\ \bibnamefont {Albert}},\ and\ \bibinfo
  {author} {\bibfnamefont {J.}~\bibnamefont {Preskill}},\ }\bibfield  {title}
  {\bibinfo {title} {Provably efficient machine learning for quantum many-body
  problems},\ }\href {https://doi.org/10.1126/science.abk3333} {\bibfield
  {journal} {\bibinfo  {journal} {Science}\ }\textbf {\bibinfo {volume}
  {377}},\ \bibinfo {pages} {eabk3333} (\bibinfo {year} {2022})}\BibitemShut
  {NoStop}%
\bibitem [{\citenamefont {Lewis}\ \emph {et~al.}(2024)\citenamefont {Lewis},
  \citenamefont {Huang}, \citenamefont {Tran}, \citenamefont {Lehner},
  \citenamefont {Kueng},\ and\ \citenamefont {Preskill}}]{lewis2024improved}%
  \BibitemOpen
  \bibfield  {author} {\bibinfo {author} {\bibfnamefont {L.}~\bibnamefont
  {Lewis}}, \bibinfo {author} {\bibfnamefont {H.-Y.}\ \bibnamefont {Huang}},
  \bibinfo {author} {\bibfnamefont {V.~T.}\ \bibnamefont {Tran}}, \bibinfo
  {author} {\bibfnamefont {S.}~\bibnamefont {Lehner}}, \bibinfo {author}
  {\bibfnamefont {R.}~\bibnamefont {Kueng}},\ and\ \bibinfo {author}
  {\bibfnamefont {J.}~\bibnamefont {Preskill}},\ }\bibfield  {title} {\bibinfo
  {title} {Improved machine learning algorithm for predicting ground state
  properties},\ }\href {https://doi.org/10.1038/s41467-024-45014-7} {\bibfield
  {journal} {\bibinfo  {journal} {Nature Comm.}\ }\textbf {\bibinfo {volume}
  {15}},\ \bibinfo {pages} {895} (\bibinfo {year} {2024})}\BibitemShut
  {NoStop}%
\bibitem [{\citenamefont {Rudin}(1987)}]{rudin1987analysis}%
  \BibitemOpen
  \bibfield  {author} {\bibinfo {author} {\bibfnamefont {W.}~\bibnamefont
  {Rudin}},\ }\href {https://doi.org/10.1007/978-981-13-0938-0} {\emph
  {\bibinfo {title} {Real and complex analysis}}}\ (\bibinfo  {publisher}
  {McGraw-Hill Book Company},\ \bibinfo {year} {1987})\BibitemShut {NoStop}%
\bibitem [{\citenamefont {O'Donnell}\ and\ \citenamefont
  {Wright}(2016)}]{ODonnel2016efficient}%
  \BibitemOpen
  \bibfield  {author} {\bibinfo {author} {\bibfnamefont {R.}~\bibnamefont
  {O'Donnell}}\ and\ \bibinfo {author} {\bibfnamefont {J.}~\bibnamefont
  {Wright}},\ }\bibfield  {title} {\bibinfo {title} {Efficient quantum
  tomography},\ }in\ \href@noop {} {\emph {\bibinfo {booktitle} {Proc. 48th
  Ann. ACM Symp. Th. Comp.}}}\ (\bibinfo {year} {2016})\ pp.\ \bibinfo {pages}
  {899--912}\BibitemShut {NoStop}%
\bibitem [{\citenamefont {Aaronson}()}]{aaronson2018shadow}%
  \BibitemOpen
  \bibfield  {author} {\bibinfo {author} {\bibfnamefont {S.}~\bibnamefont
  {Aaronson}},\ }\bibfield  {title} {\bibinfo {title} {Shadow tomography of
  quantum states},\ }in\ \href {https://doi.org/10.1145/3188745.3188802} {\emph
  {\bibinfo {booktitle} {Proc. 50th Ann. {ACM} {SIGACT} Symp. Th. Comp.}}},\
  \bibinfo {series and number} {{STOC} 2018}\ (\bibinfo  {publisher}
  {Association for Computing Machinery})\ pp.\ \bibinfo {pages}
  {325--338}\BibitemShut {NoStop}%
\bibitem [{\citenamefont {Huang}\ \emph {et~al.}(2020)\citenamefont {Huang},
  \citenamefont {Kueng},\ and\ \citenamefont {Preskill}}]{huang2020predicting}%
  \BibitemOpen
  \bibfield  {author} {\bibinfo {author} {\bibfnamefont {H.-Y.}\ \bibnamefont
  {Huang}}, \bibinfo {author} {\bibfnamefont {R.}~\bibnamefont {Kueng}},\ and\
  \bibinfo {author} {\bibfnamefont {J.}~\bibnamefont {Preskill}},\ }\bibfield
  {title} {\bibinfo {title} {Predicting many properties of a quantum system
  from very few measurements},\ }\href
  {https://doi.org/10.1038/s41567-020-0932-7} {\bibfield  {journal} {\bibinfo
  {journal} {Nature Phys.}\ }\textbf {\bibinfo {volume} {16}},\ \bibinfo
  {pages} {1050} (\bibinfo {year} {2020})}\BibitemShut {NoStop}%
\bibitem [{\citenamefont {Aaronson}(2007)}]{aaronson2007learnability}%
  \BibitemOpen
  \bibfield  {author} {\bibinfo {author} {\bibfnamefont {S.}~\bibnamefont
  {Aaronson}},\ }\bibfield  {title} {\bibinfo {title} {The learnability of
  quantum states},\ }\href {https://doi.org/10.1098/rspa.2007.0113} {\bibfield
  {journal} {\bibinfo  {journal} {Proc. Roy. Soc. A}\ }\textbf {\bibinfo
  {volume} {463}},\ \bibinfo {pages} {3089} (\bibinfo {year}
  {2007})}\BibitemShut {NoStop}%
\bibitem [{\citenamefont {De~Palma}\ \emph {et~al.}(2021)\citenamefont
  {De~Palma}, \citenamefont {Marvian}, \citenamefont {Trevisan},\ and\
  \citenamefont {Lloyd}}]{de_palma2021quantum}%
  \BibitemOpen
  \bibfield  {author} {\bibinfo {author} {\bibfnamefont {G.}~\bibnamefont
  {De~Palma}}, \bibinfo {author} {\bibfnamefont {M.}~\bibnamefont {Marvian}},
  \bibinfo {author} {\bibfnamefont {D.}~\bibnamefont {Trevisan}},\ and\
  \bibinfo {author} {\bibfnamefont {S.}~\bibnamefont {Lloyd}},\ }\bibfield
  {title} {\bibinfo {title} {{The quantum Wasserstein distance of order 1}},\
  }\href {https://doi.org/10.1109/TIT.2021.3076442} {\bibfield  {journal}
  {\bibinfo  {journal} {IEEE Trans. Inf. Th.}\ }\textbf {\bibinfo {volume}
  {67}},\ \bibinfo {pages} {6627} (\bibinfo {year} {2021})}\BibitemShut
  {NoStop}%
\bibitem [{\citenamefont {Boche}\ \emph {et~al.}()\citenamefont {Boche},
  \citenamefont {Calderbank}, \citenamefont {Kutyniok},\ and\ \citenamefont
  {Vybíral}}]{boche2015survey}%
  \BibitemOpen
  \bibfield  {author} {\bibinfo {author} {\bibfnamefont {H.}~\bibnamefont
  {Boche}}, \bibinfo {author} {\bibfnamefont {R.}~\bibnamefont {Calderbank}},
  \bibinfo {author} {\bibfnamefont {G.}~\bibnamefont {Kutyniok}},\ and\
  \bibinfo {author} {\bibfnamefont {J.}~\bibnamefont {Vybíral}},\ }\bibfield
  {title} {\bibinfo {title} {A survey of compressed sensing},\ }in\ \href
  {https://doi.org/10.1007/978-3-319-16042-9_1} {\emph {\bibinfo {booktitle}
  {Compressed sensing and its applications}}},\ \bibinfo {editor} {edited by\
  \bibinfo {editor} {\bibfnamefont {H.}~\bibnamefont {Boche}}, \bibinfo
  {editor} {\bibfnamefont {R.}~\bibnamefont {Calderbank}}, \bibinfo {editor}
  {\bibfnamefont {G.}~\bibnamefont {Kutyniok}},\ and\ \bibinfo {editor}
  {\bibfnamefont {J.}~\bibnamefont {Vybíral}}}\ (\bibinfo  {publisher}
  {Springer International Publishing})\ pp.\ \bibinfo {pages}
  {1--39}\BibitemShut {NoStop}%
\bibitem [{\citenamefont {Eldar}\ and\ \citenamefont
  {Kutyniok}(2012)}]{CompressedSensingGitta}%
  \BibitemOpen
  \bibfield  {author} {\bibinfo {author} {\bibfnamefont {Y.~C.}\ \bibnamefont
  {Eldar}}\ and\ \bibinfo {author} {\bibfnamefont {G.}~\bibnamefont
  {Kutyniok}},\ }\href {https://doi.org/10.1017/CBO9780511794308} {\emph
  {\bibinfo {title} {Compressed sensing: theory and applications}}}\ (\bibinfo
  {publisher} {Cambridge University Press},\ \bibinfo {year}
  {2012})\BibitemShut {NoStop}%
\bibitem [{\citenamefont {Candès}\ \emph {et~al.}(2011)\citenamefont
  {Candès}, \citenamefont {Eldar}, \citenamefont {Needell},\ and\
  \citenamefont {Randall}}]{candes2011compressed}%
  \BibitemOpen
  \bibfield  {author} {\bibinfo {author} {\bibfnamefont {E.~J.}\ \bibnamefont
  {Candès}}, \bibinfo {author} {\bibfnamefont {Y.~C.}\ \bibnamefont {Eldar}},
  \bibinfo {author} {\bibfnamefont {D.}~\bibnamefont {Needell}},\ and\ \bibinfo
  {author} {\bibfnamefont {P.}~\bibnamefont {Randall}},\ }\bibfield  {title}
  {\bibinfo {title} {Compressed sensing with coherent and redundant
  dictionaries},\ }\href {https://doi.org/10.1016/j.acha.2010.10.002}
  {\bibfield  {journal} {\bibinfo  {journal} {Applied Comp. Harm. Ana.}\
  }\textbf {\bibinfo {volume} {31}},\ \bibinfo {pages} {59} (\bibinfo {year}
  {2011})}\BibitemShut {NoStop}%
\bibitem [{\citenamefont {Eisert}\ \emph {et~al.}(2022)\citenamefont {Eisert},
  \citenamefont {Flinth}, \citenamefont {Gro{\ss}}, \citenamefont {Roth},\ and\
  \citenamefont {Wunder}}]{eisert2021hierarchical}%
  \BibitemOpen
  \bibfield  {author} {\bibinfo {author} {\bibfnamefont {J.}~\bibnamefont
  {Eisert}}, \bibinfo {author} {\bibfnamefont {A.}~\bibnamefont {Flinth}},
  \bibinfo {author} {\bibfnamefont {B.}~\bibnamefont {Gro{\ss}}}, \bibinfo
  {author} {\bibfnamefont {I.}~\bibnamefont {Roth}},\ and\ \bibinfo {author}
  {\bibfnamefont {G.}~\bibnamefont {Wunder}},\ }\bibinfo {title} {Hierarchical
  compressed sensing},\ in\ \href {https://doi.org/10.1007/978-3-031-09745-4_1}
  {\emph {\bibinfo {booktitle} {Compressed sensing in information
  processing}}},\ \bibinfo {editor} {edited by\ \bibinfo {editor}
  {\bibfnamefont {G.}~\bibnamefont {Kutyniok}}, \bibinfo {editor}
  {\bibfnamefont {H.}~\bibnamefont {Rauhut}},\ and\ \bibinfo {editor}
  {\bibfnamefont {R.~J.}\ \bibnamefont {Kunsch}}}\ (\bibinfo  {publisher}
  {Springer International Publishing},\ \bibinfo {address} {Cham},\ \bibinfo
  {year} {2022})\ pp.\ \bibinfo {pages} {1--35}\BibitemShut {NoStop}%
\bibitem [{\citenamefont {Candes}\ and\ \citenamefont
  {Recht}(2008)}]{CandesRecht}%
  \BibitemOpen
  \bibfield  {author} {\bibinfo {author} {\bibfnamefont {E.~J.}\ \bibnamefont
  {Candes}}\ and\ \bibinfo {author} {\bibfnamefont {B.}~\bibnamefont {Recht}},\
  }\href {https://doi.org/10.1007/s10208-009-9045-5} {\bibfield  {journal}
  {\bibinfo  {journal} {Found. Comp. Math.}\ }\textbf {\bibinfo {volume} {9}},\
  \bibinfo {pages} {717} (\bibinfo {year} {2008})}\BibitemShut {NoStop}%
\bibitem [{\citenamefont {Gross}\ \emph {et~al.}(2020)\citenamefont {Gross},
  \citenamefont {Liu}, \citenamefont {Flammia}, \citenamefont {Becker},\ and\
  \citenamefont {Eisert}}]{gross_quantum_2010}%
  \BibitemOpen
  \bibfield  {author} {\bibinfo {author} {\bibfnamefont {D.}~\bibnamefont
  {Gross}}, \bibinfo {author} {\bibfnamefont {Y.-K.}\ \bibnamefont {Liu}},
  \bibinfo {author} {\bibfnamefont {S.~T.}\ \bibnamefont {Flammia}}, \bibinfo
  {author} {\bibfnamefont {S.}~\bibnamefont {Becker}},\ and\ \bibinfo {author}
  {\bibfnamefont {J.}~\bibnamefont {Eisert}},\ }\bibfield  {title} {\bibinfo
  {title} {Quantum state tomography via compressed sensing},\ }\href
  {https://doi.org/10.1103/PhysRevLett.105.150401} {\bibfield  {journal}
  {\bibinfo  {journal} {Phys. Rev. Lett.}\ }\textbf {\bibinfo {volume} {105}},\
  \bibinfo {pages} {150401} (\bibinfo {year} {2020})}\BibitemShut {NoStop}%
\bibitem [{\citenamefont {Cohen}\ \emph {et~al.}(2009)\citenamefont {Cohen},
  \citenamefont {Dahmen},\ and\ \citenamefont
  {{DeVore}}}]{cohen2009compressed}%
  \BibitemOpen
  \bibfield  {author} {\bibinfo {author} {\bibfnamefont {A.}~\bibnamefont
  {Cohen}}, \bibinfo {author} {\bibfnamefont {W.}~\bibnamefont {Dahmen}},\ and\
  \bibinfo {author} {\bibfnamefont {R.}~\bibnamefont {{DeVore}}},\ }\bibfield
  {title} {\bibinfo {title} {Compressed sensing and best k-term
  approximation},\ }\href {https://doi.org/10.1090/S0894-0347-08-00610-3}
  {\bibfield  {journal} {\bibinfo  {journal} {J. Am. Math. Soc.}\ }\textbf
  {\bibinfo {volume} {22}},\ \bibinfo {pages} {211} (\bibinfo {year}
  {2009})}\BibitemShut {NoStop}%
\bibitem [{\citenamefont {Foucart}\ and\ \citenamefont
  {Rauhut}()}]{foucart2013mathematical}%
  \BibitemOpen
  \bibfield  {author} {\bibinfo {author} {\bibfnamefont {S.}~\bibnamefont
  {Foucart}}\ and\ \bibinfo {author} {\bibfnamefont {H.}~\bibnamefont
  {Rauhut}},\ }\href {https://doi.org/10.1007/978-0-8176-4948-7} {\emph
  {\bibinfo {title} {A mathematical introduction to compressive sensing}}},\
  Applied and numerical harmonic analysis\ (\bibinfo  {publisher} {Springer},\
  \bibinfo {address} {New York})\BibitemShut {NoStop}%
\bibitem [{\citenamefont {Rauhut}()}]{rauhut2010compressive}%
  \BibitemOpen
  \bibfield  {author} {\bibinfo {author} {\bibfnamefont {H.}~\bibnamefont
  {Rauhut}},\ }\bibfield  {title} {\bibinfo {title} {Compressive sensing and
  structured random matrices},\ }in\ \href
  {https://doi.org/10.1515/9783110226157.1} {\emph {\bibinfo {booktitle}
  {Theoretical foundations and numerical methods for sparse recovery}}},\
  \bibinfo {editor} {edited by\ \bibinfo {editor} {\bibfnamefont
  {M.}~\bibnamefont {Fornasier}}}\ (\bibinfo  {publisher} {{De} {Gruyter}})\
  pp.\ \bibinfo {pages} {1--92}\BibitemShut {NoStop}%
\bibitem [{\citenamefont {Foucart}(2011)}]{foucart2011hard}%
  \BibitemOpen
  \bibfield  {author} {\bibinfo {author} {\bibfnamefont {S.}~\bibnamefont
  {Foucart}},\ }\bibfield  {title} {\bibinfo {title} {Hard thresholding
  pursuit: An algorithm for compressive sensing},\ }\href
  {https://doi.org/10.1137/100806278} {\bibfield  {journal} {\bibinfo
  {journal} {SIAM J. Num. Ana.}\ }\textbf {\bibinfo {volume} {49}},\ \bibinfo
  {pages} {2543} (\bibinfo {year} {2011})}\BibitemShut {NoStop}%
\bibitem [{\citenamefont {Lee}\ \emph {et~al.}(2013)\citenamefont {Lee},
  \citenamefont {Bresler},\ and\ \citenamefont {Junge}}]{lee2013oblique}%
  \BibitemOpen
  \bibfield  {author} {\bibinfo {author} {\bibfnamefont {K.}~\bibnamefont
  {Lee}}, \bibinfo {author} {\bibfnamefont {Y.}~\bibnamefont {Bresler}},\ and\
  \bibinfo {author} {\bibfnamefont {M.}~\bibnamefont {Junge}},\ }\bibfield
  {title} {\bibinfo {title} {Oblique pursuits for compressed sensing},\ }\href
  {https://doi.org/10.1109/TIT.2013.2254172} {\bibfield  {journal} {\bibinfo
  {journal} {IEEE Trans. Inf. Th.}\ }\textbf {\bibinfo {volume} {59}},\
  \bibinfo {pages} {6111} (\bibinfo {year} {2013})}\BibitemShut {NoStop}%
\bibitem [{\citenamefont {Bandeira}\ \emph {et~al.}(2013)\citenamefont
  {Bandeira}, \citenamefont {Fickus}, \citenamefont {Mixon},\ and\
  \citenamefont {Wong}}]{bandeira2013road}%
  \BibitemOpen
  \bibfield  {author} {\bibinfo {author} {\bibfnamefont {A.~S.}\ \bibnamefont
  {Bandeira}}, \bibinfo {author} {\bibfnamefont {M.}~\bibnamefont {Fickus}},
  \bibinfo {author} {\bibfnamefont {D.~G.}\ \bibnamefont {Mixon}},\ and\
  \bibinfo {author} {\bibfnamefont {P.}~\bibnamefont {Wong}},\ }\bibfield
  {title} {\bibinfo {title} {The road to deterministic matrices with the
  restricted isometry property},\ }\href
  {https://doi.org//10.1007/s00041-013-9293-2} {\bibfield  {journal} {\bibinfo
  {journal} {J. Four. An. Appl.}\ }\textbf {\bibinfo {volume} {19}},\ \bibinfo
  {pages} {1123} (\bibinfo {year} {2013})}\BibitemShut {NoStop}%
\bibitem [{\citenamefont {Bourgain}\ \emph {et~al.}(2011)\citenamefont
  {Bourgain}, \citenamefont {Dilworth}, \citenamefont {Ford}, \citenamefont
  {Konyagin},\ and\ \citenamefont {Kutzarova}}]{Bourgain_2011}%
  \BibitemOpen
  \bibfield  {author} {\bibinfo {author} {\bibfnamefont {J.}~\bibnamefont
  {Bourgain}}, \bibinfo {author} {\bibfnamefont {S.}~\bibnamefont {Dilworth}},
  \bibinfo {author} {\bibfnamefont {K.}~\bibnamefont {Ford}}, \bibinfo {author}
  {\bibfnamefont {S.}~\bibnamefont {Konyagin}},\ and\ \bibinfo {author}
  {\bibfnamefont {D.}~\bibnamefont {Kutzarova}},\ }\bibfield  {title} {\bibinfo
  {title} {{Explicit constructions of RIP matrices and related problems}},\
  }\href {https://doi.org/10.1215/00127094-1384809} {\bibfield  {journal}
  {\bibinfo  {journal} {Duke Math. J.}\ }\textbf {\bibinfo {volume} {159}},\
  \bibinfo {pages} {145} (\bibinfo {year} {2011})}\BibitemShut {NoStop}%
\bibitem [{\citenamefont {Arian}\ and\ \citenamefont
  {Yilmaz}()}]{Deterministic}%
  \BibitemOpen
  \bibfield  {author} {\bibinfo {author} {\bibfnamefont {A.}~\bibnamefont
  {Arian}}\ and\ \bibinfo {author} {\bibfnamefont {O.}~\bibnamefont {Yilmaz}},\
  }\href@noop {} {\bibinfo {title} {{RIP constants for deterministic compressed
  sensing matrices-beyond Gershgorin}}},\ \Eprint
  {https://arxiv.org/abs/1911.07428} {arXiv:1911.07428} \BibitemShut {NoStop}%
\bibitem [{\citenamefont {Blumensath}\ and\ \citenamefont
  {Davies}(2009)}]{blumensath2009iterative}%
  \BibitemOpen
  \bibfield  {author} {\bibinfo {author} {\bibfnamefont {T.}~\bibnamefont
  {Blumensath}}\ and\ \bibinfo {author} {\bibfnamefont {M.~E.}\ \bibnamefont
  {Davies}},\ }\bibfield  {title} {\bibinfo {title} {Iterative hard
  thresholding for compressed sensing},\ }\href
  {https://doi.org/10.1016/j.acha.2009.04.002} {\bibfield  {journal} {\bibinfo
  {journal} {{Applied Comp. Harm. Ana.}}\ }\textbf {\bibinfo {volume} {27}},\
  \bibinfo {pages} {265} (\bibinfo {year} {2009})}\BibitemShut {NoStop}%
\bibitem [{\citenamefont {Guo}\ and\ \citenamefont
  {Rebentrost}(2023)}]{guo2023estimating}%
  \BibitemOpen
  \bibfield  {author} {\bibinfo {author} {\bibfnamefont {N.}~\bibnamefont
  {Guo}}\ and\ \bibinfo {author} {\bibfnamefont {P.}~\bibnamefont
  {Rebentrost}},\ }\href@noop {} {\bibinfo {title} {Estimating properties of a
  quantum state by importance-sampled operator shadows}} (\bibinfo {year}
  {2023}),\ \Eprint {https://arxiv.org/abs/2305.09374} {arXiv:2305.09374}
  \BibitemShut {NoStop}%
\bibitem [{\citenamefont {Haupt}\ and\ \citenamefont
  {Baraniuk}(2011)}]{5766202}%
  \BibitemOpen
  \bibfield  {author} {\bibinfo {author} {\bibfnamefont {J.}~\bibnamefont
  {Haupt}}\ and\ \bibinfo {author} {\bibfnamefont {R.}~\bibnamefont
  {Baraniuk}},\ }\bibfield  {title} {\bibinfo {title} {Robust support recovery
  using sparse compressive sensing matrices},\ }in\ \href
  {https://doi.org/10.1109/CISS.2011.5766202} {\emph {\bibinfo {booktitle}
  {2011 45th Ann. Conf. Inf. Sc. Sys.}}}\ (\bibinfo {year} {2011})\ pp.\
  \bibinfo {pages} {1--6}\BibitemShut {NoStop}%
\bibitem [{\citenamefont {Smith}\ \emph {et~al.}(1992)\citenamefont {Smith},
  \citenamefont {Palke},\ and\ \citenamefont {Gerig}}]{smith1992hamiltonians}%
  \BibitemOpen
  \bibfield  {author} {\bibinfo {author} {\bibfnamefont {S.~A.}\ \bibnamefont
  {Smith}}, \bibinfo {author} {\bibfnamefont {W.~E.}\ \bibnamefont {Palke}},\
  and\ \bibinfo {author} {\bibfnamefont {J.}~\bibnamefont {Gerig}},\ }\bibfield
   {title} {\bibinfo {title} {{The Hamiltonians of NMR. Part I}},\ }\href
  {https://doi.org/10.1002/cmr.1820040202} {\bibfield  {journal} {\bibinfo
  {journal} {C. Mag. Res.}\ }\textbf {\bibinfo {volume} {4}},\ \bibinfo {pages}
  {107} (\bibinfo {year} {1992})}\BibitemShut {NoStop}%
\bibitem [{\citenamefont {Alhambra}\ \emph {et~al.}(2020)\citenamefont
  {Alhambra}, \citenamefont {Anshu},\ and\ \citenamefont
  {Wilming}}]{PhysRevB.101.205107}%
  \BibitemOpen
  \bibfield  {author} {\bibinfo {author} {\bibfnamefont {A.~M.}\ \bibnamefont
  {Alhambra}}, \bibinfo {author} {\bibfnamefont {A.}~\bibnamefont {Anshu}},\
  and\ \bibinfo {author} {\bibfnamefont {H.}~\bibnamefont {Wilming}},\
  }\bibfield  {title} {\bibinfo {title} {Revivals imply quantum many-body
  scars},\ }\href {https://doi.org/10.1103/PhysRevB.101.205107} {\bibfield
  {journal} {\bibinfo  {journal} {Phys. Rev. B}\ }\textbf {\bibinfo {volume}
  {101}},\ \bibinfo {pages} {205107} (\bibinfo {year} {2020})}\BibitemShut
  {NoStop}%
\bibitem [{\citenamefont {Bernien}\ \emph {et~al.}(2017)\citenamefont
  {Bernien}, \citenamefont {Schwartz}, \citenamefont {Keesling}, \citenamefont
  {Levine}, \citenamefont {Omran}, \citenamefont {Pichler}, \citenamefont
  {Choi}, \citenamefont {Zibrov}, \citenamefont {Endres}, \citenamefont
  {Greiner}, \citenamefont {Vuletic},\ and\ \citenamefont
  {Lukin}}]{GreinerSpeedup}%
  \BibitemOpen
  \bibfield  {author} {\bibinfo {author} {\bibfnamefont {H.}~\bibnamefont
  {Bernien}}, \bibinfo {author} {\bibfnamefont {S.}~\bibnamefont {Schwartz}},
  \bibinfo {author} {\bibfnamefont {A.}~\bibnamefont {Keesling}}, \bibinfo
  {author} {\bibfnamefont {H.}~\bibnamefont {Levine}}, \bibinfo {author}
  {\bibfnamefont {A.}~\bibnamefont {Omran}}, \bibinfo {author} {\bibfnamefont
  {H.}~\bibnamefont {Pichler}}, \bibinfo {author} {\bibfnamefont
  {S.}~\bibnamefont {Choi}}, \bibinfo {author} {\bibfnamefont {A.~S.}\
  \bibnamefont {Zibrov}}, \bibinfo {author} {\bibfnamefont {M.}~\bibnamefont
  {Endres}}, \bibinfo {author} {\bibfnamefont {M.}~\bibnamefont {Greiner}},
  \bibinfo {author} {\bibfnamefont {V.}~\bibnamefont {Vuletic}},\ and\ \bibinfo
  {author} {\bibfnamefont {M.}~\bibnamefont {Lukin}},\ }\bibfield  {title}
  {\bibinfo {title} {Probing many-body dynamics on a 51-atom quantum
  simulator},\ }\href {https://doi.org/10.1038/nature24622} {\bibfield
  {journal} {\bibinfo  {journal} {Nature}\ }\textbf {\bibinfo {volume} {551}},\
  \bibinfo {pages} {579} (\bibinfo {year} {2017})}\BibitemShut {NoStop}%
\bibitem [{\citenamefont {Shiraishi}\ and\ \citenamefont
  {Mori}(2017)}]{PhysRevLett.119.030601}%
  \BibitemOpen
  \bibfield  {author} {\bibinfo {author} {\bibfnamefont {N.}~\bibnamefont
  {Shiraishi}}\ and\ \bibinfo {author} {\bibfnamefont {T.}~\bibnamefont
  {Mori}},\ }\bibfield  {title} {\bibinfo {title} {Systematic construction of
  counterexamples to the eigenstate thermalization hypothesis},\ }\href
  {https://doi.org/10.1103/PhysRevLett.119.030601} {\bibfield  {journal}
  {\bibinfo  {journal} {Phys. Rev. Lett.}\ }\textbf {\bibinfo {volume} {119}},\
  \bibinfo {pages} {030601} (\bibinfo {year} {2017})}\BibitemShut {NoStop}%
\bibitem [{\citenamefont {Turner}\ \emph {et~al.}(2018)\citenamefont {Turner},
  \citenamefont {Michailidis}, \citenamefont {Abanin}, \citenamefont {Serbyn},\
  and\ \citenamefont {Papic}}]{Scars}%
  \BibitemOpen
  \bibfield  {author} {\bibinfo {author} {\bibfnamefont {C.~J.}\ \bibnamefont
  {Turner}}, \bibinfo {author} {\bibfnamefont {A.~A.}\ \bibnamefont
  {Michailidis}}, \bibinfo {author} {\bibfnamefont {D.~A.}\ \bibnamefont
  {Abanin}}, \bibinfo {author} {\bibfnamefont {M.}~\bibnamefont {Serbyn}},\
  and\ \bibinfo {author} {\bibfnamefont {Z.}~\bibnamefont {Papic}},\ }\bibfield
   {title} {\bibinfo {title} {Quantum many-body scars},\ }\href
  {https://doi.org/10.1038/s41567-018-0137-5} {\bibfield  {journal} {\bibinfo
  {journal} {Nature Phys.}\ }\textbf {\bibinfo {volume} {14}},\ \bibinfo
  {pages} {745} (\bibinfo {year} {2018})}\BibitemShut {NoStop}%
\bibitem [{\citenamefont {Surace}\ and\ \citenamefont
  {Tagliacozzo}(2022)}]{surace2022fermionic}%
  \BibitemOpen
  \bibfield  {author} {\bibinfo {author} {\bibfnamefont {J.}~\bibnamefont
  {Surace}}\ and\ \bibinfo {author} {\bibfnamefont {L.}~\bibnamefont
  {Tagliacozzo}},\ }\bibfield  {title} {\bibinfo {title} {{Fermionic Gaussian
  states: an introduction to numerical approaches}},\ }\href
  {https://doi.org/10.21468/SciPostPhysLectNotes.54} {\bibfield  {journal}
  {\bibinfo  {journal} {{SciPost} Physics Lecture Notes}\ ,\ \bibinfo {pages}
  {54}} (\bibinfo {year} {2022})}\BibitemShut {NoStop}%
\bibitem [{\citenamefont {Zhao}\ \emph {et~al.}(2021)\citenamefont {Zhao},
  \citenamefont {Rubin},\ and\ \citenamefont {Miyake}}]{zhao2021fermionic}%
  \BibitemOpen
  \bibfield  {author} {\bibinfo {author} {\bibfnamefont {A.}~\bibnamefont
  {Zhao}}, \bibinfo {author} {\bibfnamefont {N.~C.}\ \bibnamefont {Rubin}},\
  and\ \bibinfo {author} {\bibfnamefont {A.}~\bibnamefont {Miyake}},\
  }\bibfield  {title} {\bibinfo {title} {Fermionic partial tomography via
  classical shadows},\ }\href {https://doi.org/10.1103/PhysRevLett.127.110504}
  {\bibfield  {journal} {\bibinfo  {journal} {Phys. Rev. Lett.}\ }\textbf
  {\bibinfo {volume} {127}},\ \bibinfo {pages} {110504} (\bibinfo {year}
  {2021})}\BibitemShut {NoStop}%
\bibitem [{\citenamefont {Gutoski}\ and\ \citenamefont
  {Watrous}()}]{gutoski2007toward}%
  \BibitemOpen
  \bibfield  {author} {\bibinfo {author} {\bibfnamefont {G.}~\bibnamefont
  {Gutoski}}\ and\ \bibinfo {author} {\bibfnamefont {J.}~\bibnamefont
  {Watrous}},\ }\bibfield  {title} {\bibinfo {title} {Toward a general theory
  of quantum games},\ }in\ \href {https://doi.org/10.1145/1250790.1250873}
  {\emph {\bibinfo {booktitle} {Proc. 39th Ann. {ACM} Symp. Th. Comp.}}},\
  \bibinfo {series and number} {{STOC} '07}\ (\bibinfo  {publisher}
  {Association for Computing Machinery})\ pp.\ \bibinfo {pages}
  {565--574}\BibitemShut {NoStop}%
\bibitem [{\citenamefont {Chiribella}\ \emph {et~al.}(2009)\citenamefont
  {Chiribella}, \citenamefont {D’Ariano},\ and\ \citenamefont
  {Perinotti}}]{chiribella2009theoretical}%
  \BibitemOpen
  \bibfield  {author} {\bibinfo {author} {\bibfnamefont {G.}~\bibnamefont
  {Chiribella}}, \bibinfo {author} {\bibfnamefont {G.~M.}\ \bibnamefont
  {D’Ariano}},\ and\ \bibinfo {author} {\bibfnamefont {P.}~\bibnamefont
  {Perinotti}},\ }\bibfield  {title} {\bibinfo {title} {Theoretical framework
  for quantum networks},\ }\href {https://doi.org/10.1103/PhysRevA.80.022339}
  {\bibfield  {journal} {\bibinfo  {journal} {Phys. Rev. A}\ }\textbf {\bibinfo
  {volume} {80}},\ \bibinfo {pages} {022339} (\bibinfo {year}
  {2009})}\BibitemShut {NoStop}%
\bibitem [{\citenamefont {Ziman}(2008)}]{ziman2008process}%
  \BibitemOpen
  \bibfield  {author} {\bibinfo {author} {\bibfnamefont {M.}~\bibnamefont
  {Ziman}},\ }\bibfield  {title} {\bibinfo {title} {Process
  positive-operator-valued measure: A mathematical framework for the
  description of process tomography experiments},\ }\href
  {https://doi.org/10.1103/PhysRevA.77.062112} {\bibfield  {journal} {\bibinfo
  {journal} {Phys. Rev. A}\ }\textbf {\bibinfo {volume} {77}},\ \bibinfo
  {pages} {062112} (\bibinfo {year} {2008})}\BibitemShut {NoStop}%
\bibitem [{\citenamefont {Caro}(2022)}]{caro2023learning}%
  \BibitemOpen
  \bibfield  {author} {\bibinfo {author} {\bibfnamefont {M.~C.}\ \bibnamefont
  {Caro}},\ }\href@noop {} {\bibinfo {title} {{Learning quantum processes and
  Hamiltonians via the Pauli transfer matrix}}} (\bibinfo {year} {2022}),\
  \Eprint {https://arxiv.org/abs/2212.04471} {arXiv:2212.04471} \BibitemShut
  {NoStop}%
\bibitem [{\citenamefont {Huang}\ \emph {et~al.}(2023)\citenamefont {Huang},
  \citenamefont {Chen},\ and\ \citenamefont {Preskill}}]{huang2022learning}%
  \BibitemOpen
  \bibfield  {author} {\bibinfo {author} {\bibfnamefont {H.-Y.}\ \bibnamefont
  {Huang}}, \bibinfo {author} {\bibfnamefont {S.}~\bibnamefont {Chen}},\ and\
  \bibinfo {author} {\bibfnamefont {J.}~\bibnamefont {Preskill}},\ }\bibfield
  {title} {\bibinfo {title} {Learning to predict arbitrary quantum processes},\
  }\href {https://doi.org/10.1103/PRXQuantum.4.040337} {\bibfield  {journal}
  {\bibinfo  {journal} {PRX Quantum}\ }\textbf {\bibinfo {volume} {4}},\
  \bibinfo {pages} {040337} (\bibinfo {year} {2023})}\BibitemShut {NoStop}%
\bibitem [{\citenamefont {Perlt}(2022)}]{Perlt}%
  \BibitemOpen
  \bibfield  {author} {\bibinfo {author} {\bibfnamefont {E.}~\bibnamefont
  {Perlt}},\ }\href {https://doi.org/10.1007/978-3-030-67262-1} {\emph
  {\bibinfo {title} {Basis sets in computational chemistry}}}\ (\bibinfo
  {publisher} {Springer},\ \bibinfo {address} {Berlin},\ \bibinfo {year}
  {2022})\BibitemShut {NoStop}%
\bibitem [{\citenamefont {Zhang}\ \emph {et~al.}(2023)\citenamefont {Zhang},
  \citenamefont {Ren}, \citenamefont {Rinke}, \citenamefont {Blum},\ and\
  \citenamefont {Scheffler}}]{Scheffler}%
  \BibitemOpen
  \bibfield  {author} {\bibinfo {author} {\bibfnamefont {I.~Y.}\ \bibnamefont
  {Zhang}}, \bibinfo {author} {\bibfnamefont {X.}~\bibnamefont {Ren}}, \bibinfo
  {author} {\bibfnamefont {P.}~\bibnamefont {Rinke}}, \bibinfo {author}
  {\bibfnamefont {V.}~\bibnamefont {Blum}},\ and\ \bibinfo {author}
  {\bibfnamefont {M.}~\bibnamefont {Scheffler}},\ }\bibfield  {title} {\bibinfo
  {title} {{Numeric atom-centered-orbital basis sets with valence-correlation
  consistency from H to Ar}},\ }\href
  {https://doi.org/10.1088/1367-2630/15/12/123033} {\bibfield  {journal}
  {\bibinfo  {journal} {New J. Phys.}\ }\textbf {\bibinfo {volume} {15}},\
  \bibinfo {pages} {123033} (\bibinfo {year} {2023})}\BibitemShut {NoStop}%
\bibitem [{\citenamefont {Budich}\ \emph {et~al.}(2014)\citenamefont {Budich},
  \citenamefont {Eisert}, \citenamefont {Bergholtz}, \citenamefont {Diehl},\
  and\ \citenamefont {Zoller}}]{PhysRevB.90.115110}%
  \BibitemOpen
  \bibfield  {author} {\bibinfo {author} {\bibfnamefont {J.~C.}\ \bibnamefont
  {Budich}}, \bibinfo {author} {\bibfnamefont {J.}~\bibnamefont {Eisert}},
  \bibinfo {author} {\bibfnamefont {E.~J.}\ \bibnamefont {Bergholtz}}, \bibinfo
  {author} {\bibfnamefont {S.}~\bibnamefont {Diehl}},\ and\ \bibinfo {author}
  {\bibfnamefont {P.}~\bibnamefont {Zoller}},\ }\bibfield  {title} {\bibinfo
  {title} {{Search for localized Wannier functions of topological band
  structures via compressed sensing}},\ }\href
  {https://doi.org/10.1103/PhysRevB.90.115110} {\bibfield  {journal} {\bibinfo
  {journal} {Phys. Rev. B}\ }\textbf {\bibinfo {volume} {90}},\ \bibinfo
  {pages} {115110} (\bibinfo {year} {2014})}\BibitemShut {NoStop}%
\bibitem [{\citenamefont {Ozoliņš}\ \emph {et~al.}(2013)\citenamefont
  {Ozoliņš}, \citenamefont {Lai}, \citenamefont {Caflisch},\ and\
  \citenamefont {Osher}}]{Ozolins}%
  \BibitemOpen
  \bibfield  {author} {\bibinfo {author} {\bibfnamefont {V.}~\bibnamefont
  {Ozoliņš}}, \bibinfo {author} {\bibfnamefont {R.}~\bibnamefont {Lai}},
  \bibinfo {author} {\bibfnamefont {R.}~\bibnamefont {Caflisch}},\ and\
  \bibinfo {author} {\bibfnamefont {S.}~\bibnamefont {Osher}},\ }\bibfield
  {title} {\bibinfo {title} {Compressed modes for variational problems in
  mathematics and physics},\ }\href {https://doi.org/10.1073/pnas.131867911}
  {\bibfield  {journal} {\bibinfo  {journal} {Proc. Natl. Acad. Sci. USA}\ ,\
  \bibinfo {pages} {18368}} (\bibinfo {year} {2013})}\BibitemShut {NoStop}%
\bibitem [{\citenamefont {Akaike}(1974)}]{Akaike}%
  \BibitemOpen
  \bibfield  {author} {\bibinfo {author} {\bibfnamefont {H.}~\bibnamefont
  {Akaike}},\ }\bibfield  {title} {\bibinfo {title} {A new look at the
  statistical model identification},\ }\href
  {https://doi.org/10.1109/TAC.1974.1100705} {\bibfield  {journal} {\bibinfo
  {journal} {IEEE Trans. Aut. Contr.}\ }\textbf {\bibinfo {volume} {19}},\
  \bibinfo {pages} {716–723} (\bibinfo {year} {1974})}\BibitemShut {NoStop}%
\bibitem [{\citenamefont {Riofrio}\ \emph {et~al.}(2017)\citenamefont
  {Riofrio}, \citenamefont {Gross}, \citenamefont {Flammia}, \citenamefont
  {Monz}, \citenamefont {Nigg}, \citenamefont {Blatt},\ and\ \citenamefont
  {Eisert}}]{ExperimentalCompressed}%
  \BibitemOpen
  \bibfield  {author} {\bibinfo {author} {\bibfnamefont {C.~A.}\ \bibnamefont
  {Riofrio}}, \bibinfo {author} {\bibfnamefont {D.}~\bibnamefont {Gross}},
  \bibinfo {author} {\bibfnamefont {S.~T.}\ \bibnamefont {Flammia}}, \bibinfo
  {author} {\bibfnamefont {T.}~\bibnamefont {Monz}}, \bibinfo {author}
  {\bibfnamefont {D.}~\bibnamefont {Nigg}}, \bibinfo {author} {\bibfnamefont
  {R.}~\bibnamefont {Blatt}},\ and\ \bibinfo {author} {\bibfnamefont
  {J.}~\bibnamefont {Eisert}},\ }\bibfield  {title} {\bibinfo {title}
  {Experimental quantum compressed sensing for a seven-qubit system},\ }\href
  {https://doi.org/10.1038/ncomms15305} {\bibfield  {journal} {\bibinfo
  {journal} {Nature Comm.}\ }\textbf {\bibinfo {volume} {8}},\ \bibinfo {pages}
  {15305} (\bibinfo {year} {2017})}\BibitemShut {NoStop}%
\bibitem [{\citenamefont {Bhatia}()}]{bhatia1997matrix}%
  \BibitemOpen
  \bibfield  {author} {\bibinfo {author} {\bibfnamefont {R.}~\bibnamefont
  {Bhatia}},\ }\href {https://doi.org/10.1007/978-1-4612-0653-8} {\emph
  {\bibinfo {title} {Matrix analysis}}},\ \bibinfo {series} {Graduate Texts in
  Mathematics}, Vol.\ \bibinfo {volume} {169}\ (\bibinfo  {publisher} {Springer
  New York})\BibitemShut {NoStop}%
\bibitem [{{\relax DLMF}()}]{NIST:DLMF}%
  \BibitemOpen
  {\relax DLMF},\ \href {https://dlmf.nist.gov/} {\bibinfo {title} {{\it NIST
  Digital Library of Mathematical Functions}}},\ \bibinfo {howpublished}
  {\url{https://dlmf.nist.gov/}, Release 1.2.1 of 2024-06-15},\ \bibinfo {note}
  {f.~W.~J. Olver, A.~B. {Olde Daalhuis}, D.~W. Lozier, B.~I. Schneider, R.~F.
  Boisvert, C.~W. Clark, B.~R. Miller, B.~V. Saunders, H.~S. Cohl, and M.~A.
  McClain, eds.}\BibitemShut {Stop}%
\end{thebibliography}%

\clearpage

\onecolumngrid
\appendix

\section{Proof of \cref{lem:sec_ToP_approximately_sparse_error_term}}
\label{sec:appendix_proof_sparse_spillover}

This section is devoted to the proof of \cref{lem:sec_ToP_approximately_sparse_error_term} which bounds the spillover of contributions from inexact sparsity into our estimate. Before we come to the proof, we establish an auxiliary lemma based on the following result.
\begin{slemma}[Proposition 6.3 of Ref.~\cite{foucart2013mathematical}]\label{lem:operator_norm_estimate_disjoint_sets_rip}
Let $A \in \bbC^{M \times D}$ and denote with $A_S = A \Pi_S$ and $A_{S'} = A \Pi_{S'}$ its restriction to disjoint sets of indices $S$ and $S'$ of cardinality at most $s$. Then,
\begin{align}
    \lVert A^{\dagger}_S A_{S'} \rVert_{\infty} \leq \Delta_{2s}(A).
\end{align}
\end{slemma}
We are now ready to bound the spillover at the level of the measurement matrix in terms of its RIP constant.

\begin{slemma}[Bound to spillover] \label{lem:operator_norm_disj_sets_pseudo_inv}
    Let $A \in \bbC^{M \times D}$ with $\Delta_{s}(A/\sqrt{M}) < 1$ and $S, S'$ disjoint sets of cardinality at most $s$ such that. Then,
    \begin{align}
        \norm{A_S^+A_{S'}} \leq \frac{\Delta_{2s}(A/\sqrt{M})}{1-\Delta_{2s}(A/\sqrt{M})} \, .
    \end{align}
\end{slemma}
\begin{proof}
    Due to the RIP constant being smaller than one, $A_S$ is injective. For injective $A_S$, the pseudo-inverse is given as $A_S^+=(A_S^{\dagger} A_S)^{-1} A_S^{\dagger}$. Then,
    \begin{align}
    \nonumber
        \lVert A_S^{+} A_{S'} \rVert_{\infty} &= \lVert (A_S^{\dagger} A_S)^{-1} A_S^{\dagger} A_{S'} \rVert_{\infty} \\
        \nonumber
        &\leq \lVert (A_S^{\dagger} A_S)^{-1} \rVert_{\infty} \lVert A_S^{\dagger} A_{S'} \rVert_{\infty} \\
        &= \left\lVert \left(\frac{A_S^{\dagger}}{\sqrt{M}} \frac{A_S}{\sqrt{M}}\right)^{-1} \right\rVert_{\infty}\left\lVert \frac{A_S^{\dagger}}{\sqrt{M}} \frac{A_{S'}}{\sqrt{M}} \right\rVert_{\infty} \\
        \nonumber
        &\leq \frac{\Delta_{2s}}{1 - \Delta_s} \\
        \nonumber
        &\leq \frac{\Delta_{2s}}{1 - \Delta_{2s}}.\label{eqn:operator_norm_crossterm}
    \end{align}
    The first inequality is the generic operator norm bound for products of matrices, the second equality adds the normalization to establish the restricted isometry property of $A$. The second inequality combines the restricted isometry property of $A$ with \cref{lem:operator_norm_estimate_disjoint_sets_rip} and the final inequality follows from $\Delta_s \leq \Delta_{2s}$.
\end{proof}
We are now ready to present the proof.

\begin{customthm}{\cref{lem:sec_ToP_approximately_sparse_error_term}}
    If $\tvert{\aalpha}_{\calO, 1} \leq \gamma_{\ell^1}$ and $\Delta_{2s}(A/\sqrt{M}) \leq \Delta/2 \leq 1/2$, then
    \begin{align}
        \tvert{A_S^+A_{\bar{S}} \aalpha_{\bar{S}}}_{\calO,2} \leq \Delta \gamma_{\ell^1} \, .
    \end{align}
\end{customthm}
\begin{proof}
    Before we commence, we reiterate that, for a given set $S$ and a vector of operators $\XX$, $\XX_S = \Pi_S \XX$ takes the value $X_i$ for all entries with indices $i \in S$ and else $0$. For an observable $O$, $\XX_O \coloneqq (\Tr[OX_1], \Tr[OX_2], \dots)$. Then, $(\XX_S)_O$ takes the value $\Tr[OX_i]$ for $i \in S$ and else $0$. 
    We can partition $\bar{S}$ into disjoint sets $S_i$ such that $\bar{S}= \bigcup_i S_i$ and $|S_i|\leq |S|$. 
    \begin{align}
        \tvert{A_S^+A_{\bar{S}} \aalpha_{\bar{S}}}_{\calO,2} &= \sup_{O\in \calO} \norm{A_S^+A_{\bar{S}} (\aalpha_{\bar{S}})_O}_2 \\
        \nonumber
        &= \sup_{O\in \calO} \norm{\sum_i A_S^+ A_{S_i} (\aalpha_{S_i})_O}_2 \\
        \nonumber
        &\leq \sup_{O \in \calO} \sum_i \norm{A_S^+ A_{S_i} (\aalpha_{S_i})_O}_2 \\
        \nonumber
        &\leq \sup_{O \in \calO} \sum_i \norm{A_S^+A_{S_i}}_{\infty} \norm{(\aalpha_{S_i})_O}_2 \\
        \nonumber
        &\leq \frac{\Delta_{2s}}{1 - \Delta_{2s}} \sup_{O \in \calO} \sum_i \norm{(\aalpha_{S_i})_O}_2 \\
        \nonumber
        &\leq \frac{\Delta_{2s}}{1 - \Delta_{2s}} \sup_{O \in \calO} \norm{(\aalpha_{S_i})_O}_1 \\
        \nonumber
        &=  \frac{\Delta_{2s}}{1 - \Delta_{2s}} \tvert{\aalpha_{\bar{S}}}_{\calO,1} \\
        &\leq \Delta \gamma_{\ell^1} \, .
        \nonumber
    \end{align}
    The first inequality is just the triangle inequality, the second uses the definition of the operator norm as the $2\to 2$ norm. The third inequality applies \cref{lem:operator_norm_disj_sets_pseudo_inv} and the fourth uses that $\norm{\xx}_2 \leq \norm{\xx}_1$ holds for all $\xx$. For the last inequality, note that for $x \leq 1/2$,
    \begin{align}
        \frac{x}{1 - x} \leq x \, .
    \end{align}
\end{proof}

\section{Proof of \cref{thrm:sec_ISS_hoeffding}}
\label{sec:appendix_proof_estimator}
\begin{stheorem}[Hoeffding's inequality]
    Let $X_1, X_2, \dots, X_M$ be a number of $M$ random variables such that $|X_i| \leq B$ for all $i$ and denote their empirical mean with $\hat{\mu} \coloneqq \sum_i X_i /M$. Then, one achieves
    \begin{align}
        \left| \frac{1}{M} \sum_{i=1}^M X_i - \bbE \left[\hat{\mu} \right] \right|\leq  \epsilon
    \end{align}
    with probability at least $1-\delta$ by using 
    \begin{align} \label{eqn:SIT_general_scaling_in_M}
        M \geq  \log\left( \frac{2}{\delta}\right) \frac{ B^2}{2\epsilon^2}
    \end{align}
    many samples.
\end{stheorem}

\begin{customthm}{\cref{thrm:sec_ISS_hoeffding}}
    Using HTP to probe linear systems of the form \cref{eqn:sec_ISS_realistic linear system} to obtain data vectors $\cc^{\#}(P)$, one achieves
    \begin{align}
        | \hat{X}_k - X_k |\leq  \epsilon 
    \end{align}
    with probability at least $1-\delta$ by using 
    \begin{align} \label{eqn:SIT_general_scaling_in_M}
        L \geq \log\left( \frac{2D}{\delta}\right) \frac{\left(1 + \kappa\right)^2}{2 \epsilon^4} 
    \end{align}
    data vectors.
\end{customthm}
\begin{proof}
    From \cref{lem:sec_ISS_coefficient_bound}, each term in the empirical mean
    \begin{align}
        \hat{X}_k \coloneqq \sqrt{\frac{1}{L} \sum_{l=1}^L |c_k^{\#}(P_l)|^2} 
    \end{align}
    is bounded as $|c_k^{\#}(P)|^2 \leq (1 + \kappa)^2$. Thus, we can apply \cref{thrm:sec_ISS_hoeffding} (Hoeffding's inequality) to the term $\frac{1}{L} \sum_{l=1}^L |c_k^{\#}(P_l)|^2$. As 
     \begin{align}
    |\sqrt{x}-\sqrt{x\pm \epsilon}| \leq \sqrt{\epsilon}, 
    \end{align}
    we need a precision of $\epsilon^2$ to ensure an error of at most $\epsilon$ for $|\hat{X}_k - X_k|$. Estimating all $D$ empirical means $\hat{X}_k$ to failure probability at most $\delta/D$ and taking the union bound gives the desired result.
\end{proof}

\section{Proof of \cref{thrm:sec_ISS_separability_criterion_lowerB_assumptions}}\label{sec:appendix_proof_lowerB_assumptions}

\begin{customthm}{\cref{thrm:sec_ISS_separability_criterion_lowerB_assumptions}}
    Choosing $\epsilon$ such that 
\begin{align}
    \min_{k \in S} \left\{\norm{\alpha_k} - \norm{\gamma_k}_2 \right\} &-  \left(1 + \frac{1}{\beta_{\cc}} \right)\max_{k' \in \bar{S}} \norm{\alpha_{k'}}_2  \nonumber \\
    &- \left(1 + \frac{1}{\beta_{\nn}} \right)\max_{k' \in \bar{S}} \norm{\gamma_{k'}}_2 \geq 2 \epsilon
\end{align}
    guarantees correct identification of the sparse support.
\end{customthm}
\begin{proof}
    The identification of the sparse support succeeds if 
    \begin{align}
        \min_{k \in S} X_k - \max_{k' \in \bar{S}} X_{k'} \geq 2 \epsilon \, .
    \end{align}
    We, therefore,  derive an operationally meaningful lower bound to $\min_{k \in S} X_k$ and an upper bound to $\max_{k' \in \bar{S}} X_{k'}$. Starting with the latter, we begin by noting that there is some family of indicator functions $\chi_{k',l}$ such that
    \begin{align}
        \sqrt{\bbE\left[|c_{k'}^{\#}(P)|^2\right]} = \left( \frac{1}{d} \sum_{l=1}^d |c_{k'}(B_l) + n_{k'}(B_l)|^2 \chi_{k',l} \right)^{\frac{1}{2}} \, .
    \end{align}
    From there, we have 
    \begin{align}
        \sqrt{\bbE\left[|c_{k'}^{\#}(P)|^2\right]} &= \left( \frac{1}{d} \sum_{l=1}^d |c_{k'}(B_l) + n_{k'}(B_l)|^2 \chi_{k',l} \right)^{\frac{1}{2}} \\
        \nonumber
        &\leq \left( \frac{1}{d} \sum_{l=1}^d |c_{k'}(B_l) + n_{k'}(B_l)|^2 \right)^{\frac{1}{2}} \\
        \nonumber
        &\leq \left( \frac{1}{d} \sum_{l=1}^d |c_{k'}(B_l)|^2 \right)^{\frac{1}{2}} + \left( \frac{1}{d} \sum_{l=1}^d |n_{k'}(B_l)|^2 \right)^{\frac{1}{2}} \\
        \nonumber
        &= \norm{\alpha_{k'}}_2 + \norm{\gamma_{k'}}_2 \, .
        \nonumber
    \end{align}
    In the second step, we have used the triangle inequality. For the lower bound, we similarly start by expressing the expectation value with an indicator function.
    \begin{align}
        \sqrt{\bbE\left[|c_k^{\#}(P)|^2\right]} = \left( \frac{1}{d} \sum_{l=1}^d |c_k(B_l) + n_k(B_l)|^2 \chi_{k,l} \right)^{\frac{1}{2}} \, ,
    \end{align}
    where 
    \begin{align}
    \chi_{k,l} \coloneqq \chi\left(|c_k(B_l)+n_k(B_l)|>\max_{k' \in \bar{S}} |c_{k'}(B_l) + n_{k'}(B_l)| \right) \, .
    \end{align}
    From there, we have
    \begin{align}
        \sqrt{\bbE\left[|c_k^{\#}(P)|^2\right]} &= \left( \frac{1}{d} \sum_{l=1}^d |c_k(B_l) + n_k(B_l)|^2 \chi_{k,l} \right)^{\frac{1}{2}} \\
        \nonumber
        &\geq \left( \frac{1}{d} \sum_{l=1}^d |c_k(B_l) + n_k(B_l)|^2 - \max_{k' \in \bar{S}}|c_{k'}(B_l) + n_{k'}(B_l)|^2 \right)^{\frac{1}{2}} \\
        \nonumber
        &\geq \left( \frac{1}{d} \sum_{l=1}^d |c_k(B_l) + n_k(B_l)|^2 - |\max_{k' \in \bar{S}} \norm{\cc({k'})}_{\infty} + \max_{k' \in \bar{S}} \norm{\nn(k')}_{\infty}|^2 \right)^{\frac{1}{2}} \\
        \nonumber
        &\geq \left( \frac{1}{d} \sum_{l=1}^d |c_k(B_l) + n_k(B_l)|^2  \right)^{\frac{1}{2}} - \left( \frac{1}{d} \sum_{l=1}^d |\max_{k' \in \bar{S}} \norm{\cc(k')}_{\infty} + \max_{k' \in \bar{S}} \norm{\nn(k')}_{\infty}|^2 \right)^{\frac{1}{2}} \\
        \nonumber
        &\geq \norm{\alpha_k}_2 - \norm{\gamma_k}_2 - \max_{k' \in \bar{S}} \norm{\cc(k')}_{\infty} - \max_{k' \in \bar{S}} \norm{\nn(k')}_{\infty} \\
        \nonumber
        &\geq \norm{\alpha_k}_2 - \norm{\gamma_k}_2 - \max_{k' \in \bar{S}} \frac{1}{\beta_{\cc} \sqrt{d}} \norm{\cc(k')}_{2} - \max_{k' \in \bar{S}} \frac{1}{\beta_{\nn} \sqrt{d}} \norm{\nn(k')}_{2} \\
        \nonumber
        &= \norm{\alpha_k}_2 - \norm{\gamma_k}_2 - \max_{k' \in \bar{S}} \frac{1}{\beta_{\cc}} \norm{\alpha_{k'}}_{2} - \max_{k' \in \bar{S}} \frac{1}{\beta_{\nn}} \norm{\gamma_{k'}}_{2} \, .
        \nonumber
    \end{align}
    For the first step, note that in case the indicator function $\chi_{k,l}$ is zero, by subtracting we get something lower than zero. For the fifth inequality, we use the definition of the flatness constant. By combining the lower and upper bound derived here and taking minima and maxima where appropriate, the desired claim follows.
\end{proof}

\section{Sparsity of time-evolution of states with sub-Gaussian energies}\label{sec:appendix_sub-Gaussian_energy_spectrum}
In this section, we want to establish rigorous constraints on the operator-valued coefficients $\alpha_k$ of the time evolution of a state $\rho_0$ under an NMR Hamiltonian $H = \sum_{e = 0}^{E_{\max}} e \Pi_e$,
\begin{align}
    \rho(t) = \sum_{k = - E_{\max}}^{E_{\max}} \alpha_k e^{i k t},
\end{align}
under the assumption of a sub-Gaussian energy spectrum of $\rho_0$, \emph{i.e.},  that
\begin{align}\label{eqn:sub-Gaussian_energies_appendix}
    \Tr[\rho_0 \Pi_e] \leq \tau \exp\left(-\frac12 \frac{(e-e_0)^2}{\sigma^2} \right)
\end{align}
for some constants $\tau$, $\sigma$ and $e_0$.

To do so, we first establish an auxiliary lemma.
\begin{lemma}\label{lem:projector_cross_term_bound}
    Let $\Pi$ and $\Pi'$ denote orthogonal projectors. Then, for any quantum state $\rho$, we have that
    \begin{align}
        \lVert \Pi \rho \Pi' \rVert_1 \leq \sqrt{ r \lVert \Pi \rho \Pi \rVert_{1} \lVert \Pi' \rho \Pi' \rVert_1},
    \end{align}
    where $r = \min\{ \operatorname{rank}(\Pi), \operatorname{rank}(\Pi') \}$.
\end{lemma}
\begin{proof}
    The projector $\Pi + \Pi'$ defines a principal submatrix of $\rho$ given by
    \begin{align}
        \begin{pmatrix}
            \Pi \rho \Pi & \Pi \rho \Pi' \\
            \Pi' \rho \Pi & \Pi' \rho \Pi'
        \end{pmatrix} \eqqcolon \begin{pmatrix}
            A & X \\ X^{\dagger} & B
        \end{pmatrix},
    \end{align}
    where the entries of the matrix are understood to be meant on their support.
    
    Principal submatrices of positive semi-definite operators are again positive semidefinite, which follows, for example, from the Cauchy interlacing theorem~\cite{bhatia1997matrix}. The positive semi-definiteness of a block-matrix is equivalent to the positive semi-definiteness of the Schur complement, hence
    \begin{align}
        A \geq X B^{-1} X^{\dagger}.
    \end{align}
    This implies that
    \begin{align}
        \Tr[A] &\geq \Tr[ X B^{-1} X^{\dagger} ] \\
        \nonumber
        &= \Tr[ X^{\dagger} X B^{-1}] \\
        &\geq \Tr[X^{\dagger} X] \lambda_{\min}(B^{-1}).
        \nonumber
    \end{align}
    As $A \geq 0$ the trace is identical to the trace norm and $\Tr[X^{\dagger} X] = \lVert X \rVert_2^2$ with the Schatten-2 or Hilbert-Schmidt norm, this rearranges to
    \begin{align}
        \lVert X \rVert_2 &\leq \sqrt{\lVert A \rVert_1 \lVert B \rVert_{\infty}} \\
        \nonumber
        &\leq \sqrt{\lVert A \rVert_1 \lVert B \rVert_{1}}.
    \end{align}
    Now, we use that
    \begin{align}
        \lVert X \rVert_2 \geq \frac{1}{\sqrt{\operatorname{rank}(X)}} 
        \lVert X \rVert_1 
    \end{align}
    and rearrange to conclude the statement from the fact that the rank of $X$ it at most the minimal rank between the two projectors $\Pi$ and $\Pi'$.
\end{proof}

Equipped with this lemma, we can now prove a bound on the $\gamma_{\ell^1}$ sparsity defect with respect to the $\ell$-local trace norm induced by the set 
\begin{align}
    \calO_{\ell} = \{ O : \lVert O \rVert_{\infty} \leq 1, O \text{ is $\ell$-local} \},
\end{align}
which we will in the proof control with the regular trace norm.
\begin{lemma}\label{lem:sub-Gaussian_energy_sparsity_guarantee}
    Let $\rho_0$ be a state sub-Gaussian energy distribution relative to an NMR Hamiltonian as in \cref{eqn:sub-Gaussian_energies_appendix}. 
    Then, we can guarantee that the parametrized state $\rho(\cdot)$ has sparsity defect $\gamma_{\ell^1} \leq \gamma$ with respect to the support $S = \{-R, -R + 1, \dots, R-1, R\}$ and the local trace norm if
    \begin{align}
        R = \tilde{O}\left( \sigma \sqrt{n + \log \frac{\tau}{\gamma} } \right).
    \end{align}
\end{lemma}
\begin{proof}  
Let us fix a sparse support $S = \{ -R, -R+1, \dots, R-1, R\}$ with $R\geq2$. In this case, we have that
\begin{align}
    \gamma_{\ell^1} &\leq \sum_{|k| > R} \lVert \alpha_k \rVert_{\calO_\ell} \\
    \nonumber
    &\leq \sum_{|k| > R} \lVert \alpha_k \rVert_{1} \\
    \nonumber
    &\leq \sum_{|k| > R} \left\lVert \sum_{e=-E_{\max}}^{E_{\max}-k} \Pi_e \rho_0 \Pi_{e+k} \right\rVert_{1} \\
    \nonumber
    &\leq \sum_{|k| > R} \sum_{e=-E_{\max}}^{E_{\max}-k} \left\lVert \Pi_e \rho_0 \Pi_{e+k} \right\rVert_{1} \\
    \nonumber
    &\leq \sum_{|k| > R} \sum_{e=-E_{\max}}^{E_{\max}-k} \sqrt{2^n \lVert \Pi_e \rho_0 \Pi_{e} \rVert_{1} \lVert \Pi_{e+k} \rho_0 \Pi_{e+k} \rVert_{1}} \\
    \nonumber
    &\leq \sum_{|k| > R} \sum_{e=-E_{\max}}^{E_{\max}-k} \sqrt{2^n \tau^2 \exp\left(-\frac{1}{2}\frac{(e-e_0)^2 + (e-(e_0-k))^2}{\sigma^2} \right)} \\
    \nonumber
    &= 2^{n/2} \tau \sum_{|k| > R} \sum_{e=-E_{\max}}^{E_{\max}-k} \exp\left(-\frac{1}{2}\frac{(e-e_0)^2 + (e-(e_0-k))^2}{2\sigma^2} \right) \\
    \nonumber
    &\leq 2^{n/2} \tau \sum_{|k| > R} \sum_{e=-\infty}^{\infty} \exp\left(-\frac{1}{2}\frac{(e-e_0)^2 + (e-(e_0-k))^2}{2\sigma^2} \right).
    \nonumber
\end{align}
The first inequality uses the triangle inequality to bound $\gamma_{\ell^1}$. The second inequality bounds the $\ell$-local trace norm with the regular trace norm. The third inequality uses the definition of the coefficients $\alpha_k$ given in \cref{eqn:energy_differences_nmr}. The fourth inequality bounds the norm of the $\alpha_k$ using the triangle inequality. The fifth inequality applies \cref{lem:projector_cross_term_bound} with the trivial rank bound $r \leq 2^n$. The sixth inequality then applies the sub-Gaussian energy assumption. The first equality takes the square root and the final inequality bounds the sum over $e$ by letting $E_{\max}\to \infty$.

We will now bound the sum over the Gaussian terms using the integral upper bound for monotonic functions. To this end, we split the sum and can simplify by exploiting the evident symmetries:
\begin{align}
    \gamma_{\ell^1} &\leq 2^{n/2} \tau 2 \sum_{k > R} \left( \sum_{e=-\infty}^{e_0 - k} + \sum_{e = e_0 - k}^{e_0 - k/2} + \sum_{e = e_0 - k/2}^{e_0} + \sum_{e = e_0}^{\infty} \right) \exp\left(-\frac{1}{2}\frac{(e-e_0)^2 + (e-(e_0-k))^2}{2\sigma^2} \right) \\
    &\leq 2^{n/2} \tau 4 \sum_{k > R} \left(\sum_{e = e_0 - \lfloor k/2 \rfloor}^{e_0} + \sum_{e = e_0}^{\infty} \right) \exp\left(-\frac{1}{2}\frac{(e-e_0)^2 + (e-(e_0-k))^2}{2\sigma^2} \right). 
    \nonumber
\end{align}
In the first line, it does not matter for uneven $k$ to which of the two sums we attribute $\lfloor k/2 \rfloor$.
For the first sum, the function is monotonically increasing, for the second it is monotonically decreasing. Hence,
\begin{align}
    \gamma_{\ell^1} &\leq  2^{n/2} \tau 4 \sum_{k > R} \left(\int_{e_0 - \lfloor k/2 \rfloor}^{e_0 + 1} \diff e  +   \int_{e = e_0-1}^{\infty} \diff e \right) \exp\left(-\frac{1}{2}\frac{(e-e_0)^2 + (e-(e_0-k))^2}{2\sigma^2} \right).
\end{align}
The indefinite integral on the right evaluates to
\begin{align}
    \int \diff e\, \exp\left(-\frac{1}{2}\frac{(e-e_0)^2 + (e-(e_0-k))^2}{2\sigma^2} \right) &=
    \sqrt{\frac{\pi}{2}} \sigma \exp\left(-\frac{k^2}{8 \sigma^2} \right) \operatorname{Erf}\left( \frac{2(e-e_0) + k}{\sqrt{8} \sigma} \right),
\end{align}
and hence
\begin{align}
    \gamma_{\ell^1} &\leq  2^{n/2} \tau 4 \sum_{k > R}  \sqrt{\frac{\pi}{2}} \sigma \exp\left(-\frac{k^2}{8 \sigma^2} \right)\left( 1 + \operatorname{Erf}\left( \frac{k+2}{\sqrt{8} \sigma} \right) - \operatorname{Erf}\left( \frac{k-2}{\sqrt{8} \sigma} \right) - \operatorname{Erf}\left( \frac{k - 2\lfloor k/2 \rfloor}{\sqrt{8} \sigma} \right)  \right).
\end{align}
As $|\operatorname{Erf}(x)| \leq 1$, we can bound the term with the error functions by $4$ and obtain
\begin{align}
    \gamma_{\ell^1} &\leq 16 \, 2^{n/2} \tau  \sqrt{\frac{\pi}{2}} \sigma \sum_{k > R} \exp\left(-\frac{k^2}{8 \sigma^2} \right) 
    \\
    \nonumber 
    &\leq  16 \, 2^{n/2} \tau  \sqrt{\frac{\pi}{2}} \sigma \int_{R-1}^{\infty} \diff k \, \exp\left(-\frac{k^2}{8 \sigma^2} \right) 
    \\
     \nonumber 
    \nonumber
    &= 16 \, 2^{n/2} \tau  \pi \sigma^2 \left( 1- \operatorname{Erf}\left(\frac{R-1}{\sqrt{8}\sigma} \right) \right),
    \nonumber
\end{align}
where we have again used the integral upper bound for monotonically decreasing sums. Applying the standard inequality $1-\operatorname{Erf}(x) \leq e^{-x^2}$, we arrive at our end result
\begin{align}
    \gamma_{\ell^1} &\leq 16 \, 2^{n/2} \tau  \pi \sigma^2 \exp\left( - \frac{(R-1)^2}{8 \sigma^2} \right).
\end{align}
Now, if we wish to achieve a certain value $\gamma_{\ell^1} \leq \gamma$, we have to have
\begin{align}
    R \geq 1 + \sqrt{8 \sigma^2 \log \frac{2^{n/2 + 4} \tau \pi \sigma^2}{\gamma}}
\end{align}
which means
\begin{align}
    R = \tilde{O}\left( \sigma \sqrt{n + \log \frac{\tau}{\gamma} }\right) 
\end{align}
is sufficient to achieve a certain target sparsity defect as claimed.
\end{proof}

\section{Chebychev approximation of free fermionic time evolution}\label{sec:appendix_chebychev}

We consider the problem of approximating the time evolution of free fermions $\rho(t)$ for $t \in [-1,1]$. In this case, the Chebychev polynomials form an appropriate set of basis functions. They are the unique polynomials satisfying the relation
\begin{align}\label{eqn:chebychev_defining_property}
    T_k(\cos t) = \cos kt
\end{align}
for integer $k \geq 0$. They fulfill the orthogonality relation
\begin{align}
    \int_{-1}^1 \diff t \, \frac{T_k(t) T_{l}(t)}{\sqrt{1-t^2}} &=  \delta_{kl} \times \begin{cases}
        \pi & \text{ if } k=0 ,\\
        \frac{\pi}{2} & \text{ if } k \geq 1.
    \end{cases}
\end{align}
As 
\begin{align}
    \int_{-1}^1 \diff t \, \frac{1}{\sqrt{1-t^2}} = \pi,
\end{align}
we obtain an orthonormal function basis in the sense of our paper by choosing a measure $\tilde{\mu}$ and normalized Chebychev function $\tilde{T}_k$ as
\begin{align}
    \tilde{\mu}(t) &= \frac{1}{\pi}\frac{1}{\sqrt{1-t^2}} ,\\
    \tilde{T}_k(t) &= T_k(t) \times  \begin{cases}
        1 & \text{ if } k=0, \\
        \sqrt{2} & \text{ if } k \geq 1
    \end{cases} \\
    &\eqqcolon T_k(t) \xi_k,
\end{align}
where we denoted the different normalization factor as $\xi_k$.
To translate the time evolution into Chebychev polynomials, we need to find the Chebychev expansion of $\exp (- i \omega t)$. We have the following.
\begin{lemma}[Chebychev expansion]\label{lemma:exponential_to_chebychev}
    Let $\omega \in \bbR$ and $J_k$ the $k$-th Bessel function of the first kind. Then,
    \begin{align}
        \exp( -i \omega t) &= \sum_{k=0}^{\infty} i^k \xi_k J_k(\omega) \tilde{T}_k(t).
    \end{align}
\end{lemma}
\begin{proof}
    The prefactor $\tilde{c}_k$ of $\tilde{T}_k$ in the expansion of $\exp (-i \omega t)$ is given by
    \begin{align}
        \tilde{c}_k &= \int_{-1}^1 \diff \tilde\mu(t) \, \tilde{T}_k(t) \exp (-i \omega t) \\
        &= \frac{\xi_k}{\pi} \int_{-1}^1 \diff t \, \frac{T_k(t) \exp (-i \omega t)}{\sqrt{1-t^2}}.
        \nonumber
    \end{align}
    We aim to make use of the defining property of the Chebychev polynomials in \cref{eqn:chebychev_defining_property}. To do so, we substitute 
    \begin{align}
        t = \cos \tau \qquad \frac{\diff t}{\diff \tau} = - \sin \tau \qquad \int_{-1}^1 \diff t = -\int_{-\pi}^0 \diff \tau \, \sin \tau 
    \end{align}
    to obtain
    \begin{align}
        \tilde{c}_k &= -\frac{\xi_k}{\pi} \int_{-\pi}^0 \diff \tau \, \sin \tau \frac{T_k(\cos \tau) \exp (-i \omega \cos \tau)}{\sqrt{1-\cos^2 \tau}} \\
         \nonumber
        &= - \frac{\xi_k}{\pi} \int_{-\pi}^0 \diff \tau \, \frac{\sin \tau}{|\sin \tau|} {\cos (k \tau) \exp (-i \omega \cos \tau)} \\
        &=  \frac{\xi_k}{\pi} \int_{0}^\pi \diff \tau \, {\cos (k \tau) \exp (-i \omega \cos \tau)},
         \nonumber
    \end{align}
    where we used that in the given range of integration $\sin \tau / |\sin \tau| = -1$ and the symmetry of the integrand.
    Next, we recall an integral formula for the Bessel function of the first kind~\cite{NIST:DLMF}
    \begin{align}
        J_k(\omega) &= \frac{(-i)^k}{\pi} \int_0^{\pi} \diff \tau \, \cos (k \tau)\exp (i \omega \cos \tau). 
    \end{align}
    Hence, we have that
    \begin{align}
        \tilde{c}_k &= i^k \xi_k J_k(-\omega).
    \end{align}
    The claimed expansion follows from the (anti)symmetry of the Bessel functions of the first kind
    \begin{align}
        J_k(-\omega) = (-1)^k J_k(\omega).
    \end{align}
\end{proof}

Additionally, we need a bound on the Bessel function of the first kind.
\begin{lemma}\label{lemma:bessel_bound}
    Let $J_k$ be the $k$-th Bessel function of the first kind. Then,
    \begin{align}
        |J_k(\omega)| \leq \left( \frac{e \omega}{2 k} \right)^k.
    \end{align}
\end{lemma}
\begin{proof}
    We simply combine the bound~\cite{NIST:DLMF}
    \begin{align}
    |J_k(\omega)| \leq \frac{1}{k!} \left|\frac{\omega}{2}\right|^k
    \end{align}
    with the Stirling lower bound $k! \geq (k/e)^k$.
\end{proof}

\end{document}